\documentclass[12pt,a4paper,makeidx,reqno]{amsart}

\makeindex \topmargin=0in \headheight=8pt \headsep=0.5in \textheight=8.9in
\newcommand\tempskip[1]{#1}

\usepackage{amsmath,amssymb,amsthm}
\usepackage{mathrsfs}
\usepackage{eucal}
\usepackage[matrix,arrow,curve]{xy}
\usepackage{textcomp}
\usepackage[rightcaption]{sidecap}
\usepackage{wrapfig}
\usepackage{caption, subcaption}

\usepackage{graphicx} 
\graphicspath{ {images/} }

\usepackage{bm}

\usepackage{hyperref}
\usepackage[normalem]{ulem}
\usepackage{xcolor}

\newcommand{\bl}[1]{#1}

\newtheorem{theorem}{Theorem}[section]

\newtheorem{lemma}[theorem]{Lemma}
\newtheorem{cor}[theorem]{Corollary}
\newtheorem{prop}[theorem]{Proposition}
\newtheorem{defin}[theorem]{Definition}

\theoremstyle{remark}
\newtheorem{rem}[theorem]{Remark}

\numberwithin{equation}{section}

\newcounter{np} 

\def\lan{\left \langle }
\def\ran{\right \rangle}
\def\bar{\begin{array}}
\def\ear{\end{array}}
\def\beq{\begin{equation}}
\def\eeq{\end{equation}}

\def\dim{\mathrm{dim\,}}

\def\sl{\mathrm{SL}_2(\mathbb C)}
\def\su{\mathrm{SU}_2(\mathbb C)}
\def\gl{{\mathrm{GL}}_2(\mathbb C)}

\def\liesl{\mathfrak{sl}_2(\mathbb C)}
\def\psl{\mathrm{PSL}_2(\mathbb R)}
\def\slr{\mathrm{SL}_2(\mathbb R)}
\def\Hom{\mathrm{Hom}}
\def\Fun{\mathrm{Fun}}
\def\Tr{\mathrm{Tr}\,}
\def\End{\mathrm{End}}
\def\Sym{\mathrm{Sym}\,}

\def\F{\Fun_{\mathrm{alg}}}
\def\Fhol{\Fun_{\mathrm{hol}}}

\def\Ed{\CMcal{E}}
\def\Fa{\CMcal{F}}

\def\G{\CMcal{T}_\Omega} 
\def\L{\CMcal{L}(\G)}

\def\s{\mathrm{sign}}
\def\h{\mathrm{hol}}

\def\X{X}
\def\Xs{X_{\mathrm{twist}}}

\def\C{\CMcal{C}}

\def\el{\sigma_{\mathrm{left}}}
\def\er{\sigma_{\mathrm{right}}}
\def\pair{\Lambda^\Ed}

\def\Id{\mathrm{Id}}
\def\sym{\mathrm{sym}}
\def\sk{\mathrm{Sk}}

\renewcommand\Xi{\Theta}
\newcommand\lloop[1]{#1^\circ}

\begin{document}

\title[Tau-functions \`a la Dub\'edat for double-dimers and CLE(4)]
{Tau-functions \`a la Dub\'edat and probabilities of cylindrical events for double-dimers and CLE(4)}

\author[Mikhail Basok]{Mikhail Basok$^\mathrm{a,c}$}

\author[Dmitry Chelkak]{Dmitry Chelkak$^\mathrm{b,c}$}

\thanks{\textsc{${}^\mathrm{A}$ Laboratory of Modern Algebra and Applications, Saint-Petersburg State University, 14th Line, 29b, 199178 Saint-Petersburg, Russia.}}

\thanks{\textsc{${}^\mathrm{B}$ D\'epartement de math\'ematiques et applications de l'ENS, \'Ecole Normale Sup\'erieure PSL Research University, CNRS UMR 8553, Paris 5\`eme.}}

\thanks{\textsc{${}^\mathrm{C}$ St.Petersburg Department of Steklov Mathematical Institute (PDMI RAS).
Fontanka~27, 191023 St.Petersburg, Russia.}}

\thanks{\bl{The research of M.B.~was supported by the RScF grant 16-11-10039 ``Combinatorial, discrete and enumerative geometry'' and by the grant of the Government of the Russian Federation for the state support of scientific research carried out under the supervision of leading scientists, agreement 14.W03.31.0030 dated 15.02.2018. D.C. is a holder of the ENS--MHI chair funded by the MHI. On the final stage, the research of D.C. was partly supported by the ANR--18--CE40--0033 project DIMERS}}

\thanks{{\it E-mail addresses:} \texttt{m.k.basok@gmail.com}, \texttt{dmitry.chelkak@ens.fr}}

\begin{abstract} Building upon recent results of Dub\'edat~\cite{Dubedat} on the convergence of topological correlators in the double-dimer model considered on Temperleyan approximations~$\Omega^\delta$ to a simply connected domain~$\Omega\subset\mathbb C$ we prove the convergence of probabilities of cylindrical events for the \emph{double-dimer loop ensembles} on~$\Omega^\delta$ as~$\delta\to 0$. More precisely, let~$\lambda_1,\dots,\lambda_n\in\Omega$ and~$L$ be a macroscopic lamination on $\Omega\setminus\{\lambda_1,\dots,\lambda_n\}$, i.e., a collection of disjoint simple loops surrounding at least two punctures considered up to homotopies. We show that the probabilities~$P_L^\delta$ that one obtains~$L$ after withdrawing all loops surrounding no more than one puncture from a double-dimer loop ensemble on~$\Omega^\delta$ converge to a conformally invariant limit~$P_L$ as~$\delta \to 0$, for each~$L$.

Though our primary motivation comes from 2D statistical mechanics and probability, the proofs are of a purely analytic nature. The key techniques are the analysis of entire functions on the representation variety~$\Hom(\pi_1(\Omega\setminus\{\lambda_1,\dots,\lambda_n\})\to\sl)$ and on its (non-smooth) subvariety of locally unipotent representations.  In particular, we do \emph{not} use any RSW-type arguments for double-dimers.

The limits~$P_L$ of the probabilities~$P_L^\delta$ are defined as coefficients of the isomonodormic tau-function studied in~\cite{Dubedat} with respect to the Fock--Goncharov lamination basis on the representation variety. The fact that~$P_L$ coincides with the probability to obtain~$L$ from a sample of the nested CLE(4) in~$\Omega$ requires a small additional input, namely a mild crossing estimate for this nested conformal loop ensemble.
\end{abstract}



\maketitle

\newpage

\section{Introduction and main results}
Convergence of double-dimer interfaces and loop ensembles to SLE(4) and CLE(4), respectively, is a well-known prediction made by Kenyon after the introduction of SLE curves by Schramm, see~\cite[Section~2.3]{schramm-icm-06}. In particular, this provided a strong motivation to study couplings between Conformal Loop Ensembles (CLE) and the two-dimensional Gaussian Free Field (GFF), a subject which remained very active during the last fifteen years and led to several breakthroughs in the understanding of SLEs and CLEs via the Imaginary Geometry techniques, e.g. see~\cite{miller-sheffield-werner} and references therein.

Originally, this prediction was strongly supported by the convergence of dimer height functions to the GFF proved (for Temperleyan approximations on~$\mathbb Z^2$) by Kenyon \cite{kenyon-gff-a,kenyon-gff-b} and the fact that the level lines of the GFF are SLE(4) curves, see~\cite{schramm-sheffield} and~\cite{wang-wu}. More recently it received even more support due to the breakthrough works of Kenyon~\cite{Kenyon} and Dub\'edat~\cite{Dubedat} on the convergence of topological observables for double-dimer loop ensembles.
Our paper should be considered as a complement to the work of Dub\'edat who writes (see~\cite[Corollary~3]{Dubedat}) \emph{``By general principles, \ \ \dots\ \  the assumptions 1.~$(\mu_\delta)_\delta$ is tight, and 2. a probability measure $\mu$ on loop ensembles in a simply-connected domain D is uniquely characterized by the expectations of the functionals\ \ \dots\ \ imply weak convergence of the $\mu_\delta$'s to the $\mathrm{CLE}_4(D)$ measure as $\delta\to 0$.''}

To the best of our knowledge, there are still no available results on the first assumption (tightness), thus Kenyon's prediction should not be considered as fully proven yet. The main goal of this paper is to give a solid ground to the second assumption: we show that the topological observables treated by Dub\'edat in~\cite{Dubedat} \emph{do} characterize the measure on loop ensembles in the sense which is described below.

It is worth noting that several approaches to the convergence of (double-)dimer height functions to the GFF are known nowadays (e.g., see~\cite{berestycki-laslier-ray} and~\cite{bufetov-gorin}) besides the original one of Kenyon~\cite{kenyon-gff-a,kenyon-gff-b}, which is based on the analysis of the scaling limit of the Kasteleyn matrix by means of discrete complex analysis and is also the starting point for~\cite{Kenyon} and~\cite{Dubedat}. Also, the choice of discrete approximations~$\Omega^\delta$ to a simply connected domain~$\Omega$ is a very delicate question; see~\cite{russkikh-hedgehog} for another (not Temperleyan) special case when the discrete complex analysis machinery works well. To be able to build upon the results of~\cite{Dubedat}, below we assume that~$\Omega^\delta$ are Temperleyan approximations on the square grids of mesh~$\delta$ though this setup can be enlarged in several directions.

Recall that, given a Temperlean simply connected discrete domain~$\Omega^\delta\subset\delta{\mathbb Z}^2$, a \emph{double-dimer loop ensemble} on~$\Omega^\delta$ is obtained by superimposing two dimer configurations on~$\Omega^\delta$ chosen independently uniformly at random: this produces a number of loops and double-edges, the latter should be withdrawn. We denote by~$\Xi^\delta_{\Omega}$ the random collection of simple pairwise disjoint loops obtained in this way. The \emph{nested conformal loop ensemble} CLE(4) in~$\Omega$ is a conjectural limit of~$\Xi^\delta_\Omega$ as~$\delta\to 0$. The CLE(4) can be defined and effectively studied purely in continuum, see~\cite{sheffield-werner,qian-werner} and references therein for background. We denote by~$\Xi^\star_\Omega$ a random sample of this loop ensemble. Note that~$\Xi^\star_\Omega$ almost surely contains infinitely many loops but most of them are very small: almost surely, for each cut-off~$\varepsilon>0$ only finitely many of the loops of~$\Xi^\star_\Omega$ have diameter greater than~$\varepsilon$.

Let~$\lambda_1,\dots,\lambda_n\in\Omega$ be a collection of pairwise distinct punctures in~$\Omega$. A \emph{lamination}~$\Gamma$ is a \emph{finite} collection of disjoint simple loops in~$\Omega\setminus\{\lambda_1,\dots,\lambda_n\}$ considered up to homotopies. We call a lamination \emph{macroscopic} if each of these loops surrounds at least two of the punctures. For a random loop ensemble~$\Xi$ and a deterministic macroscopic lamination~$\Gamma$, let $\Xi\sim \Gamma$ denote the event that withdrawing all loops surrounding no more than one puncture from~$\Xi$ one obtains~$\Gamma$. We call the events~$\Xi\sim\Gamma$ \emph{cylindrical}: their probabilities (for all~$n\ge 1$, all~$\lambda_1,\dots,\lambda_n\in\Omega$ and all macroscopic laminations~$L$) determine the law of~$\Xi$ for reasonable topologies on the space of loop ensembles.


\begin{figure}
  \centering
  \begin{minipage}{.496\textwidth}
    \centering
    \tempskip{\includegraphics[clip, trim=0cm 20cm 0cm 1.2cm, width = 0.96\textwidth]{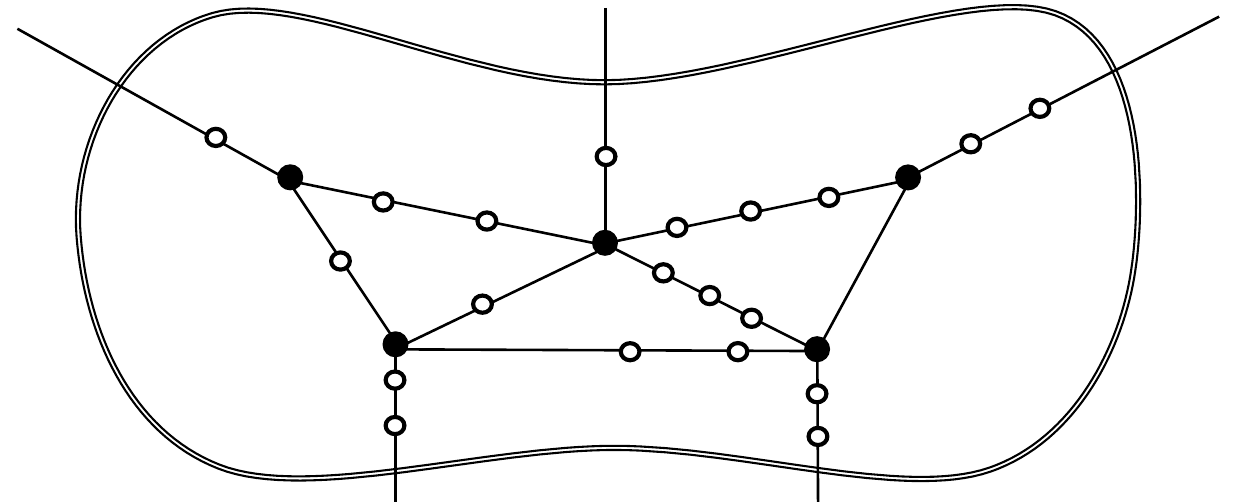}}
    \caption*{Given $\mathbf n$, put $n_e$ points at the interior of each edge $e\in\Ed$.}
  \end{minipage}
    \begin{minipage}{.496\textwidth}
    \centering
    \tempskip{\includegraphics[clip, trim=0cm 20cm 0cm 1.2cm, width = 0.96\textwidth]{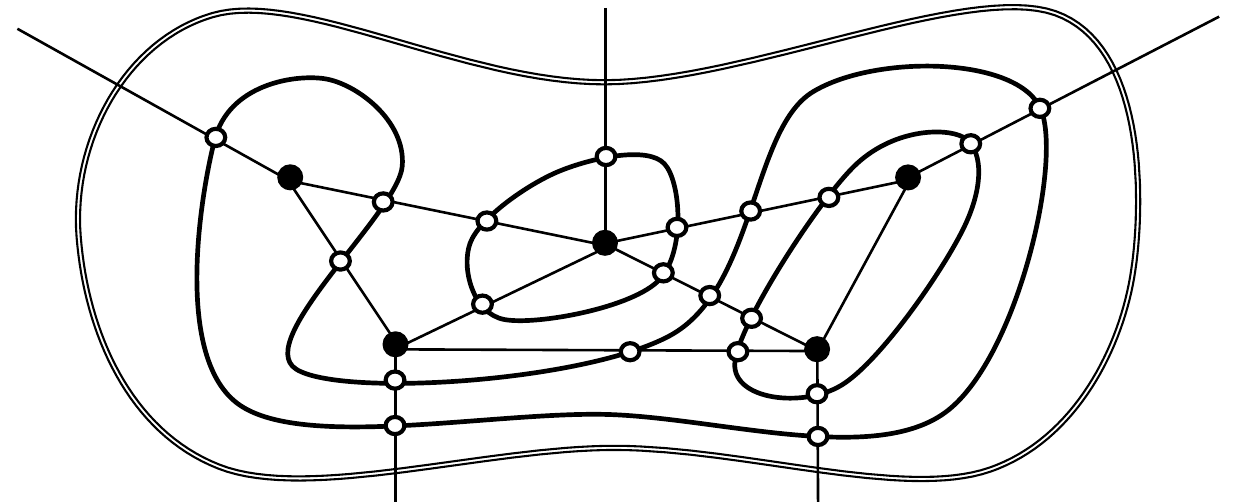}}
    \caption*{Connect these points by chords inside each triangle.}
  \end{minipage}
  \caption{A triangulation~$\G$ of~$\Omega\setminus\{\lambda_1,\dots,\lambda_n\}$ and a lamination~$\Gamma$ recovered from a multi-index~$\mathbf n\in\mathbb Z_{\ge 0}^{\Ed}$. One has~$|\Gamma|=\sum_{e\in\Ed}n_e=20$.}
  \label{fig:integers_to_laminations}
\end{figure}


\bl{An important notion that is constantly used in our paper is the \emph{complexity}~$|\Gamma|$ of a lamination~$\Gamma$. From now onwards, we fix a triangulation~$\G$ of~$\Omega\setminus\{\lambda_1,\dots,\lambda_n\}$ whose $n+1$ vertices correspond to $\lambda_1,\dots,\lambda_n$ and to the boundary of $\Omega$; see Fig.~\ref{fig:integers_to_laminations}. Note that we allow two triangles to share several edges; e.g., see Fig.~\ref{fig:n=2} for an example. However, for simplicity we do \emph{not} allow $\lambda_1,\ldots,\lambda_n$ or~$\partial\Omega$ to have degree~$1$ in the triangulation~$\G$.
\begin{defin}
Given a triangulation~$\G$ of~$\Omega\setminus\{\lambda_1,\dots,\lambda_n\}$, we define the complexity~$|\Gamma|$ of a lamination $\Gamma$ to be the minimal possible (after applying homotopies) number of intersections of loops constituting~$\Gamma$ with the edges~$e\in\Ed$ of~$\G$.
\end{defin}

\begin{rem}\label{rem:complexity}
(i) In fact, one can parameterize laminations on~$\Omega\setminus\{\lambda_1,\dots,\lambda_n\}$ by multi-indices~$\mathbf{n}=(n_e)\in\mathbb Z_{\ge 0}^\Ed$ satisfying certain conditions, see Fig.~\ref{fig:integers_to_laminations} and Section~\ref{ReparametrizingLaminations} for more details. Under this parametrization one has $|\Gamma|=\sum_{e\in\Ed}n_e$.

\smallskip

\noindent (ii) The notion of complexity introduced above \emph{depends} on the choice of a triangulation~$\G$. However, it is easy to see that, for each two such choices $\G'$ and $\G''$, the complexities $|\Gamma|'$, $|\Gamma|''$ differ no more than by a multiplicative factor independent of $\Gamma$.

\smallskip

\noindent (iii) Let us emphasize that, once $\G$ is fixed, the complexity $|\Gamma|$ \emph{cannot} be estimated via the number of loops in~$\Gamma$ (below we denote the latter quantity by~$\#\mathrm{loops}(\Gamma)$). Indeed, provided that $n\ge 3$, the complexity $|\Gamma|$ can be arbitrarily large even if $\#\mathrm{loops}(\Gamma)=1$.
\end{rem}}

We denote by
\[
X:=\Hom(\pi_1(\Omega\setminus\{\lambda_1,\dots,\lambda_n\})\to\sl)
\]
the (smooth) variety of $\sl$-representations of the free non-abelian fundamental group~$\pi_1(\Omega\setminus\{\lambda_1,\dots,\lambda_n\})$. Note that one could view~$X$ as~$(\sl)^n$  by fixing generators of the fundamental group but this viewpoint is not invariant enough for the analytic tools that we use below. Let~$\lloop{\lambda_i}$ denote the loop surrounding a single puncture~$\lambda_i$ and
\begin{equation}
  \label{eq:Xunip_def}
    X_{\mathrm{unip}}\ :=\ \{\rho\in X: \Tr(\rho(\lloop{\lambda_i}))=2\ \text{for all}\ i=1,\dots, n\}
\end{equation}
be the (non-smooth) subvariety of locally unipotent representations~$\rho\in X$. In this paper we study \emph{entire functions}~$f:X\to \mathbb C$ or~$f:X_{\mathrm{unip}}\to\mathbb C$, in the latter case we mean that~$f$ is continuous on~$X_{\mathrm{unip}}$ and is holomorphic on its regular part. Moreover, we are interested only in those entire functions that are invariant under the action of~$\sl$ on~$X$ or~$X_{\mathrm{unip}}$ given by the conjugation~$\rho(\cdot)\mapsto A^{-1}\rho(\cdot)A$, $A\in\sl$. \bl{This reflects the fact that $\sl$-representations $\rho\in X$, are used below to build observables that distinguish \emph{free} homotopy classes of 
loops in the punctured domain, which 
correspond to conjugacy classes of elements in the fundamental group. Below we use the notation~$\Fhol(X)^{\sl}$ and~$\Fhol(X_{\mathrm{unip}})^{\sl}$ for the spaces of holomorphic $\sl$-invariant functions on $X$ and on $X_\mathrm{unip}$, respectively.

}

Given a lamination (not necessarily macroscopic)~$\Gamma$ on~$\Omega\setminus\{\lambda_1,\dots,\lambda_n\}$, we set
\begin{equation}
\label{eq:f-gamma-def}
f_\Gamma(\rho)\ :=\ \prod\nolimits_{\gamma\in\Gamma} \Tr(\rho(\gamma)),\qquad \rho\in X.
\end{equation}
Since~$\Tr A = \Tr A^{-1}$ for~$A\in\sl$, this definition does not require to fix an orientation of the loops~$\gamma\in\Gamma$. Clearly, one has~$f_\Gamma\in \Fhol(X)^{\sl}$ and this function can be also treated as an element of~$\Fhol(X_{\mathrm{unip}})^{\sl}$ by taking the restriction to locally unipotent monodromies~$\rho\in X_{\mathrm{unip}}$.

The main results of our paper can be loosely formulated as follows: each entire function~$f\in\Fhol(X)^{\sl}$ admits a unique expansion via the functions~$f_\Gamma$ \bl{while} each function~$f\in\Fhol(X_{\mathrm{unip}})^{\sl}$ admits a unique expansion via the functions~$f_\Gamma$ indexed by \emph{macroscopic} laminations~$\Gamma$, with coefficients decaying faster than exponentially.

For~$\rho\in X_{\mathrm{unip}}$, let
\begin{align}
\label{eq:tau-delta-def}
 \tau^\delta(\rho)\ :=&\ \mathbb E_{\operatorname{dbl-d}}\left[\,\prod\nolimits_{\gamma\in\Xi^\delta(\Omega^\delta)} (\tfrac{1}{2}\Tr(\rho(\gamma)))\,\right]\ \\ \notag
 =&\ \sum\nolimits_{\Gamma - \mathrm{macroscopic}} p^\delta_\Gamma f_\Gamma(\rho),\qquad p^\delta_\Gamma\ :=\ 2^{-\#\mathrm{loops}(\Gamma)}\mathbb P_{\operatorname{dbl-d}}\left[\,\Xi^\delta_\Omega\sim \Gamma\,\right],
\end{align}
and similarly
\begin{align}\label{eq:tau-cle-def}
\tau^\star(\rho)\ :=& \ \mathbb E_{\operatorname{CLE(4)}}\!\left[\,\prod\nolimits_{\gamma\in\Xi^\star(\Omega)} (\tfrac{1}{2}\Tr(\rho(\gamma)))\,\right]\\ \notag
=&\ \sum\nolimits_{\Gamma - \mathrm{macroscopic}} p^\star_\Gamma f_\Gamma(\rho),\qquad
p^\star_\Gamma\ :=\ 2^{-\#\mathrm{loops}(\Gamma)}\mathbb P_{\operatorname{CLE(4)}}\left[\,\Xi^\star_\Omega\sim \Gamma\,\right].
\end{align}
Following~\cite{Dubedat} we call the functions~$\tau^\delta$ and~$\tau^\star$ \emph{topological correlators} of the loop ensembles~$\Xi^\delta$ and~$\Xi^\star$, respectively.

While~$\tau^\delta$ is actually a finite linear combination of~$f_\Gamma$, one should be more careful with the infinite series~\eqref{eq:tau-cle-def}. It is checked in~\cite{Dubedat} that~$\mathbb P_{\operatorname{dbl-d}}\left[\,\Xi^\star_\Omega\sim \Gamma\,\right]=O(R_0^{-|\Gamma|})$ for \emph{some}~$R_0>1$, therefore~$\tau^\star$ is correctly defined at least in a vicinity of the trivial representation~$\Id\in X_{\mathrm{unip}}$. It seems to be known in the folklore that these probabilities actually decay \emph{super-}exponentially as~$|\Gamma|\to\infty$ (e.g., see~\cite[Section~4]{wu-ising-exponents} and references therein for related results) but we were unable to find an explicit reference to this fact and thus prefer to keep it as an assumption in Corollary~\ref{cor:p=p-star}, see also Remark~\ref{rem:non-uniqueness}.

In the paper~\cite{Dubedat} Dub\'edat also introduced a notion of the \emph{isomonodromic tau-function} \mbox{$\tau_{\Omega}(\lambda_1,\dots,\lambda_n;\rho)$}, $\rho\in X_{\mathrm{unip}}$, on a simply connected domain~$\Omega$ which is defined as follows. For~$\Omega=\mathbb H$ (the upper half-plane), consider a representation of the fundamental group of the punctured Riemann sphere
\[
\rho':\mathbb C\mathrm P^1\setminus\{\lambda_1,\dots,\lambda_n,\overline{\lambda}_n,\dots,\overline{\lambda}_1\}\to \sl
\]
constructed so that the monodromies of~$\rho'$ around the punctures~$\lambda_i$ match those of~$\rho$ while the monodromies of~$\rho'$ around~$\overline{\lambda}_i$ are their inverses, see~\cite{Dubedat} for more details. For each~$\rho\in X_{\mathrm{unip}}$ one can check that the classical Jimbo--Miwa--Ueno~\cite{jimbo-miwa-ueno,Palmer} isomonodronic tau-function~$\tau_{\mathbb C\mathrm P^1}(\lambda_1,\dots,\lambda_n,\overline{\lambda}_n,\dots,\overline{\lambda}_1;\rho')$ can be normalized so that it equals to~$1$ when all the pairs of punctures~$\lambda_i,\overline{\lambda}_i$ collide on the real line. \bl{A priori, this tau function is well-defined only on the universal cover of the space parameterizing pairwise distinct punctures $\lambda_1,\dots,\lambda_n\in\Omega$. However, it turns out that}
\[
\tau_{\mathbb H}(\rho;\lambda_1,\dots,\lambda_n):= \tau_{\mathbb C\mathrm P^1}(\rho';\lambda_1,\dots,\lambda_n,\overline{\lambda}_n,\dots,\overline{\lambda}_1)
\]
is invariant under braid moves as well as under M\"obius automorphisms of the upper half-plane~$\mathbb H$. This allows one to define~$\tau_\Omega(\rho;\lambda_1,\dots,\lambda_n)$ for general simply connected domains~$\Omega$ by conformal invariance.

Although one usually considers an isomonodromic tau-function as a function of~$\lambda_1,\dots,\lambda_n$, the one discussed above can be also viewed as a function of a locally unipotent representation~$\rho$ as its multiplicative normalization does not depend on~$\rho$. In our paper both~$\Omega$ and~$\lambda_1,\dots,\lambda_n$ can be usually thought of as fixed once forever, thus from now onwards we use the shorthand notation
\begin{equation}
\label{eq:tau-iso-def}
\tau(\rho):=\tau_\Omega(\lambda_1,\dots,\lambda_n;\rho),\qquad \rho\in X_{\mathrm{unip}}.
\end{equation}
From the above construction one can see that~$\tau\in \Fhol(X_{\mathrm{unip}})^{\sl}$, note that this also can be easily deduced from the following theorem, which is the main result of~\cite{Dubedat}.

Recall that the functions~$\tau^\delta$ and~$\tau^\star$ are defined by~\eqref{eq:tau-delta-def} and~\eqref{eq:tau-cle-def}.
\begin{theorem}[Dub\'edat]
\label{thm:dubedat} Let $\Omega$ be a planar simply connected domain, $\Omega^\delta$ be a sequence of Temperleyan approximations to~$\Omega$, and~$\lambda_1,\dots,\lambda_n\in\Omega$. Then, the following holds:

\noindent (i) For each locally unipotent representation~$\rho\in X_{\mathrm{unip}}$ one has
\[
\tau^\delta(\rho)\ \to\ \tau(\rho)\quad \text{as}\ \ \delta\to 0
\]
and the convergence is uniform on compact subsets of~$X_{\mathrm{unip}}$.

\noindent (ii) Moreover, $\tau(\rho)=\tau^\star(\rho)$ if~$\rho\in X_{\mathrm{unip}}$ is close enough to the trivial representation.
\end{theorem}

The next theorem \bl{(see also Theorem~\ref{main_theorem_actually})} is the main result of our paper.

\begin{theorem}
\label{thm:main-thm} Each entire function~$f\in\Fhol(X_{\mathrm{unip}})^{\sl}$ admits a unique expansion
\begin{equation}
\label{eq:expansion-of-f}
f(\rho)\ =\ \sum\nolimits_{\Gamma - \mathrm{macroscopic}}p_\Gamma f_\Gamma(\rho),\qquad \rho\in X_{\mathrm{unip}},
\end{equation}
where the functions~$f_\Gamma$ are given by~\eqref{eq:f-gamma-def} and~$|\Gamma|^{-1}\log |p_\Gamma|\to -\infty$ as~$|\Gamma|\to\infty$. Moreover, for each~$R>0$ there exists a compact subset~$K_R\subset X_{\mathrm{unip}}$ and a constant~$C_R>0$ independent of~$f$ such that one has
\begin{equation}
\label{eq:estimate-of-pG}
|p_\Gamma|\le C_R\cdot R^{-|\Gamma|}\cdot \|f\|_{L^\infty(K_R)}
\end{equation}
for all macroscopic laminations~$L$.
\end{theorem}

\begin{rem}
\label{rem:non-uniqueness} 
It is worth noting that our results do \emph{not} guarantee the uniqueness of the expansion~\eqref{eq:expansion-of-f} for functions defined just in a small vicinity of~$\Id\in X_{\mathrm{unip}}$. To illustrate a possible catch one can think about expanding entire functions of one complex variable in the basis~$1,z-1,z^2-z,\dots,z^n-z^{n-1},\dots$. In the full plane such expansions always exist and are unique but~$1+(z-1)+(z^2-z)+\dots =0$ in a vicinity of the origin. Since the functions~$f_\Gamma$ are far from being a Fourier basis, we expect a similar
(though more involved) phenomenon in our setup.
\end{rem}

It is easy to see that a combination of Theorem~\ref{thm:dubedat} and Theorem~\ref{thm:main-thm} imply the convergence of probabilities of cylindrical events. Let
\[
\tau(\rho)\ =\ \sum\nolimits_{\Gamma - \mathrm{macroscopic}}p^{\mathrm{iso}}_\Gamma f_\Gamma(\rho),\qquad \rho\in X_{\mathrm{unip}},
\]
be the expansion of the isomonodromic tau-function provided by Theorem~\ref{thm:main-thm}.

\begin{cor}
\label{cor:p-delta->p}
Let $\Omega$ be a planar simply connected domain, $\Omega^\delta$ be a sequence of Temperleyan approximations to~$\Omega$, and $\lambda_1,\dots,\lambda_n\in\Omega$. Then, for each macroscopic lamination $\Gamma$ on $\Omega\setminus\{\lambda_1,\dots,\lambda_n\}$, one has
\[
p^\delta_\Gamma=2^{-\#\mathrm{loops}(\Gamma)}\mathbb P_{\operatorname{dbl-d}}\left[\,\Xi^\delta_\Omega\sim \Gamma\,\right]\ \to\ p^{\mathrm{iso}}_\Gamma\quad \text{as}\ \ \delta\to 0.
\]
Moreover, $R^{|\Gamma|}\cdot |p^\delta_\Gamma - p^{\mathrm{iso}}_\Gamma|\to 0$ as~$\delta\to 0$ uniformly in~$\Gamma$ for each~$R>0$.
\end{cor}

\begin{proof} By definition,~$p^\delta_\Gamma$ are nothing but the coefficients in the expansion~\eqref{eq:expansion-of-f} of the function~$\tau^\delta$. Therefore, for each~$R>0$, Theorem~\ref{thm:main-thm} implies the uniform estimate
\[
R^{|\Gamma|}\cdot|p^\delta_\Gamma-p^{\mathrm{iso}}_\Gamma| \le C_R\cdot \|\tau^\delta-\tau\|_{L^{\infty}(K_R)}
\]
and the right-hand side vanishes as~$\delta \to 0$ due to Theorem~\ref{thm:dubedat}(i).
\end{proof}

\begin{cor}
\label{cor:p=p-star}
In the same setup, assume that
\begin{equation}
\label{eq:superexp-decay-assumption}
\mathbb P_{\operatorname{CLE(4)}}\left[\,\Xi^\star_\Omega\sim \Gamma\,\right] = O(R^{-|\Gamma|})\ \ \text{as}\ \ |\Gamma|\to\infty,\ \ \text{for all}\ \ R>0.
\end{equation}
Then,~$p^{\mathrm{iso}}_\Gamma=p^\star_\Gamma= 2^{-\#\mathrm{loops}(\Gamma)}\mathbb P_{\operatorname{CLE(4)}}\left[\,\Xi^\star_\Omega\sim \Gamma\,\right]$ for all macroscopic laminations~$\Gamma$.
\end{cor}

\begin{proof} If the assumption~\eqref{eq:superexp-decay-assumption} holds, then the series~\eqref{eq:tau-cle-def} converges for all~$\rho\in X_{\mathrm{unip}}$ and~$\tau^\star\in \Fhol(X_{\mathrm{unip}})^{\sl}$. Due to Theorem~\ref{thm:dubedat}(ii) the entire functions~$\tau$ and~$\tau^\star$ coincide in a vicinity of~$\rho=\Id$ and hence everywhere on~$X_{\mathrm{unip}}$. Therefore, the uniqueness of the expansion~\eqref{eq:expansion-of-f} gives~$p^{\mathrm{iso}}_\Gamma=p^\star_\Gamma$ for all macroscopic laminations~$\Gamma$.
\end{proof}

The main ideas of the proof of Theorem~\ref{thm:main-thm} are discussed in Section~\ref{sec:strategy}. We conclude the introduction by the following vague remark on a possible \emph{deformation} of the topological observables~\eqref{eq:tau-cle-def}. Though the first naive idea would be just to replace the CLE(4) measure in their definition by CLE($\kappa$) with~$\kappa\ne 4$, thus obtained functions do not look very natural. In view of the material discussed in Section~\ref{skein_algebra_section} it actually looks more promising to simultaneously deform the functions~$f_\Gamma$ given by~\eqref{eq:f-gamma-def} by using the quantum trace functionals~\cite{BonahonWong} instead of the usual traces. Having in mind the famous predictions on the scaling limits of loop O($n$) and FK($q$) models (e.g., see~\cite[Section~2.4]{schramm-icm-06} and \cite[Section~2]{smirnov-icm-06}), it sounds plausible that the quantization parameter should be then tuned so that the simple loop weight in the corresponding skein algebras is equal to~$\pm 2\cos(4\pi/k)$. We believe that developing tools to analyze such deformations might be of great interest, both from CLE and lattice models perspectives.


\section{Toy example $n=2$ and the strategy of the proof}\label{sec:strategy}

In this section we informally describe the main ideas of the proof of Theorem~\ref{thm:main-thm}. Sometimes we use the case~$n=2$ as a toy example. Certainly, this is a classical and well-studied setup (e.g., see~\cite[Section~10]{Kenyon} and~\cite[Corollary~2]{Dubedat}): each macroscopic lamination on~$\Omega\setminus\{\lambda_1,\lambda_2\}$ is given by~$k\ge 0$ loops homotopic to the simple loop~$\bl{\ell}$ surrounding both punctures. Let~~$z:=\Tr(\rho(\bl{\ell}))\in\mathbb{C}$ and~$p_k^\delta$ be the probability, divided by~$2^k$, to see exactly $k$ copies of~$\bl{\ell}$ in the double-dimer model loop ensemble on~$\Omega^\delta$. In this situation our main result can be rephrased as follows: the convergence of entire functions
\[
\textstyle \tau^\delta(z)=\sum_{k\ge 0} p_k^\delta z^k \ \to \tau(z)=\sum_{k\ge 0} p_kz^k\quad \text{as}\ \ \delta\to 0
\]
on compact subsets of~$\mathbb{C}$ implies that~$p_k^\delta\to p_k$ as~$\delta\to 0$ for each~$k\ge 0$.

Unfortunately, such a straightforward argument does not work for~$n>2$ since the set of functions~$f_\Gamma:X_{\mathrm{unip}}\to \mathbb C$ indexed by macroscopic laminations does not have a structure of the Fourier basis. \bl{Thus, more involved tools should be used.

Recall that the manifold~$X=\Hom(\pi_1(\Omega\setminus\{\lambda_1,\dots,\lambda_n\})\to\sl)$ can be viewed as~$(\sl)^n$ (a parametrization is given by a choice of generators of the fundamental group) which provides it with the structure of an affine algebraic group. The definition~\eqref{eq:Xunip_def} makes $X_{\mathrm{unip}}$ to be an algebraic subvariety of $X$. We denote by $\F(X)$ and $\F(X_{\mathrm{unip}})$ the rings of algebraic functions on these varieties. Note that $\sl$ acts on $X$ and $X_{\mathrm{unip}}$ algebraically, so let us denote the corresponding rings of invariants by} 
\[
  \bl{\F(X)^{\sl}\subset \Fhol(X)^{\sl}\quad \text{and}\quad \F(X_{\mathrm{unip}})^{\sl}\subset \Fhol(X_{\mathrm{unip}})^{\sl}.}
\]

\bl{One can relatively easily deduce from the Fock-Goncharov theorem~\cite[Theorem~12.3]{FG} that~$f_\Gamma$ form an {algebraic} basis of the space \mbox{$\F(X_{\mathrm{unip}})^{\sl}$}. In other words, each \emph{polynomial} function $f\in\Fhol(X_{\mathrm{unip}})^{\sl}$ admits a unique (finite) expansion \mbox{$f(\rho)=\sum_{\Gamma - \text{macroscopic}}p_\Gamma f_\Gamma(\rho)$}. However, even the uniqueness of such (infinite) expansions for arbitrary \emph{entire} functions $f$ is not at all obvious. To see a possible difficulty, the reader can think, e.g., about replacing the Fourier basis~$1,z,z^2,\dots$ by its lower-diagonal transform~$1,z-2,z^2-4z,\dots,z^n-2^nz^{n-1},\dots$ in the toy example discussed above: one has $1+\frac{1}{2}(z-2)+\ldots+2^{-n(n+1)/2}(z^n-2^nz^{n-1})+\ldots=0$, $z\in\mathbb{C}$.}

Even if the aforementioned existence and uniqueness issues are settled, it still might be problematic to extract the convergence of coefficients~$p_k^\delta$ from the convergence of functions~$\tau^\delta$ unless an a priori estimate similar to~\eqref{eq:estimate-of-pG} is available. 
Since, to the best of our knowledge, no explicit analogue of the Cauchy formula (which settles the toy case~$n=2$) on~$X_{\mathrm{unip}}$ is known if~$n>2$, we develop a set of general tools to analyze $\sl$-invariant entire functions on~$X_{\mathrm{unip}}$ as sketched below.


\begin{figure}
  \centering
  \begin{minipage}{.496\textwidth}
    \centering
    \tempskip{\includegraphics[clip, trim=0cm 20cm 0cm 1.2cm, width = 0.96\textwidth]{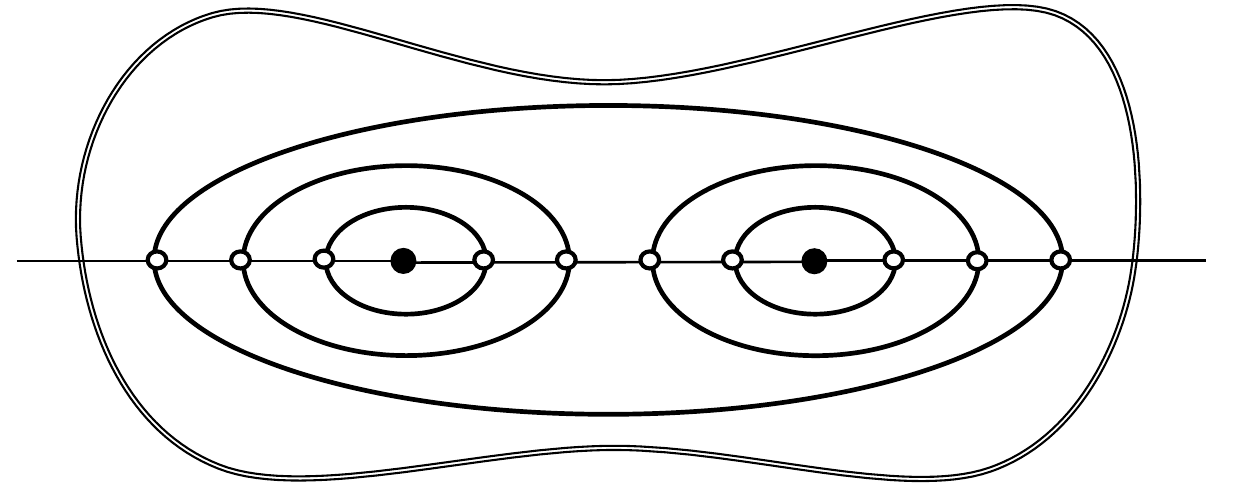}}
  \end{minipage}
    \begin{minipage}{.496\textwidth}
    \centering
    \tempskip{\includegraphics[clip, trim=0cm 20cm 0cm 1.2cm, width = 0.96\textwidth]{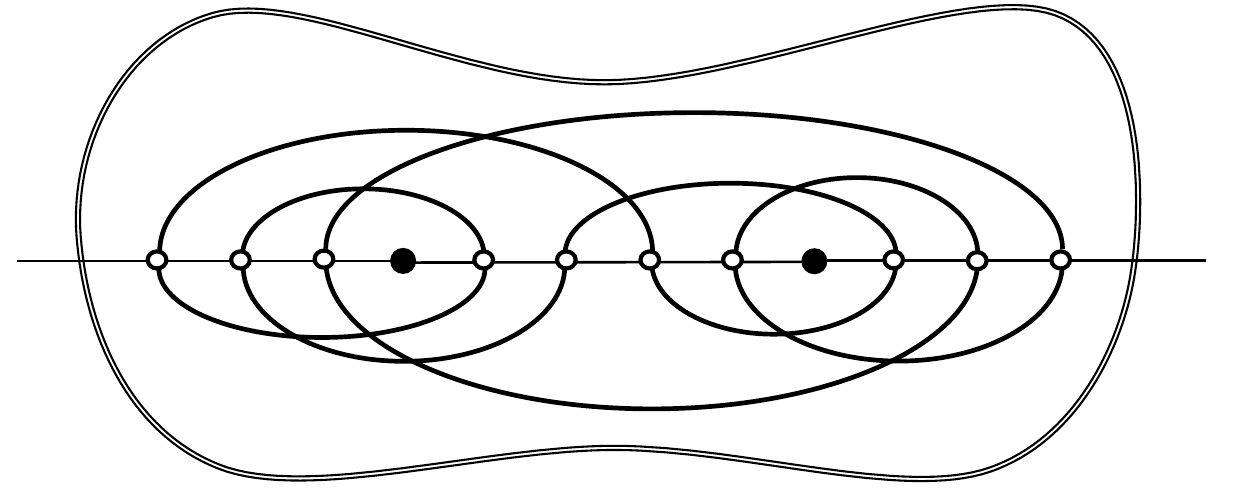}}
  \end{minipage}
  \caption{The case~$n=2$. A lamination corresponding to the multi-index~$\mathbf{n}=(3,4,3)$ and one of the multi-curves from the set~$\mathcal{C}(3,4,3)$.}
  \label{fig:n=2}
\end{figure}


\bl{Though our ultimate goal is to study the space~$\Fhol(X_{\mathrm{unip}})^{\sl}$, we begin with discussing the rings~$\F(X)^{\sl}$ and~$\Fhol(X)^{\sl}$ before descending this analysis to~$\F(X_{\mathrm{unip}})^{\sl}$ and $\Fhol(X_\mathrm{unip})^{\sl}$. Since $X$ carries a (non-canonical) structure of an algebraic group (whereas $X_\mathrm{unip}$ does not), we can expect to have a natural vector-space basis of $\F(X)^{\sl}$ obtained by an application the Peter-Weyl theorem.
It turns out that such a basis can be labeled by laminations and is related to the basis $\{f_\Gamma\}$ by a lower-triangular transform.}
Following Fock and Goncharov~\cite[Section~12]{FG} this goes as follows:

-- One encodes the laminations~$\Gamma$ by multi-indices~$\mathbf n=(n_e)_{e\in\Ed}\in \mathbb Z_{\ge 0}^\Ed$ indexed by the edges of a triangulation~$\G$ of~$\Omega\setminus\{\lambda_1,\dots,\lambda_n\}$, see Fig.~\ref{fig:integers_to_laminations}. If~$n=2$, these multi-indices are just triples~$(n_1,n_2,n_3)$ of non-negative integers (see Fig.~\ref{fig:n=2}) \bl{such that
\begin{equation}
\label{eq:lam-intro}
\bl{(n_1,n_2,n_3)\ \text{satisfy the triangle inequality}\ \ \&\ \ n_1+n_2+n_3\ \text{is even};}
\end{equation}
the condition~\eqref{eq:lam-intro} is called the \emph{lamination condition}.}

-- Flat $\sl$ connections on~$\Omega\setminus\{\lambda_1,\dots,\lambda_n\}$ can be parameterized by putting transition matrices~$A_e\in\sl$ onto the edges of~$\G$ and factorizing over the natural action (change of the bases of the rank-two vector bundle) of~$\sl$ matrices assigned to the faces of~$\G$. This parametrization provides a natural isomorphism
\begin{equation}
  \label{eq:F(X)_cong_F(slEd)}
  \F(X)^{\sl}\cong \F(\sl^\Ed)^{\sl^\Fa}.
\end{equation}

-- As already mentioned, the functions~$f_{(n_1,n_2,n_3)}\in\F(X)^{\sl}$ (for simplicity, we focus on the case~$n=2$ though the same arguments work well in the general case) do not have a Fourier basis structure. Nevertheless, \bl{applying the Peter-Weyl theorem to the right-hand side of~\eqref{eq:F(X)_cong_F(slEd)}
one can obtain} another collection of functions~$g_{(n_1,n_2,n_3)}\in\F(X)^{\sl}$ indexed by the \emph{same set} of triples satisfying the lamination condition~\eqref{eq:lam-intro} which is an orthogonal basis in the space~$L^2(\su^\Ed)^{\su^\Fa}$; we call these functions the \emph{Peter--Weyl basis}. In fact, the functions~$g_{(n_1,n_2,n_3)}$ can be constructed by the following symmetrization procedure. Let~$\mathcal{C}(n_1,n_2,n_3)$ be the set of all possible collections of (not necessarily simple or disjoint) loops obtained by concatenating two collections of arcs with~$(n_1,n_2,n_3)$ endpoints drawn inside each of the two faces of~$\G$ so that no arc connects two points on the same edge, see Fig.~\ref{fig:n=2}. Then
\[
g_{(n_1,n_2,n_3)}\ =\ |\mathcal{C}(n_1,n_2,n_3)|^{-1}\cdot \sum\nolimits_{\C\in\mathcal{C}(n_1,n_2,n_3)} f_\C,
\]
where~$f_\C(\rho):=\prod_{\gamma\in\C} \Tr(\rho(\gamma))$.

-- The Fock--Goncharov theorem claims that the bases~$f_\Gamma$ and~$g_\Gamma$ are related by a lower-triangular transform:
\[
g_\Gamma=\sum\nolimits_{\Delta:\Delta\le\Gamma}c_{\Gamma\Delta}f_\Delta,
\]
where the partial order on the set of laminations is given by the partial order on the set of multi-indices~$\mathbf n\in\mathbb Z_{\ge 0}^\Ed$.

By analogy with entire functions \bl{on $\mathbb{C}$,} one can expect that each holomorphic function~$f\in \Fhol(X)^{\sl}$ admits a Fourier-type expansion~$f=\sum_{\Gamma}q_\Gamma g_\Gamma$ with coefficients~$q_\Gamma$ decaying faster than exponentially. Formally, this implies that one can also write
\[
f=\sum\nolimits_{\Delta}p_\Delta f_\Delta,\quad \text{where}\quad p_{\Delta}=\sum\nolimits_{\Gamma:\Delta\le\Gamma}c_{\Gamma\Delta}q_\Gamma,
\]
but there is a catch: to pass from the former series to the latter rigorously one needs an exponential upper bound for the coefficients~$c_{\Gamma\Delta}$ of the Fock--Goncharov change of the bases. To the best of our knowledge, such an estimate is not available in the existing literature thus we prove the (non-optimal) upper bound~$|c_{\Gamma\Delta}|\le 4^{|\Gamma|}$ in our paper.

It is well-known that the identities between the functions~$f_\Gamma$, $f_\C$ and~$g_\Gamma$ can be equivalently written in terms of the Kauffman \emph{skein algebras}~$\sk(M,\pm 1)$ of framed knots in~$M:=(\Omega\setminus\{\lambda_1,\dots,\lambda_n\})\times [0,1]$. In fact, $\F(X)^{\sl}\cong \sk(M,-1)\cong \sk(M,1)$ and the skein relations reflect the identity~$\Tr(AB)+\Tr(AB^{-1})=\Tr(A)\Tr(B)$ for \mbox{$A,B\in\sl$}. This isomorphism suggests the idea of expanding each of the terms~$f_\C$ in the definition of~$g_\Gamma$ as~$f_\C=\sum_{\Delta:\Delta\le\Gamma}c_{\C\Delta}f_\Delta$ by resolving all the crossings of~$\C$ via the skein relations. Unfortunately, this is not an easy thing to do: a collection of loops~$\C$ may contain about~$|\Gamma|^2$ crossings, hence one must analyse highly non-trivial cancelations arising along the way in order to end up with an exponential bound for~$c_{\C\Delta}$.

Following the advice of Vladimir Fock we circumvent these complications by using the positivity of the co-called \emph{bracelet basis} of~$\sk(M,1)$ proved by D.~Thurston~\cite{Thurston} in combination with a representation of~$\sk(M,1)$ in the space of Laurent polynomials coming from Thurston's shear coordinates of hyperbolic structures on~$\Omega\setminus\{\lambda_1,\dots,\lambda_n\}$ (e.g., see~\cite{BonahonWong}). Remarkably enough, in this representation all functions $f_\C$ are mapped to Laurent polynomials with positive integer coefficients. Together, these positivity results imply the desired exponential estimate of coefficients as the sums contain no cancellations anymore.

The arguments briefly described above allows one to prove counterparts of our main results for holomorphic functions living on the whole manifold~$X$ and their expansions in the basis~$f_\Gamma$ indexed by all, not necessarily macroscopic, laminations. The last but not the least part of our analysis is devoted to the translation of these results to holomorphic functions living on the subvariety~$X_{\mathrm{unip}}\subset X$. \bl{Recall that we denote by~$\lloop{\lambda_i}$ the loop surrounding a single puncture~$\lambda_i$.} Note that each lamination admits a unique decomposition into a macroscopic lamination~$\Gamma$, $k_1$ copies of~$\lloop{\lambda_1}$, $k_2$ copies of~$\lloop{\lambda_2}$ etc., we use the notation~$\Gamma\sqcup \lloop{\lambda_{\mathbf k}}$ to describe such a lamination.

The main ingredient of the proof of the existence part of Theorem~\ref{thm:main-thm} is a `controlled' extension of holomorphic functions from~$X_{\mathrm{unip}}$ to~$X$ provided by Manivel's \bl{Ohsawa--Takegoshi}-type theorem~\cite{Manivel}. Roughly speaking, given~$f\in\Fhol(X_{\mathrm{unip}})^{\sl}$ one can first extend it to a function~$F\in\Fhol(X)^{\sl}$ and then group all the terms of the expansion of~$F$ that correspond to the laminations~$\Gamma\sqcup \lloop{\lambda_{\mathbf k}}$ with a fixed~$\Gamma$:
\[
  \textstyle \sum\nolimits_{\mathbf k\in\mathbb Z_{\ge 0}^n}p_{\Gamma\sqcup \lloop{\lambda_{\mathbf k}}}f_{\Gamma\sqcup \lloop{\lambda_{\mathbf k}}}(\rho)\ =\ \bl{\Bigl( \sum\nolimits_{\mathbf k\in\mathbb Z_{\ge 0}^n}2^{|\mathbf k|}p_{\Gamma\sqcup \lloop{\lambda_{\mathbf k}}}\Bigr)}\cdot f_{\Gamma}(\rho) \quad \text{for}\ \ \rho\in X_{\mathrm{unip}}.
\]

Finally, the uniqueness part of Theorem~\ref{thm:main-thm} can be deduced from a version of Hilbert's Nullstellensatz for invariant holomorphic functions on~$X$ as follows. Assume that a series \mbox{$f:=\sum_{\Gamma - \text{macroscopic}}p_\Gamma f_\Gamma$} vanishes on~$X_{\mathrm{unip}}$. Since~$X_{\mathrm{unip}}$ is cut off from~$X$ by the equations~$f_{\lloop{\lambda_i}}-2=0$, one can find functions~$h_i\in\Fhol(X)^{\sl}$ such that
\[
f(\rho)=h_1(\rho)\cdot (f_{\lloop{\lambda_1}}(\rho)-2)+\ldots +h_n(\rho)\cdot (f_{\lloop{\lambda_n}}(\rho)-2)\quad \text{for}\ \ \rho\in X.
\]
Due to the existence of expansions~\eqref{eq:expansion-of-f} one has~$h_i=\sum_{\Gamma\sqcup \lloop{\lambda_{\mathbf k}}}p^{(i)}_{\Gamma\sqcup \lloop{\lambda_{\mathbf k}}}f^{\vphantom{()}}_{\Gamma\sqcup \lloop{\lambda_{\mathbf k}}}$. Using the uniqueness of the expansion of~$f\in \Fhol(X)^{\sl}$ in the basis~$f_{\Gamma\sqcup \lloop{\lambda_{\mathbf k}}}$ and the fact that the product~$f_{\Gamma\sqcup \lloop{\lambda_{\mathbf k}}}\cdot f_{\lloop{\lambda_i}}$ is again a basis function one obtains the identity
\[
\textstyle p_\Gamma f_\Gamma(\rho) = \sum_{\mathbf k\in \mathbb Z_{\ge 0}^n}\sum_{i=1}^n p^{(i)}_{\Gamma\sqcup \lloop{\lambda_{\mathbf k}}} f^{\vphantom{()}}_{\Gamma\sqcup \lloop{\lambda_{\mathbf k}}}(\rho)\cdot (f_{\lloop{\lambda_i}}(\rho)-2)\quad \text{for}\ \ \rho\in X
\]
by collecting all the terms corresponding to a given macroscopic lamination~$\Gamma$. The right-hand side vanishes at~$\rho=\Id$ and hence~$p_\Gamma=0$.

Certainly, the informal discussion given in this section is far from being complete or rigorous, with many important details omitted. For instance, we actually work with functions $F\in\Fhol(\mathbb B_R)^{\su^{\Fa}}$ defined on poly-balls~$\mathbb B_R\subset (\mathbb C^{2\times 2})^{\Ed}$ rather than with~$F\in \Fhol(\sl^{\Ed})^{\sl^\Fa}$ as sketched above. Nevertheless, we hope that this discussion might help the reader to understand the general structure of our arguments and the set of tools used in the proof.

The rest of the paper is organized as follows. We collect relevant basic facts of the representation theory of~$\sl$ and discuss the Fock--Goncharov theorem in Section~\ref{representation_theory_section}. Section~\ref{skein_algebra_section} is devoted to the proof of the exponential estimate of the coefficients~$c_{\Gamma\Delta}$: as explained above, the Kauffman skein algebra~$\sk(M,1)$ plays a central role here. In Section~\ref{complex_analisys_section} we prove our main results. In particular, Section~\ref{sub:HolExt} is devoted to holomorphic extensions of functions defined on compact subsets of~$X_{\mathrm{unip}}$ and Section~\ref{sub:Nullstellensatz} contains a precise version of the Nullstellensatz that we need. All these ingredients are used in Section~\ref{sub:Expansions_macro} to prove Theorem~\ref{main_theorem_actually} which is a slightly stronger version of Theorem~\ref{thm:main-thm}.

\section{Preliminaries and Fock--Goncharov theorem}\label{representation_theory_section}

Let $X = \Hom(\pi_1(\Omega\setminus\{\lambda_1,\dots,\lambda_n\})\to\sl)$ be the affine variety parameterizing representations of $\pi_1(\Omega\setminus\{\lambda_1,\dots,\lambda_n\})$ in~$\sl$, note that~$\X$ can be also thought of as an algebraic variety. Let $\F(\X)$ be the ring of algebraic functions on~$\X$. The group $\sl$ acts on $\X$ by conjugations and, since this action is algebraic, one can consider the ring of invariants $\F(\X)^{\sl}$. Clearly, all functions $f_{\Gamma}$ belong to $\F(\X)^{\sl}$. A famous theorem due to Fock and Goncharov~\cite[Theorem 12.3]{FG} states that these functions actually form a basis in the vector space $\F(\X)^{\sl}$. For the sake of completeness and in order to introduce a consistent notation, we begin this section with some basic facts of the representation theory of $\sl$ and then repeat the proof of this theorem in Section~\ref{fViagAndgViaf} following~\cite{FG}. In Section~\ref{further} we discuss extensions of the functions~$f_\Gamma$ and~$g_\Gamma$ from~$\sl^\Ed$ to the Euclidean space~$(\mathbb{C}^{2\times 2})^{\Ed}$.
Finally, Section~\ref{sec:lowerbound} is devoted to the lower bound for the norms of thus obtained extensions of~$g_\Gamma$ which play an important role in the core part of the paper.


\subsection{Basics of the representation theory of $\bm{\sl}$}\label{sl2reprTheory} In this section we collect basic facts of the representation theory of~$\sl$. We use the well known correspondence between representations of groups and their Lie algebras, which holds for~$\sl$. Due to this correspondence, one can work with the Lie algebra $\liesl$ instead of $\sl$ itself. The proofs are mostly omitted, an interested reader can easily find them in the classical literature (e.g., see~\cite[Lecture 11]{FultonHarris} for a nice exposition).

\begin{lemma}\label{liesl2basis}
  The Lie algebra $\liesl$ of $\sl$ is given by
  \begin{equation*}
    \liesl = \{e,f,h\,\mid\, [f,e] = h,\ [h,e] = -2e,\ [h,f] = 2f\}.
  \end{equation*}
\end{lemma}

Consider now an irreducible finite-dimensional representation $\liesl\to \End(V)$.

\begin{lemma}\label{liesl2eigenvectors}
(i) If $v\in V$ is an eigenvector of $h$, then $ev$ and $fv$ are eigenvectors of~$h$. Moreover, if $hv = \lambda v$, then $h(ev) = (\lambda - 2)(ev)$ and $h(fv) = (\lambda+2)(fv)$.
\smallskip

\noindent (ii) There exists a non-zero vector $v\in V$ such that the following holds:

-- $v$ is an eigenvector for $h$;

-- $ev = 0$;

-- vectors $v, fv, f^2v, \dots, f^{\dim V-1}v$ form a basis of $V$.

\smallskip

\noindent (iii) Let $hv=\lambda v$, where the vector $v$ is introduced in (ii). Then,
\[
\textstyle 0=\Tr h = \sum_{k = 0}^{\dim V-1}(\lambda+2k) = \dim V \cdot (\lambda + \dim V-1)
\]
and hence~$\lambda=-(\dim V-1)$.
\end{lemma}

The above considerations can be summarized as follows:

\begin{prop}\label{sl2basic}
(i) The irreducible finite-dimensional representations of $\liesl$ are enumerated by non-negative integers and $\dim V_n = n+1$, where $V_n$ denotes the $n$-th representation. For each $n$, the representation $V_n$ has a basis $v_0, v_1, \dots, v_n$ such that
\begin{equation*}
  \begin{split}
     hv_k & = (2k-n)v_k,\\
     ev_k & = kv_{k-1},\\
     fv_k & = (n-k)v_{k+1}.
  \end{split}
\end{equation*}

\noindent (ii) The representation $V_1$ corresponds to the standard matrix representation of $\sl$ and $V_n$ is isomorphic (as $\liesl$ module) to the space $\Sym^n V_1$, which is in its turn isomorphic to the space of homogeneous polynomials of two variables of degree $n$. Under this isomorphism $v_k = x^{n-k}y^k$, where $\{x,y\}$ is the basis of $V_1$.
\end{prop}

Recall that the action of $\liesl$ on the tensor product of representations is defined as $a\mapsto a\otimes \Id + \Id\otimes a$.

\begin{lemma}\label{multiplicationofgen}
(i) For each $n,m\ge 0$, the following isomorphism of~$\liesl$ modules
holds:
\begin{equation}
  V_n\otimes V_m \simeq V_{n+m} \oplus V_{n+m-2}\oplus V_{n+m-4}\oplus\dots\oplus V_{|m-n|}.
  \label{eq:productOfTwoReprExpandsAs}
\end{equation}
The projection onto the first component is given by the symmetrization
\[
V_n\otimes V_m \simeq \Sym^n V_1\otimes \Sym^m V_1 \to \Sym^{n+m} V_1 \simeq V_{n+m}.
\]
\noindent (ii)
Furthermore, one has
\begin{equation}
\label{eq:V_1^n-decomposition}
V_1^{\otimes n} \simeq \bigoplus\nolimits_{0\le m\le n}V_{m}^{\oplus l(n,m)},\quad \text{where}\quad l(n,m)=\frac{m+1}{n+1}\binom{n+1}{\frac{1}{2}(n\!-\!m)}
\end{equation}
if~$n-m$ is even and~$l(n,m)=0$ if~$n-m$ is odd. The projection onto the first component is given by the symmetrization~$V_1^{\otimes n}\to \Sym^n V_1\simeq V_n$, note that~$l(n,n)=1$.
\end{lemma}

Given a representation  $\sl\to \End(V)$, we denote by $V^{\sl}$ the subspace of invariant (under the action of~$\sl$) vectors in~$V$; \bl{this subspace is nothing but the sum of all copies of $V_0$ arising in the decomposition of $V$ into irreducible representations.}

\begin{cor}\label{laminationCondition}
  Let $n,m,k$ be non-negative integers. The subspace $(V_n\otimes V_m\otimes V_k)^{\sl}$ is one-dimensional if
  \begin{equation} \label{eq:lamCon}
  \bl{|n-m| \leq k \leq n+m \quad \text{and}\quad n+m+k\ \text{is even}.}
  \end{equation}
 Otherwise, this invariant subspace is trivial.
\end{cor}

\begin{proof}
  Using \eqref{eq:productOfTwoReprExpandsAs} one gets
  \begin{equation*}
    (V_n\otimes V_m\otimes V_k)^{\sl} \simeq \bigoplus\nolimits_{l\,:\,0\leq l\leq \min(n,m)} \left( V_{n+m-2l}\otimes V_k \right)^{\sl}.
  \end{equation*}
  Applying~\eqref{eq:productOfTwoReprExpandsAs} again, one sees that $\left( V_{n+m-2l}\otimes V_k \right)^{\sl}$ is non-zero if and only if $n+m = k+2l$. Moreover, in this case $\dim_{\mathbb C} \left( V_{n+m-2l}\otimes V_k \right)^{\sl} = 1$ since~$V_0$ may appear only once in the right-hand side of~\eqref{eq:productOfTwoReprExpandsAs}. The claim follows.
\end{proof}

\subsection{Parametrization of laminations by multi-indices on~$\bm{\Ed}$}\label{ReparametrizingLaminations}
Given a collection of punctures~$\lambda_1,\dots,\lambda_n\in\Omega$, we fix a triangulation
$\G=(\{\lambda_1,\dots, \lambda_n,\partial\Omega\}, \Ed, \Fa)$ of~$\Omega$ and an (arbitrary chosen) orientation of its edges~$\Ed$. Sometimes we also use the notation~$\G^\circ$ for the dual graph, which have~$\Fa$ as the set of vertices. Let
\begin{equation*}
\begin{array}{r}
  \L := \big\{\mathbf{n}=(n_e)_{e\in\Ed}\in \mathbb Z_{\geq 0}^{\Ed}\,\mid\, \text{$({n}_{e_1}, {n}_{e_2}, {n}_{e_3})$ satisfy the lamination condition \eqref{eq:lamCon}}\phantom{\big\}} \\ \text{for each $\sigma\in \Fa$ with $\partial \sigma = \{e_1,e_2,e_3\}$}\big\}.
\end{array}
\end{equation*}

Following~\cite{FG}, one can construct a bijection between $\L$ and the set of laminations as follows. Given $\mathbf n\in \L$ and an edge $e\in\Ed$, let $\Gamma_{e}$ be an arbitrary collection of $n_e$ distinct points lying on the interior of $e$. Consider a face $\sigma\in \Fa$ with $\partial \sigma = \{e_1,e_2,e_3\}$. The condition \eqref{eq:lamCon} for ($n_{e_1}, n_{e_2}, n_{e_3}$) ensures that we can draw $\frac{1}{2}(n_{e_1} + n_{e_2} + n_{e_3})$ non-intersecting chords inside $\sigma$, each of them starting at some point from $\Gamma_{e_i}$ and ending at some point from $\Gamma_{e_j}$ with $i\neq j$. Moreover, such a drawing is unique up to homotopies. Once all such chords in all faces $\sigma$ are drawn and concatenated in a natural way, one obtains a collection of simple disjoint curves that represents a lamination $\Gamma=\Gamma(\mathbf{n})$, see Figure~\ref{fig:integers_to_laminations}.

Conversely, given a lamination $\Gamma$ there is a unique way to represent each of the curves $\gamma\in \Gamma$ as a closed non-backtracking path on the dual graph $\G^{\circ}$. Set $n_e$ to be the total number of crossing of an edge~$e$ by these paths on~$\G^{\circ}$. It is easy to see that $\mathbf n:= (n_e)_{e\in \Ed}\in \L$ and that the procedure explained above leads to~$\Gamma(\mathbf n)=\Gamma$. Thus, we have constructed a bijection between the set of all laminations and the set~$\L$.
We also obtain a partial order: given two laminations $\Gamma_1, \Gamma_2$ and the corresponding multi-indices $\mathbf n_1,\mathbf n_2\in \L$ we say that
\[
\text{$\Gamma_1\leq \Gamma_2$ if and only if $\mathbf n_1\leq \mathbf n_2$ coordinate-wise.}
\]
We define the {\it complexity} $|\Gamma|$ of a lamination $\Gamma$ as the sum of all coordinates of the corresponding multi-index $\mathbf n$, it is worth noting that this definition depends on the (arbitrary) choice of the triangulation~$\G$ of $\Omega$. Finally, given a lamination $\Gamma$ and a face $\sigma$ we say that $c$ is a {\it chord} of $\Gamma$ in $\sigma$ if $c$ is one of the chords drawn inside $\sigma$ along the procedure explained above, see Figure~\ref{fig:integers_to_laminations}.

\subsection{Peter--Weyl basis}\label{ReparametrizingFunctions} Following~\cite{FG}, in this section we discuss the application of the (algebraic) Peter--Weyl theorem to the space $\F(\X)^{\sl}$. For this purpose we move from $\sl$-representations of~$\pi_1(\Omega\setminus\{\lambda_1,\dots,\lambda_n\})$ to holonomies of flat $\sl$ connections on~$\Omega\setminus\{\lambda_1,\dots,\lambda_n\}$. These connections can be parameterized by collections $( A_e )_{e\in \Ed}$ of matrices assigned to edges~$e\in\Ed$: one can think that there is a copy $\mathbb C^2_\sigma$ of~$\mathbb C^2$ assigned to each of the faces~$\sigma$ of the triangulation and that $A_e:\mathbb C^2_{\el}\to\mathbb C^2_{\er}$ encodes the change of the bases when moving across~$e$ from left to right (recall that we have fixed an orientation of all the edges~$e\in\Ed$ once forever).

This parametrization allows one to study the ring of functions~$\F(\X)$ as a subring of the ring $\F(\sl^{\Ed})$. The action of~$\sl$ on~$X$ corresponds to an action of the bigger group $\sl^{\Fa}$ on $\sl^{\Ed}$, the latter is given by changes of bases in the spaces $\mathbb C^2_\sigma$ assigned to faces of~$\G$. Namely, given~$\mathbf A =(A_e)_{e\in \Ed}\in \sl^{\Ed}$ and~$\mathbf C =(C_{\sigma})_{\sigma\in \Fa}\in \sl^{\Fa}$, one defines
\begin{equation}
  (\mathbf C[\mathbf A])_e\ :=\ C_{\el}^{-1} A_e C_{\er}^{\vphantom{-1}},\qquad \mathbf C[\mathbf A]:=(\mathbf C[\mathbf A])_{e\in\Ed}\in \sl^{\Ed},
  \label{eq:actionRule}
\end{equation}
where $\el$ and $\er$ stand for two faces adjacent to an edge $e$. Quotients under these actions of~$\sl$ and $\sl^{\Ed}$ coincide (see Lemma~\ref{newCoordinates} below) \bl{and so do the rings of invariants:}
\begin{equation}
\label{eq:F(X/SL)=F(SLE/SLF)}
\F(\X)^{\sl} \simeq \F(\sl^{\Ed})^{\sl^\Fa}.
\end{equation}
Then one can use the representation theory of the group~$\sl^{\Ed}$ to analyze the ring of~$\sl^\Fa$-invariant functions on it. More details are given below.

\bl{Let $V^\vee$ be the dual space to $V$.} Classically, given an algebraic group $G$ and a representation $\rho: G\to \End(V)$ one can construct a map $V\otimes V^{\vee}\to \F(G)$ by setting
\begin{equation}
\label{eq:def-peter-weyl-alg}
v\otimes w\ \mapsto\ \lan \rho(\cdot)v, w \ran.
\end{equation}
The following algebraic version of the Peter--Weyl theorem asserts that each algebraic function on~$\sl^{\Ed}$ can be obtained in this way.

\begin{theorem}[{\bf algebraic Peter--Weyl theorem}]\label{algebraicPeterWeyl}
Let $G$ be a reductive linear algebraic group, $\F(G)$ denote the space of algebraic functions on $G$ and $\widehat{G}$ be the set of all irreducible finite-dimensional representations of $G$. Then the mapping~\eqref{eq:def-peter-weyl-alg} defines an isomorphism of $G\times G$-modules:
\begin{equation*}
  \xymatrix{\displaystyle\bigoplus_{V\in \widehat{G}}V\otimes V^{\vee}\ \ar[r]^-\simeq &\ \F(G)}
\end{equation*}
\end{theorem}
\begin{proof}
E.g., see~\cite[Theorem 27.3.9]{TauvelYu}.
\end{proof}

\bl{
  Let $\gamma$ be a path on the dual graph~$\G^{\circ}$ that crosses edges $e_0,e_1,\dots,e_{m-1}\in \Ed$ consequently, and let $\mathbf A\in \sl^{\Ed}$. We denote by $A_{e_i}$ the matrix from the collection $\mathbf A$ assigned to the edge $e_i$. Let $\sigma_0,\sigma_1,\dots,\sigma_m$ be the faces of $\G$ visited by $\gamma$, enumerated so that $e_i$ is adjacent to $\sigma_i$ and $\sigma_{i+1}$.
}
The \emph{holonomy} of~$\mathbf A$ along~$\gamma$ is defined as
\begin{equation}
\label{eq:hol-def}
\h(\mathbf A, \gamma):= A_{e_0}^{\s(\sigma_0, e_0)}A_{e_1}^{\s(\sigma_1, e_1)}\dots A_{e_{m-1}}^{\s(\sigma_{m-1}, e_{m-1})},
\end{equation}
where $\s(\sigma,e)=+1$ if~$\sigma$ lies to the left of~$e$ and~$\s(\sigma,e)=-1$ otherwise (recall that the orientation of all the edges~$e\in\Ed$ is fixed once forever).

By definition,~$\h(\mathbf A, \gamma_1\cdot \gamma_2) = \h(\mathbf A, \gamma_1)\h(\mathbf A, \gamma_2)$ provided~$\gamma_1$ ends at the face where $\gamma_2$ and  begins $\gamma_1\cdot\gamma_2$ stands for the concatenation of~$\gamma_1$ and~$\gamma_2$. As the fundamental groups of the punctured domain~$\Omega\setminus\{\lambda_1,\dots,\lambda_n\}$ and \bl{of the 1-skeleton} of~$\G^\circ$ coincide, we see that the mapping $\h(\mathbf A,\cdot)$ induces a $\sl$-representation of $\pi_1(\Omega\setminus\{\lambda_1,...,\lambda_n\})$. In particular we obtain an algebraic mapping
\begin{equation}
  \phi: \sl^{\Ed}\to \X,\qquad \mathbf A\mapsto \h(\mathbf A,\cdot).
  \label{eq:triangularMap}
\end{equation}
Recall that the action of the group~$\sl^{\Fa}$ on~$\sl^{\Ed}$ is given by~\eqref{eq:actionRule}.

\begin{lemma}\label{newCoordinates}
The mapping~\eqref{eq:triangularMap} intertwines the action of $\sl^{\Fa}$ on $\sl^{\Ed}$ and the action of $\sl$ on $X$. The induced mapping
\begin{equation*}
  \phi^*: \F(X)^{\sl}\to \F(\sl^{\Ed})^{\sl^{\Fa}}
\end{equation*}
is an isomorphism.
\end{lemma}

\begin{proof}
  The fact that $\phi$ intertwines two actions is straightforward. Let $T$ be a spanning tree for the triangulation $\G$ and a subvariety $Y\subset \sl^{\Ed}$ be defined as
  \begin{equation*}
	Y := \left\{ \mathbf A=(A_e)_{e\in \Ed} \,\mid\, A_e = \Id \text{ if $e\notin T$}\right\}.
  \end{equation*}
  It is easy to see that $\phi$ restricted to $Y$ is an isomorphism, so let $\psi: \X\to Y$ be the inverse mapping. Further, for each $\mathbf A\in \sl^{\Ed}$ there exists $\mathbf C\in \sl^{\Fa}$ such that $\mathbf C[\mathbf A]\in Y$. Thus the restriction mapping $\F(\sl^{\Ed})\to \F(Y)$ sends the subring of~$\sl^{\Fa}$-invariant functions on~$\sl^{\Ed}$ isomorphically onto its image, which we denote by~$\F(\sl^{\Ed})^{\sl^{\Fa}}|_Y$. Hence, the composition
  \begin{equation*}
	\xymatrix{\F(\sl^{\Ed})^{\sl^{\Fa}}\ \ar[r]^-{\simeq} &\ \F(\sl^{\Ed})^{\sl^{\Fa}}|_Y\ \ar[r]^-{\,\psi^*} &\ \F(X)^{\sl}}
  \end{equation*}
  provides an inverse mapping to $\phi^*$.
\end{proof}

Recall that we denote by~$V_1$ the standard two-dimensional representation of $\sl$ and that all finite-dimensional irreducible representations of~$\sl$ are provided by $V_n\simeq\Sym V_1^{\otimes n}$. Let~$x,y$ be the standard basis of~$V_1$ so that $x\otimes x,\frac{1}{2}(x\otimes y+y\otimes x),y\otimes y$ is a basis of~$V_2$ and, more generally, $V_d$ is spanned by the vectors
\begin{equation*}
  (x^ky^{n-k})^{\sym} := (n!)^{-1} \sum\limits_{\pi \in \Sym(d)}\pi_* \biggl(\underbrace{x\otimes \dots\otimes x}_{\text{$k$ times}}\otimes\underbrace{y\otimes \dots\otimes y}_{\text{$n-k$ times}}\biggr),\quad 0\le k\le n.
\end{equation*}

We now define two pairings on $V_n$. To get the first we identify $V_1\wedge V_1$ with $\mathbb C$ by setting $x\wedge y = 1$. This gives rise to a skew-symmetric bilinear form $(v, w)\mapsto v\wedge w$ on~$V_1$. Once can extend it on $V_1^{\otimes n}$ and then restrict to $V_n$. The result is a non-degenerate (skew-symmetric if~$n$ is odd and symmetric if~$n$ is even) bilinear form on~$V_n$, which we still denote by $v \wedge w$.

The second pairing on~$V_n$ is obtained in a similar manner from the Hermitian scalar product on~$V_1$ (which is defined by $\lan x,y \ran = 0$ and $\lan x,x \ran = \lan y,y \ran = 1$ on $V_1$), extended to $V^{\otimes n}$ and then restricted to~$V_n$. One can easily see that
\begin{equation*}
  (-1)^{n-k}(x^ky^{n-k})^{\sym}\wedge (x^{n-k}y^k)^{\sym} = \lan (x^ky^{n-k})^{\sym}, (x^ky^{n-k})^{\sym} \ran = \binom{n}{k}^{-1}.
\end{equation*}

We need an additional notation. For a face $\sigma\in \Fa$ and an edge $e\in \partial \sigma$ adjacent to it, let $V_{\sigma, e, n}$ be a copy of $V_n$. Given a multi-index $\mathbf n\in \mathbb Z_{\geq 0}^{\Ed}$ we introduce a space
\begin{equation*}
  \mathbf{V}_{\mathbf n}\ :=\ \bigotimes_{e\in \Ed} V_{\el(e), e, n_e}\otimes V_{\er(e), e, n_e},
\end{equation*}
recall that $\el(e)$ is the face lying to the left of $e$ and $\er(e)$ is the one to the right, with respect to the once forever fixed orientation of~$e$. There is a linear mapping $\pair: \mathbf{V}_{\mathbf n}\to \mathbb C$ given by
\begin{equation*}
  \textstyle \pair\bigl ({ \bigotimes_{e\in \Ed}} (v_{\el(e), e, n_e}\otimes v_{\er(e), e, n_e})\bigr)\ :=\ \prod_{e\in \Ed} (v_{\el(e), e, n_e}\wedge v_{\er(e), e, n_e}).
\end{equation*}
Define an action of the group $\sl^{\Ed}$ on $\mathbf{V}_{\mathbf n}$ as follows:
\begin{equation}
\label{eq:actionAdef}
  \textstyle \mathbf A\bigl[\,{\bigotimes_{e\in \Ed}} (v_{\el(e), e, n_e}\otimes v_{\er(e), e, n_e})\bigr]\ :=\ {\textstyle \bigotimes_{e\in \Ed}} (v_{\el(e), e, n_e}\otimes A_ev_{\er(e), e, n_e}),
\end{equation}
where~$\mathbf A=(A_e)_{e\in\Ed}$. One can now construct a linear mapping
\begin{equation}
\label{eq:morePeterWeylRepresentation}
\begin{array}{rcl}
\displaystyle \Upsilon:\ \mathbf{V} := \bigoplus_{\mathbf n\in \mathbb Z_{\geq 0}^{\Ed}}\mathbf{V}_{\mathbf n} & \to & \F(\sl^{\Ed}),\\[-12pt]
\bl{\mathbf v} &\bl{\mapsto}& \Bigl( \Upsilon(\mathbf v):  \mathbf A \mapsto \pair(\mathbf A[\mathbf v])\Bigr).
\end{array}
\end{equation}
Applying Theorem~\ref{algebraicPeterWeyl} on each of the edges of~$\Ed$, one sees that~\eqref{eq:morePeterWeylRepresentation} is an isomorphism.

In order to study the subring of~$\sl^\Fa$-invariant functions on~$\sl^\Ed$, we rearrange the factors in the definition of~$\mathbf{V}_{\mathbf n}$ and define an action of $\sl^{\Fa}$ on $\mathbf{V}_{\mathbf n}$ as
\begin{equation*}
\mathbf{V}_{\mathbf n} = \bigotimes_{\sigma\in \Fa}\bigotimes_{e\in \partial \sigma}V_{\sigma, e, n_e},\qquad \textstyle
   \mathbf C\bigl[\,\bigotimes_{\sigma\in \Fa}\bigotimes_{e\in \partial \sigma}v_{\sigma, e, n_e} \bigr]\ :=\ \bigotimes_{\sigma\in \Fa}\bigotimes_{e\in \partial \sigma}C_{\sigma}v_{\sigma, e, n_e},
\end{equation*}
where $\mathbf C=(C_\sigma)_{\sigma\in\Fa}$.

\begin{lemma}\label{wedgeTrick}
The mapping~\eqref{eq:morePeterWeylRepresentation} commutes with the action of~$\sl^{\Fa}$.
\end{lemma}
\begin{proof}
  Let $\mathbf v=\bigotimes_{e\in \Ed} v_{\el(e), e, n_e}\otimes v_{\er(e), e, n_e}\in\mathbf V_{\mathbf n}$ and $\Upsilon(\mathbf v)$ be the corresponding algebraic function on $\sl^{\Ed}$. Note that for each $v,w\in V_d$ and $C\in \sl$ one has $Cv\wedge w = v\wedge C^{-1}w$. Using this observation and~\eqref{eq:actionRule} one gets
  \begin{equation*}
	\begin{split}
	  (\Upsilon(\mathbf v))(\mathbf C[\mathbf A])
    & \textstyle = \pair\bigl( \bigotimes_{e\in \Ed} (v_{\el(e), e, n_e}\otimes (\mathbf C[\mathbf A])_ev_{\er(e), e, n_e})\bigr)\\
    &= \textstyle \prod_{e\in \Ed}\bigl( v_{\el(e), e, n_e}\wedge C_{\el(e)}^{-1}A_eC_{\er(e)}^{\vphantom{-1}}v_{\er(e), e, n_e} \bigr)\\
	&= \textstyle \prod_{e\in \Ed}\bigl( C_{\el(e)}v_{\el(e), e, n_e} \wedge A_eC_{\er(e)}v_{\er(e), e, n_e}\bigr)\\
  &= \pair\bigl(\mathbf A[\mathbf C[\mathbf v]]\bigr) = \Upsilon (\mathbf C[\mathbf v])\bl{(\mathbf A)},
    \end{split}
  \end{equation*}
  thus the actions of~$\sl^{\Fa}$ on~$\sl^{\Ed}$ and~$\mathbf V_{\mathbf n}$ commute with~$\Upsilon$.
\end{proof}

One can now write
\begin{equation*}
\mathbf{V}_{\mathbf n}^{\sl^{\Fa}} = \bigotimes\limits_{\sigma\in \Fa}\biggl(\,\bigotimes\limits_{e\in \partial \sigma}V_{\sigma, e,n_e}\biggr)^{\sl}
\end{equation*}
and apply Corollary~\ref{laminationCondition} to each of the~$\sl$-invariant subspaces corresponding to faces~$\sigma\in \Fa$. It follows that ~$\mathbf{V}_{\mathbf n}^{\sl^{\Fa}}$ is one-dimensional if $\mathbf n\in \L$, i.e. if the lamination condition~\eqref{eq:lamCon} holds true for all~$\sigma$, and is trivial otherwise. Due to Lemma~\ref{newCoordinates} and Lemma~\ref{wedgeTrick}, this leads to the following decomposition of $\F(\X)^{\sl}$ into a direct sum of one-dimensional spaces:
\begin{equation}
  \F(\X)^{\sl} \ \simeq\ \F(\sl^{\Ed})^{\sl^{\Fa}}
  \ \simeq\ \bigoplus_{\mathbf n\in \L}\mathbf{V}_{\mathbf n}^{\sl^{\Fa}}.
  \label{eq:lamDecomp}
\end{equation}

\begin{defin}
\label{def:g_gamma}
Let $g_{\Gamma} \in \F(\X)^{\sl}$ be the projection of $f_{\Gamma}$ onto the subspace corresponding to the space $\mathbf{V}_{\mathbf n(\Gamma)}^{\sl^{\Fa}}$ in the decomposition~\eqref{eq:lamDecomp}, where $\mathbf n(\Gamma)\in \L$ is the multi-index corresponding to~$\Gamma$ as discussed in Section~\ref{ReparametrizingLaminations}. We call functions~\bl{$\{g_\Gamma\}$} the \emph{Peter--Weyl basis} of~$\F(\X)^{\sl}$.
\end{defin}

\begin{rem} Since all the spaces in the decomposition~\eqref{eq:lamDecomp} are one-dimensional, in order to prove that the functions~$g_\Gamma$ indeed form a basis of~$\F(\X)^{\sl}$, it is enough to check that these functions do not vanish. This fact follows from Lemma~\ref{gViaLoops}.
\end{rem}


\subsection{Lamination basis~$\bm{f_\Gamma}$ and Peter--Weyl basis~$\bm{g_\Gamma}$ via each other}\label{fViagAndgViaf} A famous theorem due to Fock and Goncharov~\cite[Theorem 12.3]{FG} claims that the functions~$f_\Gamma$ also form a basis in the space~$\F(\X)$ and that the corresponding change of bases between~$f_\Gamma$ and~$g_\Gamma$ is given by lower-triangular (with respect to the partial order on~$\L$) matrices:
\[
g_\Gamma=\sum_{\Delta\le\Gamma}c_{\Gamma\Delta}f_\Delta,\qquad f_\Gamma=\sum_{\Delta\le\Gamma}\tilde{c}_{\Gamma\Delta}g_\Delta.
\]
The main goal of this section is to set up a framework for the analysis of the coefficients~$c_{\Gamma\Delta}$. While doing this we also repeat the proof of~\cite[Theorem 12.3]{FG}.

Recall that a \emph{multi-curve} is a smooth immersion \mbox{$c:\sqcup_{i = 1}^m S^1\to \Omega\setminus \{\lambda_1,\dots,\lambda_n\}$} of the union of a number of disjoint circles ($m$ is not fixed), considered up to homotopies in the space of \emph{immersions}: in Section~\ref{skein_algebra_section} the total winding (rotation of the tangent vector) of multi-curves will be of importance.

\begin{wrapfigure}{r}{0.42\textwidth}
  \centering
  \tempskip{\includegraphics[width = 0.28\textwidth]{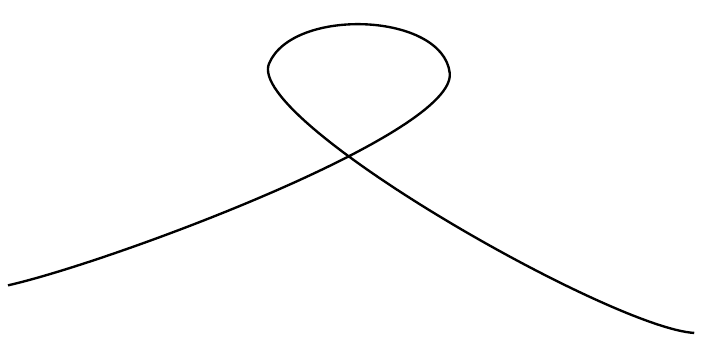}}
  \caption{Nugatory \mbox{self-crossing}.}
  \label{fig:localtwist}
\end{wrapfigure}

In particular, if some component of a multi-curve has a nugatory self-crossing (i.e., a local twist that can be removed by the first Reidemeister move, see~Fig.~\ref{fig:localtwist}), then it \bl{may be not} homotopic to a multi-curve having no such twists. We call a multi-curve \emph{minimal} if it has neither nugatory self-crossings nor homotopically trivial components. \bl{The homotopy class of a  lamination contains a unique multi-curve. Conversely,} not all minimal multi-curves correspond to laminations: only those in which all intersections can be removed do.

Given a minimal multi-curve $\C$ one can view its components as non-backtracking loops on~$\G^{\circ}$ and define $n_e$ to be the total number of intersections between an edge~$e$ and thus obtained loops on~$\G^\circ$. It is easy to see that $\mathbf n=(n_e)_{e\in\Ed}\in \L$ and hence there exists a lamination~$\Gamma(\mathbf n)$; \bl{cf.~Fig.~\ref{fig:n=2}}. We denote~$\Gamma(\C):=\Gamma(\mathbf n)$ and~$|\C| := |\Gamma(\C)|$.

Let~$\gamma_1,...,\gamma_m$ be components of~$\C$. Define
\begin{equation*}
  \textstyle f_{\C}(\rho) := \prod_{k=1}^m\Tr(\rho(\gamma_k)), \qquad f_{\C}\in\F(X)^{\sl}.
\end{equation*}
Our next goal is to study the image of~$f_{\C}$ under isomorphisms~\eqref{eq:lamDecomp}, in particular its image~$\phi^*(f_{\C})\in \F(\sl^{\Ed})$, see Lemma~\ref{newCoordinates}. Fix a multi-index $\mathbf n\in \L$ and let $\Gamma$ be the lamination corresponding to~$\mathbf n$. Minimal multi-curves $\C$ with $\Gamma(\C) = \Gamma$ can be encoded by collections of permutations $\pi\in\prod_{\sigma\in \Fa\,\,e\in \partial \sigma} \Sym(n_e)$ in the following way. Let $\pi_{\sigma,e}\in\Sym(n_e)$ be components of~$\pi$ and $\Gamma_e = \{p_{e,1},\dots, p_{e, n_e}\}$ denote the set of intersection points of $\Gamma$ and a given edge $e\in\Ed$, enumerated from the beginning of~$e$ to its endpoint according to the fixed orientation of~$e$. For each face $\sigma\in \Fa$ and for each pair of points $p_{e_1,k_1}, p_{e_2, k_2}$ connected by a chord of $\Gamma$ draw a simple smooth arc inside $\sigma$ from $p_{e_1, \pi_{\sigma, e_1}(k_1)}$ to $p_{e_2, \pi_{\sigma, e_2}(k_2)}$. Concatenating these arcs at the points~$p_{e,k}$ in a smooth way one obtains a minimal multi-curve which we denote by~$\C(\Gamma, \pi)$.

We need even more notation. Given a multi-index~$\mathbf n\in\L$, a face~$\sigma\in\Fa$, an edge~$e\in\partial\sigma$ adjacent to~$\sigma$, and
a number~$k$ such that $1\leq k\leq n_e$, let $W_{\sigma,e,k}$ be a copy of $V_1$. Introduce the space
\begin{equation}
  \mathbf W_{\mathbf n}\ :=\  \bigotimes_{\sigma\in \Fa}\bigotimes_{e\in\partial\sigma} \bigotimes_{1\leq k\leq n_e} W_{\sigma, e, k}\ =\ \bigotimes_{e\in \Ed}\bigotimes_{1\leq k\leq n_e}\bigl(W_{\el(e), e, k}\otimes W_{\er(e), e, k}\bigr).
  \label{definW}
\end{equation}
Note that the space $\mathbf{V}_{\mathbf n}$ can be realized as a subspace of $\mathbf W_{\mathbf n}$ if one realizes $V_{\sigma, e, n_e}$ as a result of the symmetrization in the space $\bigotimes_{1\leq k\leq n_e}W_{\sigma, e, k}\simeq V_1^{\otimes n_e}$ for each $\sigma\in\Fa$ and $e\in \partial\sigma$. One can extend the pairing~$\pair$ and the action~\eqref{eq:actionAdef} of $\sl^{\Ed}$ from~$\mathbf V_{\mathbf n}$ to~$\mathbf W_{\mathbf n}$ by
\begin{equation*}
  \textstyle \pair\bigl( \bigotimes_{e\in \Ed}\bigotimes_{1\leq k\leq n_e}w_{\el(e), e, k}\otimes w_{\er(e), e, k} \bigr):=\prod_{e\in \Ed}\prod_{1\leq k\leq n_e}(w_{\el(e), e, k}\wedge w_{\er(e), e, k})
\end{equation*}
and
\begin{equation*}
  \textstyle \mathbf A \bigl[ \bigotimes_{e\in \Ed}\bigotimes_{1\leq k\leq n_e}w_{\el(e), e, k}\otimes w_{\er(e), e, k} \bigr]:=\bigoplus_{e\in \Ed}\bigoplus_{1\leq k\leq n_e}(w_{\el(e), e, k}\otimes A_ew_{\er(e), e, k}),
\end{equation*}
respectively. Moreover, one can also lift the mapping~$\Upsilon$ given by~\eqref{eq:morePeterWeylRepresentation} from~$\mathbf V_{\mathbf n}$ to~$\mathbf W_{\mathbf n}$ using the same definition~$\Upsilon(\mathbf w):\mathbf A\mapsto \pair(\mathbf A[\mathbf w])$ for~$\mathbf{A}\in \sl^\Ed$.

We now choose an arbitrary orientation of the components of~$\C(\Gamma, \pi)$ and set
\begin{equation*}
  \s(\C(\Gamma, \pi)) := (-1)^{\#\text{edges of~$\G$ crossed by~$\C(\Gamma,\pi)$ from right to left, counted with multiplicity}}.
\end{equation*}
For a chord $c$ of $\C(\Gamma,\pi)$ running from the $k_1$-th point on $e_1$ to the $k_2$-th point on $e_2$ inside~$\sigma$, denote
\begin{equation*}
  u_{c,\Gamma,\pi} := v^{(1)}_{\sigma, e_1, k_1}\otimes v^{(0)}_{\sigma, e_2,k_2} - v^{(0)}_{\sigma, e_1,k_1}\otimes v^{(1)}_{\sigma, e_2,k_2}
  \in W_{\sigma, e_1, k_1}\otimes W_{\sigma, e_2,k_2},
\end{equation*}
where~$v^{(0)}_{\sigma,e,k}$ and~$v^{(1)}_{\sigma,e,k}$ correspond to the vectors~$x$ and~$y$, respectively, under the identification of~$W_{\sigma,e,k}$ and~$V_1$. Finally, denote
\begin{equation}
\label{eq:w_CGammaPi}
  \textstyle \mathbf w_{\C(\Gamma, \pi)} := \s(\C(\Gamma, \pi))\cdot \bigotimes_{\text{$c$ -- chord of $\C(\Gamma,\pi)$}} u_{c,\Gamma,\pi}  \in \mathbf W_{\mathbf n}.
\end{equation}
It is easy to see that $\mathbf w_{\C(\Gamma, \pi)}$ does not depend on the choice of the orientation of $\C(\Gamma, \pi)$: changing the orientation of some of its components one changes the signs of all vectors~$u_{c,\Gamma,\pi}$ corresponding to these components which is compensated by the same number of changes in the signs of crossings of~$\C(\Gamma,\pi)$ and oriented edges~$e\in\Ed$.
\begin{lemma}
\label{FCgamma=}
Let $\mathbf A\in \sl^{\Ed}$ and $\rho=\phi(\mathbf A)\in\X$ be the corresponding $\sl$-representation of $\pi_1(\Omega\setminus\{\lambda_1,\dots,\lambda_n\})$, see~\eqref{eq:triangularMap}. Then,
\begin{equation*}
  f_{\C(\Gamma, \pi)}(\rho) = \pair( \mathbf A[\mathbf w_{\C(\Gamma, \pi)}]).
\end{equation*}
\end{lemma}

\begin{proof}
  We begin with the special case when $\C(\Gamma, \pi)$ consists of a single curve $\gamma$. \bl{This curve can be thought of as a non-backtracking loop on the dual graph~$\G^\circ$, say, crossing edges $e_0,\dots,e_{m-1},e_m=e_0$ of $\G$ consequently. Recall that in this case one has
\[
f_{\C(\Gamma, \pi)}(\rho)=\Tr(\h(\mathbf A,\gamma))= \Tr(A_{e_0}^{s_0}A_{e_1}^{s_1}\dots A_{e_{m-1}}^{s_{m-1}}),
\]
where the sign~$s_i$ equals to~$+1$ if~$\gamma$ crosses the edge~$e_i$ from left to right and to~$-1$ otherwise.
Let $\sigma_0,\dots, \sigma_m$, $\sigma_m=\sigma_0$, be the sequence of faces of $\G$ corresponding to $\gamma$, so that $e_i$ is adjacent to $\sigma_i$ and $\sigma_{i+1}$. Expanding each of the vectors~$u_{c,\Gamma,\pi}$ one obtains}
\[
\begin{split}
\mathbf w_{\C(\Gamma, \pi)} =\ & \textstyle  \s(\C(\Gamma, \pi))\cdot \sum_{\mathbf j\in\{0,1\}^m} \bigotimes_{i=1}^{m}\bigl((-1)^{1-j_i}v^{(j_i)}_{\sigma_i,e_{i-1},k_{i-1}}\otimes v^{(1-j_i)}_{\sigma_i,e_i,k_i}\bigr)\\
=\ & \textstyle \s(\C(\Gamma, \pi))\cdot\sum_{\mathbf j\in\{0,1\}^m} \biggl(\prod_{i=0}^{m-1}(-1)^{1-j_i}\cdot \bigotimes_{i=0}^{m-1}\bigl(v^{(1-j_i)}_{\sigma_i,e_i,k_i}\otimes v^{(j_{i+1})}_{\sigma_{i+1},e_i,k_i}\bigr)\biggr).
\end{split}
\]
It is easy to see that the contribution of~$v^{(1-j_i)}_{\sigma_i,e_i,k_i}\otimes v^{(j_{i+1})}_{\sigma_{i+1},e_i,k_i}$ to~$\pair(\mathbf A[\mathbf w_{\C(\Gamma,\pi)}])$ is
\[
v^{(1-j_i)}_{\sigma_i,e_i,k_i}\wedge A_{e_i}v^{(j_{i+1})}_{\sigma_{i+1},e_i,k_i} = (-1)^{1-j_i}(A_{e_i})_{j_i,j_{i+1}}
\]
provided that~$\gamma$ crosses~$e_i$ from left to right. Otherwise, this contribution is given by
\[
v^{(j_{i+1})}_{\sigma_{i+1},e_i,k_i}\wedge
A_{e_i}v^{(1-j_i)}_{\sigma_i,e_i,k_i} = (-1)^{j_{i+1}}(A_{e_i})_{1-j_{i+1},1-j_i} = (-1)^{j_i}(A_{e_i}^{-1})_{j_i,j_{i+1}}.
\]
Therefore, all the signs cancel out and one gets
\[
\textstyle \pair(\mathbf A[\mathbf w_{\C(\Gamma,\pi)}])= \sum_{\mathbf j\in\{0,1\}^m} (A_{e_i}^{s_i})_{j_i,j_{i+1}} = \Tr(A_{e_0}^{s_0}A_{e_1}^{s_1}\dots A_{e_{m-1}}^{s_{m-1}}).
\]

In the general case one simply repeats the same computation for each of the components~$\gamma_1,\dots,\gamma_m$ of a multi-curve~$\C(\Gamma,\pi)$.
\end{proof}

We now apply \bl{Decomposition}~\eqref{eq:V_1^n-decomposition} for each of the edges~$e\in\Ed$ and obtain
\[
\mathbf W_{\mathbf n}\ \simeq\ \bigotimes_{e\in\Ed}\biggl(\biggl(\,\bigoplus_{0\le m_e\le n_e}V_{\el(e), e, m_e}^{\oplus l(n_e,m_e)} \biggr){\textstyle \bigotimes} \biggl(\,\bigoplus_{0\le m_e\le n_e}V_{\er(e), e, m_e}^{\oplus l(n_e,m_e)} \biggr) \biggr),
\]
recall that~$l(n,n)=1$ and~$l(m,n)=0$ if~$n-m$ is odd. Let $\mathbf l(\mathbf n,\mathbf m):=\prod_{e\in\Ed}(l(n_e,m_e))^2$, where~$\mathbf m=(m_e)_{e\in\Ed}$. Rearranging factors we get
\[
\begin{split}
\textstyle \mathbf W_{\mathbf n}\ \simeq\ \bigl(\bigoplus_{\mathbf 0\le \mathbf m\le \mathbf n}\mathbf V_{\mathbf m}^{\oplus \mathbf l(\mathbf n,\mathbf m)}\bigr)\,\oplus\, & \mathbf W_{\mathbf n}^0,\\
\text{where}\ & \textstyle \mathbf W_{\mathbf n}^0\simeq\bigoplus_{\mathbf 0\le \mathbf m,\mathbf m'\le \mathbf n\,:\, \mathbf m\ne \mathbf m'} \bigl(V_{\el(e), e,m_e}^{\oplus l(n_e,m_e)} \otimes V_{\er(e), e, m'_e}^{\oplus l(n_e,m'_e)}\bigr).
\end{split}
\]
\begin{lemma} The following diagram commutes:
\begin{equation}
  \begin{gathered}\xymatrix{ \mathbf W_{\mathbf n} \ar[d]^-\Upsilon && \bigl(\bigoplus_{\mathbf 0\le \mathbf m\le \mathbf n}\mathbf V_{\mathbf m}^{\oplus \mathbf l(\mathbf n,\mathbf m)}\bigr)\oplus\mathbf W_{\mathbf n}^0 \ar[ll]_-\simeq \ar[d]   \\  \F(\sl^{\Ed}) && \mathbf \bigoplus_{\mathbf m\in  \mathbb Z_{\geq 0}^{\Ed}}\mathbf V_{\mathbf m}\ =\ \mathbf V, \ar[ll]_-\simeq^-{\text{\eqref{eq:morePeterWeylRepresentation}}} }\end{gathered}
  \label{eq:finalDecomp}
\end{equation}
where~$\Upsilon$ is defined by~$\Upsilon(\mathbf w):\mathbf A\mapsto \pair(\mathbf A [\mathbf w])$ while the map on right vanishes on~$\mathbf W_{\mathbf n}^0$ and is defined on each of the components \mbox{$\mathbf V_{\mathbf m}^{\oplus \mathbf l(\mathbf n,\mathbf m)}$} as~$\bigoplus_{i=1}^{\mathbf l(\mathbf n, \mathbf m)}\mathbf v_i \mapsto \sum_{i=1}^{\mathbf l(\mathbf n,\mathbf m)}\mathbf v_i$.
\end{lemma}
\begin{proof}
The only non-trivial ingredient is to check that one has~$\Upsilon(\mathbf w)=0$ for~$\mathbf w\in\mathbf W_{\mathbf n}^0$. This immediately follows from the fact that~$v\wedge v' =0$ if $v\in V_m$, $v'\in V_{m'}$ and~$m\ne m'$, where both~$V_m,V_{m'}$ are identified with subspaces of~$V_1^{\otimes n}$ via decomposition~\eqref{eq:V_1^n-decomposition}. Clearly, it is enough to prove this fact for elements~$v_k\in V_m$,~$v_{k'}\in V_{m'}$ of the bases introduced in Proposition~\ref{sl2basic}(i). Let~$m<m'$ and assume that~$k+k'\le m$. Then~$e^{k+1}v_k=0$ while~$f^{k+1}v'_{k'}\ne 0$ and one can write
\[
{\textstyle\binom{m'-k'}{k+1}\binom{k'+k+1}{k+1}}\cdot v\wedge v' = v\wedge e^{k+1}f^{k+1}v' = (-1)^{k+1}e^{k+1}v\wedge f^{k+1}v' =0
\]
since~$u\wedge eu' = -eu\wedge u'$ for each~$u,u'\in V_1^{\otimes n}$. The case~$k+k'\ge m+1$ can be handled in the same way starting with~$f^{m-k+1}v_k=0$ and~$e^{m-k+1}v'_{k'}\ne 0$.
\end{proof}

\begin{lemma}\label{gViaLoops} Let~$\mathbf n\in \L$, $\Gamma=\Gamma(\mathbf n)$ be the corresponding lamination, and~$g_\Gamma$ be given by Definition~\ref{def:g_gamma}. The following identity holds:
\[
g_{\Gamma} = \prod\nolimits_{\sigma\in \Fa,\,e\in \partial \sigma}(n_e!)^{-1}\cdot \sum\nolimits_{\,\pi\in \prod_{\sigma\in \Fa,\,e\in\partial\sigma}\Sym(n_e)} f_{\C(\Gamma, \pi)}.
\]
In particular,~$g_{\Gamma}(\Id)>0$ and hence~$g_\Gamma$ does not vanish.
\end{lemma}

\begin{proof} Recall that the~$\sl^{\Fa}$-invariant function~$g_\Gamma$ is defined as the component of~$f_\Gamma$ lying in the subspace~$\mathbf{V}_{\mathbf{n}}$ under the isomorphism~\eqref{eq:morePeterWeylRepresentation}.
Lemma~\ref{FCgamma=} gives
\[
f_\Gamma=f_{\C(\Gamma,\Id)}=\Upsilon(\mathbf w_{\C(\Gamma,\Id)}).
\]
Note that the map on the right of~\eqref{eq:finalDecomp} acts identically on~$\mathbf V_{\mathbf n}$ and, under the isomorphism on top of~\eqref{eq:finalDecomp}, each of the factors~$V_{\sigma,e,n_e}\simeq V_{n_e}$ of~$\mathbf{V}_{\mathbf{n}}$ is obtained as a symmetrization in the corresponding factor~$\bigotimes_{k=1}^{n_e} W_{\sigma,e,k}\simeq V_1^{\otimes n_e}$ of~$\mathbf {W}_{\mathbf n}$. The fact that this diagram commutes yields
\begin{equation}
\label{eq:g=Upsilon()}
g_\Gamma= \prod\nolimits_{\sigma\in \Fa,\,e\in\partial\sigma}(n_e!)^{-1}\cdot \sum\nolimits_{\,\pi\in\prod_{\sigma\in \Fa,\,e\in \partial \sigma} \Sym(n_e)} \Upsilon(\pi_*\mathbf w_{\C(\Gamma,\Id)}).
\end{equation}
The claim follows since~$\Upsilon(\pi_*\mathbf w_{\C(\Gamma,\Id)})=\Upsilon(\mathbf w_{\C(\Gamma,\pi)})=f_{\C(\Gamma,\pi)}$ for each~$\pi$.
\end{proof}

\begin{rem} Since all the maps involved into~\eqref{eq:finalDecomp} commute with the action of $\sl^{\Fa}$, we obtain a similar commutative diagram for the invariant subspaces:
\begin{equation}
  \begin{gathered}\xymatrix{ \mathbf W_{\mathbf n}^{\sl^{\Fa}} \ar[d] & \biggl(\bigoplus_{\substack{\mathbf m\in \L\\ \mathbf 0\le \mathbf m\le \mathbf n}} \bigl(\mathbf V_{\mathbf m}^{\sl^{\Fa}}\bigr)^{\oplus \mathbf l(\mathbf n,\mathbf m)}\biggr) \oplus \bigl(\mathbf W_{\mathbf n}^0\bigr)^{\sl^{\Fa}} \ar[l]_-\simeq \ar[d]   \\  \F(\sl^{\Ed})^{\sl^{\Fa}} & \bigoplus_{\mathbf m\in \L}\mathbf{V}_{\mathbf m}^{\sl^{\Fa}}= \mathbf{V}^{\sl^{\Fa}}.\ar[l]_-\simeq^-{\text{\eqref{eq:lamDecomp}}}} \end{gathered}
  \label{eq:finalInvDecomp}
\end{equation}
\end{rem}

Since each of the spaces~$\mathbf V_{\mathbf m}^{\sl^{\Fa}}$,~$\mathbf m\in\L$, is one-dimensional, the functions~$g_\Gamma$ form a basis in the space \mbox{$\F(X)^{\sl}\!\simeq\F(\sl^{\Ed})^{\sl^{\Fa}}\!\simeq \mathbf V_{\mathbf n}$}, which we call the Peter--Weyl basis. We are now in the position to show that the functions~$f_\Gamma$ also form a basis in this space.

\begin{theorem}[{\cite[Theorem 12.3]{FG}}]
\label{thm:FG}
The functions $f_{\Gamma}$ form a basis in $\F(X)^{\sl}$. Moreover, the change of the bases is given by a low-triangular (with respect to the partial order on the set of laminations) matrix:
\begin{equation}
\label{eq:c_Gamma_delta}
g_\Gamma=\sum\nolimits_{\Delta\le\Gamma}c_{\Gamma\Delta}f_\Delta,\quad \text{where}\quad c_{\Gamma\Gamma}=1.
\end{equation}
\end{theorem}

\begin{proof} Since the functions~$g_\Gamma$ form a basis in~$\F(\X)^{\sl}$, one can write a decomposition~$f_\Gamma=\sum_{\Delta}\tilde{c}_{\Gamma\Delta}g_\Delta$, note that~$\tilde{c}_{\Gamma\Gamma}=1$ due to Definition~\ref{def:g_gamma}. The diagram~\eqref{eq:finalInvDecomp} is commutative, hence this decomposition must be the image of a similar decomposition of the~$\sl^{\Fa}$-invariant component of the vector~$\mathbf w_{\C(\Gamma,\Id)}$ under the isomorphism on the top of~\eqref{eq:finalInvDecomp}. Therefore, the coefficient~$\tilde{c}_{\Gamma\Delta}$ vanishes unless~$\Delta\le \Gamma$. Thus the matrix~$(\tilde{c}_{\Gamma\Delta})$ is low-triangular and hence the inverse one is low-triangular as well.
\end{proof}

\begin{rem} It is worth noting that~$c_{\Gamma\Delta}=\tilde{c}_{\Gamma\Delta}=0$ if~$n_e(\Gamma)-n_e(\Delta)$ is odd for some edge~$e\in\Ed$. This easily follows from the fact that~$m_e$ and~$n_e$ should have the same parity in order that~$V_{m_e}$ appears in the decomposition~\eqref{eq:V_1^n-decomposition} of~$V_1^{\otimes n_e}$.
\end{rem}


\subsection{Extension of~$\bm{f_\Gamma,g_\Gamma}$ to~$\bm{(\mathbb C^{2\times 2})^{\Ed}}$ and orthogonality on poly-balls}\label{further}
In this section we discuss natural extensions of functions~$f_\Gamma$ and~$g_\Gamma$ from~$\sl^\Ed$ to~$(\mathbb C^{2\times 2})^\Ed$. We apply the analytic Peter--Weyl theorem for the group~$\su^\Ed$ in order to show that thus obtained extensions~$G_{\Gamma,\mathbf m}$ of~$g_\Gamma$ are orthogonal on poly-balls
\begin{equation}
\label{eq:B_Rdef}
\mathbb B_R:=\mathrm B_R^\Ed \subset (\mathbb C^{2\times 2})^\Ed,\qquad \mathrm B_R:=\{A\in \mathbb C^{2\times 2}\,\mid\, \Tr AA^*< R^2\}
\end{equation}
with respect the Euclidean measure on $(\mathbb C^{2\times 2})^\Ed\simeq (\mathbb C^4)^\Ed$. Then, using the interpretation of~$V_n$ as the spaces of homogeneous polynomials of degree~$n$ we derive an exponential lower bound for the $L^2$-norms of~$G_{\Gamma,\mathbf d}$ on these poly-balls which is required for the further analysis performed in Section~\ref{complex_analisys_section}.

Since~$\sl^\Ed$ is an algebraic subvariety of $(\mathbb C^{2\times 2})^\Ed$ one has a trivial surjection
\[
\F((\mathbb C^{2\times 2})^\Ed)\ \to\ \F(\sl^\Ed).
\]
The isomorphism~\eqref{eq:morePeterWeylRepresentation} provides a way to construct a left inverse
\begin{equation}
  \F(\sl^{\Ed})\ \hookrightarrow\ \F((\mathbb C^{2\times 2})^{\Ed})
  \label{eq:functionsToPolynomials}
\end{equation}
to this surjection. Namely, there is a natural extension of the action~\eqref{eq:actionAdef} of~$\sl^\Ed$ on~$\mathbf V_{\mathbf n}$ to an action of~$(\mathbb C^{2\times 2})^\Ed$: to~define the latter on a factor~$V_{n_e}$ of~$\mathbf{V}_{\mathbf n}$ one simply extends the standard action of~$\sl$ on~$V_1$ to an action of~$\mathbb C^{2\times 2}$ and view~$V_{n_e}$ as the symmetrization of~$V_1^{\otimes n_e}$. Clearly, for each~$\mathbf v\in \mathbf V_{\mathbf n}$, this provides a natural extension of the function~$\Upsilon(\mathbf v):\mathbf A\mapsto \pair(\mathbf A[\mathbf v])$ from~$\sl^\Ed$ to~$(\mathbb C^{2\times 2})^\Ed$.

Note that the action~\eqref{eq:actionRule} of $\sl^{\Fa}$ on $\sl^{\Ed}$ naturally extends to an action on $(\mathbb C^{2\times 2})^{\Ed}$ and that~\eqref{eq:functionsToPolynomials} commutes with this action. Therefore, one has an injection
\begin{equation}
  \iota: \F(\X)^{\sl}\ \simeq\ \F(\sl^{\Ed})^{\sl^{\Fa}}\ \hookrightarrow\ \F((\mathbb C^{2\times 2})^{\Ed})^{\sl^{\Fa}}.
  \label{eq:invariantsToPolynomials}
\end{equation}
Given a lamination~$\Gamma$ and a multi-index~$\mathbf m=(m_e)_{e\in \Ed}$ we define
\begin{equation}
\label{eq:G_Gamma_m-def}
G_\Gamma:=\iota(g_\Gamma),\qquad G_{\Gamma, \mathbf m}(\mathbf A) := G_{\Gamma}(\mathbf A)\cdot \prod\limits_{e\in \Ed}(\det A_e)^{m_e}.
\end{equation}
\begin{rem}
\label{rem:G-homdegree}
By construction, $G_{\Gamma,\mathbf m}$ is a homogeneous polynomial of degree~$|\Gamma|+2|\mathbf m|$ on~$(\mathbb C^{2\times 2})^{\Ed}\simeq \mathbb C^{4|\Ed|}$, invariant under the action of~$\sl^{\Fa}$ on~$(\mathbb C^{2\times 2})^{\Ed}$. More precisely, $G_{\Gamma,\mathbf m}$ is a homogeneous polynomial of degree~\mbox{$n_e+2m_e$} in coordinates of the space~$\mathbb C^{2\times 2}$ assigned to an edge~$e$, for each~$e\in\Ed$. Below we call such polynomials homogeneous of multi-degree~$\mathbf d$.
\end{rem}

\begin{lemma}\label{GisBasis}
For each~$d\in\mathbb Z_{\ge 0}$, polynomials $\{G_{\Gamma, \mathbf m}\,\mid\, |\Gamma| + 2|\mathbf m| = d\}$ form a basis in the space of $\sl^{\Fa}$-invariant homogeneous polynomials on $(\mathbb{C}^{2\times 2})^\Ed$ of degree $d$.
\end{lemma}

\begin{proof}
  Note that the action $\sl^{\Fa}$ respects the degrees of homogeneity in coordinates assigned to an edge~$e$. Therefore, it is enough to show that, given $\mathbf d\in \mathbb Z_{\geq 0}^{\Ed}$, the set $\left\{ G_{\Gamma, \mathbf m}\,\mid\, \mathbf n(\Gamma) + 2\mathbf m = \mathbf d \right\}$ is a basis in the space of homogeneous $\sl^{\Fa}$-invariant polynomials of multi-degree~$\mathbf d$, see Remark~\ref{rem:G-homdegree}.

  Let~$\F^{(\mathbf d)}((\mathbb C^{2\times 2})^{\Ed})$ denote the space of all homogeneous polynomials on~$(\mathbb{C}^{2\times 2})^\Ed$ of multi-degree $\mathbf d$. Let
  \begin{equation}
    \bigoplus_{\mathbf m\,:\, 2\mathbf m \leq \mathbf d} \mathbf V_{\mathbf d - 2\mathbf m}\ \to\ \F^{(\mathbf d)}((\mathbb C^{2\times 2})^{\Ed})
    \label{eq:proofThatGisBasis}
  \end{equation}
  which sends a vector~$\mathbf v\in \mathbf V_{\mathbf d - 2\mathbf m}$ to a polynomial~$\mathbf A \mapsto \pair(\mathbf A[\mathbf v])\cdot\prod_{e\in\Ed}(\det A_e)^{m_e}$.
  Due to the algebraic Peter--Weyl theorem, the composition of this mapping with the restriction of functions from $(\mathbb C^{2\times 2})^{\Ed}$ to $\sl^{\Ed}$ is an injection, hence the mapping~\eqref{eq:proofThatGisBasis} itself is an injection. On the other hand we have
  \begin{equation*}
    \begin{split}
      \textstyle \dim_{\mathbb C} \bigl(\bigoplus_{\mathbf m\,:\, 2\mathbf m \leq \mathbf d} \mathbf V_{\mathbf d - 2\mathbf m}\bigr)\ & \textstyle =\  \sum_{\mathbf m\,:\,2\mathbf m \leq \mathbf d} \prod_{e\in \Ed} (d_e - 2m_e + 1)^2 \\
      & \textstyle=\ \prod_{e\in \Ed}\sum_{m_e\,:\,2m_e \leq \mathbf d_e} (d_e - 2m_e + 1)^2\\
      & \textstyle=\ \prod_{e\in \Ed} ((d_e+1)(d_e + 2)(d_e + 3)/6) \\
      & \textstyle =\ \dim_{\mathbb C} \bigl( \F^{(\mathbf d)}((\mathbb C^{2\times 2})^{\Ed})\bigr)
    \end{split}
  \end{equation*}
  as the dimension of the space of homogeneous polynomials of degree~$d_e$ on~$(\mathbb C)^{2\times 2}\simeq \mathbb C^4$ is equal to~$\binom{d_e+3}{3}$. Since the dimensions coincide, the mapping~\eqref{eq:proofThatGisBasis} is an isomorphism. Using the fact that this mapping commutes with the action of~$\sl^{\Fa}$ we get
  \begin{equation*}
    \bigoplus_{\mathbf m\,:\,2\mathbf m \leq \mathbf d} \mathbf V_{\mathbf d - 2\mathbf m}^{\sl^{\Fa}}\ \simeq \bigoplus_{\substack{\mathbf m\,:\,2\mathbf m \leq \mathbf d \\ \mathbf d - 2\mathbf m\in \L}} \mathbf V_{\mathbf d - 2\mathbf m}^{\sl^{\Fa}}\  \simeq\ \F^{(\mathbf d)}((\mathbb C^{2\times 2})^{\Ed})^{\sl^{\Fa}}.
  \end{equation*}
  Recall that each of the spaces $V_{\mathbf d- 2\mathbf m}^{\sl^{\Fa}}$, where $\mathbf d-2\mathbf m \in \L$, is one-dimensional and the image of its generating vector is the polynomial $G_{\Gamma(\mathbf d- 2\mathbf m), \mathbf m}$. Therefore, these polynomials indeed form a basis in the space~$\F^{(\mathbf d)}((\mathbb C^{2\times 2})^{\Ed})^{\sl^{\Fa}}$.
\end{proof}

We now extend each of the functions~$f_\Gamma$ from~$\sl^\Ed$ to~$(\mathbb C^{2\times 2})^\Ed$. For this purpose we need to give a meaning to the holonomy
\[
\h(\mathbf A,\gamma)=A_{e_0}^{\s(\sigma_0, e_0)}A_{e_1}^{\s(\sigma_1, e_1)}\dots A_{e_{m-1}}^{\s(\sigma_{m-1}, e_{m-1})}
\]
of~$\mathbf A$ along a \bl{loop~$\gamma$ (see~\eqref{eq:hol-def}) for general matrices~$A_e\in\mathbb C^{2\times 2}$, where \mbox{$e_0,e_1,\dots,e_{m-1}\in \Ed$} are the edges crossed by $\gamma$. Recall that each loop $\gamma$ from a lamination $\Gamma$ can be represented by a non-backtracking loop on the $\G^\circ$ in a unique way.

Let~$\h^\vee(\cdot,\gamma)$ be the extension of $\h(\cdot,\gamma)$ on~$(\mathbb C^{2\times 2})^\Ed$ defined as follows: replace each inverse matrix~$A_e^{-1}$ appearing in the above definition by the adjugate matrix~$A_e^\vee$, which is defined by the identity~$A_e^\vee v\wedge w = v\wedge A_ew$; note that~$A_e+A_e^\vee=\Tr A_e\cdot \Id$.}

It is easy to see that~$\Tr\h^\vee(\mathbf A,\gamma)$ does not depend on the orientation of~$\gamma$. \bl{Moreover, if $\gamma$ is a loop, then $\Tr\h^\vee(\mathbf A,\gamma)$ also does not depend on} a choice of the starting point \bl{of $\gamma$. This} allows one to extend the functions~$f_\Gamma$ from~$\sl^\Ed$ to~$(\mathbb C^{2\times 2})^\Ed$ as
\begin{equation}
\label{eq:F_Gamma_m-def}
F_\Gamma(\mathbf A):=\prod_{\gamma\in\Gamma}\Tr(\h^\vee(\mathbf A,\gamma)), \qquad F_{\Gamma, \mathbf m}(\mathbf A) := F_{\Gamma}(\mathbf A)\cdot \prod\limits_{e\in \Ed}(\det A_e)^{m_e},
\end{equation}
note that all functions~$F_{\Gamma,\mathbf m}$, $\mathbf m=(m_e)_{e\in\Ed}\in\mathbb Z_{\ge 0}^\Ed$, coincide with~$f_\Gamma$ on~$\sl^\Ed$.

\begin{lemma}
\label{lem:GGmViaFDn}
For each lamination~$\Gamma$ and~$\mathbf m\in\mathbb Z_{\ge 0}^\Ed$, the following identity holds:
\[
\textstyle G_{\Gamma,\mathbf m}\ =\ \sum_{\Delta\le\Gamma}c_{\Gamma\Delta}F_{\Delta,\mathbf m+\frac{1}{2}(\mathbf n(\Gamma)-\mathbf n(\Delta))},
\]
where the coefficients~$c_{\Gamma\Delta}$ are the same as in Theorem~\ref{thm:FG}.
\end{lemma}
\begin{proof}
Let~$f_\Gamma=\sum_{\Delta\le\Gamma}\tilde{c}_{\Gamma\Delta}g_\Delta$, recall that~$\tilde{c}_{\Gamma\Delta}=0$ if~$n_e(\Gamma)-n_e(\Delta)$ is odd for some~$e\in\Ed$. It is easy to see that
\[
\textstyle F_{\Gamma,\mathbf m}\ =\ \sum_{\Delta\le\Gamma}\tilde{c}_{\Gamma\Delta}G_{\Delta,\mathbf m+\frac{1}{2}(\mathbf n(\Gamma)-\mathbf n(\Delta))}.
\]
Indeed, since this identity holds on~$\sl^\Ed$ and both sides are homogeneous polynomials of multi-degree~$\mathbf n+2\mathbf m$, this also holds on the open subset~$\gl^\Ed$ of~$(\mathbb C^{2\times 2})^\Ed$. The claim follows since the matrices~$(c_{\Gamma\Delta})_{\Delta\le\Gamma}$ and~$(\tilde{c}_{\Gamma\Delta})_{\Gamma\le\Delta}$ are inverse to each other.
\end{proof}

We now move to the analysis of functions~$G_{\Gamma,\mathbf m}$ as elements of the Hilbert space $L^2(\mathbb B_R)$ 
on poly-balls~$\mathbb B_R=\mathrm B_R^\Ed$, see~\eqref{eq:B_Rdef}. Recall that the Euclidean measure on~$\mathbb B_R$ is given by the scalar product defined on each of the components as
\begin{equation}
\label{eq:scalarProduct}
\lan A_1, A_2\ran_{\mathbb C^{2\times 2}}=\Tr(A_1A_2^*).
\end{equation}

Below we use the analytic Peter--Weyl theorem, applied to the group $G=\su^\Ed$.

\begin{theorem}[{\bf analytic Peter--Weyl theorem}]\label{analyticPeterWeyl}
Let $G$ be a compact Lie group and $\widehat{G}$ denote the set of all its irreducible finite-dimensional unitary representations. For each representation $(V,\rho)\in \widehat G$ let $v^V_1,\dots,v^V_{\dim V}$ be the orthonormal basis in $V$ and
\begin{equation*}
  \bl{u^V_{i,j}(g)} := \sqrt{\dim V}\cdot \lan \rho(g)v_i^V, v_j^V \ran.
\end{equation*}
Then, the set~$\{u_{i,j}^V\,\mid\, V\in \widehat G,\ 1\leq i,j\leq \dim V\}$ is an orthonormal basis in~$L^2(G)$.
\end{theorem}
\begin{proof}
E.g., see \cite[Chapter~5]{Robert}.
\end{proof}

Classically, irreducible unitary representations of~$\su$ are given by the restrictions of irreducible representations~$\sl\to\End(V_n)$ provided that the scalar product on~$V_n\simeq \Sym V_1^{\otimes n}$ is obtained from the standard scalar product on~$V_1$. Let~$\Gamma_1\ne\Gamma_2$ be two distinct laminations and
\[
\textstyle \mathbf u_1\in \bigotimes_{\sigma\in \Fa}\bigotimes_{e\in \partial\sigma}V_{n_e(\Gamma_1)},\qquad \mathbf u_2\in \bigotimes_{\sigma\in \Fa}\bigotimes_{e\in \partial\sigma}V_{n_e(\Gamma_2)}.
\]
Since~$n_e(\Gamma_1)\ne n_e(\Gamma_2)$ for some edge~$e\in\Ed$, it follows from Theorem~{\ref{analyticPeterWeyl}} that
\begin{equation}
\label{eq:ort_U[u]}
\int_{\su^\Ed} \pair( \mathbf U[\mathbf u_1])\,\overline{\pair(\mathbf U [\mathbf u_2])}\, \mu_{\mathrm{Haar}}(\mathbf d\mathbf U)\ =\ 0,
\end{equation}
where~$\mathbf U=(U_e)_{e\in\Ed}\in \su^\Ed$ and~$\mu_{\mathrm{Haar}}(\mathbf d\mathbf U)=\prod_{e\in E}\mu_{\mathrm{Haar}}(dU_e)$. The next step is to deduce the following lemma from the orthogonality condition~\eqref{eq:ort_U[u]}.

\begin{lemma}
\label{GareOrt}
For each~$R>0$, the polynomials $G_{\Gamma, \mathbf m}$ are orthogonal in~$L^2(\mathbb B_R)$:
\[
\lan G_{\Gamma_1,\mathbf m_1},G_{\Gamma_2,\mathbf m_2}\ran_{L^2(\mathbb B_R)} = 0\quad \text{if}\ \ (\Gamma_1,\mathbf m_1)\ne(\Gamma_2,\mathbf m_2).
\]
\end{lemma}

\begin{proof} Recall that~$G_{\Gamma,\mathbf {m}}$ is a homogeneous polynomial of degree~$|\Gamma|+2|\mathbf m|$. Therefore, one can assume~$R=1$ without loss of generality. Let~$\mathbf{n}(\Gamma_1)+2\mathbf m_1\ne\mathbf{n}(\Gamma_2)+2\mathbf m_2$. Since the ball~$\mathbb B_1$ is invariant under rotations~$\mathbf A =(A_e)_{e\in \Ed}\mapsto (e^{2\pi it_e}A_e)_{e\in \Ed}=:e^{2\pi i\mathbf t}\mathbf A$, in this case one easily gets
\begin{equation*}
  \begin{split}
    \lan G_{\Gamma_1, \mathbf m_1}, G_{\Gamma_2, \mathbf m_2} \ran_{L^2(\mathbb B_1)}\, &=\, \int_{\mathbf t\in [0,1]^\Ed} \mathbf{dt}\int_{\mathbb B_1} G_{\Gamma_1, \mathbf m_1}(e^{2\pi i \mathbf t}\mathbf A)\,\overline{G_{\Gamma_2, \mathbf m_2}(e^{2\pi i \mathbf t}\mathbf A)}\,\lambda_{\mathbb B_1}(\mathbf{dA})\\
    &=\, \int_{\mathbf t\in[0,1]^\Ed} e^{2\pi i \mathbf t\cdot ((\mathbf n(\Gamma_1)+2\mathbf m_1)-(\mathbf n(\Gamma_2)+2\mathbf m_2))}\mathbf{dt}\cdot \lan G_{\Gamma_1, \mathbf m_1}, G_{\Gamma_2, \mathbf m_2} \ran_{L^2(\mathbb B_1)}\\
    &=\, 0.
  \end{split}
\end{equation*}

Assume now that~$\mathbf{n}(\Gamma_1)+2\mathbf m_1=\mathbf{n}(\Gamma_2)+2\mathbf m_2$, hence~$\Gamma_1\ne \Gamma_2$ as otherwise one would have~$(\Gamma_1,\mathbf m_1)=(\Gamma_2,\mathbf m_2)$. Let~$\mathrm B_1^+:=\{H\in \mathrm B_1\,\mid\,H\ge 0\}$ be the set of non-negative Hermitian matrices $H\in\mathbb C^{2\times 2}$ satisfying~$\Tr(HH^*)<1$.
Consider the mapping
\begin{equation}
  A:[0,1]\times \su\times \mathrm B_1^+\to  \mathrm B_1,\qquad A(\theta,U,H):=e^{\pi i\theta}UH.
  \label{eq:fubini}
\end{equation}
This mapping is a bijection modulo zero measure sets due to the existence and the uniqueness of the polar decomposition of generic matrices \mbox{$A\in \mathrm B_1$}. Note that the scalar product~\eqref{eq:scalarProduct}
satisfies \bl{the identity
\[
  \lan A(\theta,U,H),A(\theta,U,H)\ran_{\mathbb C^{2\times 2}}~=~ \Tr(H^2).
\]
Therefore, the Euclidean volume $\lambda_{\mathrm B_1}(dA(\theta,U,H))$ on $\mathrm B_1$ can be factorized as $\nu(d\theta,dU)\cdot \nu(dH)$, where $\nu(dH)$ is an absolutely continuous measure on $\mathrm B_1^+$ while a probability measure $\nu(d\theta,dU)$ is invariant under translations $\theta\mapsto\theta+\theta_0$ and $U\mapsto U_0U$. Factorizing $\nu(d\theta,dU)$ as the product of two Haar measures we arrive at the factorization
\begin{equation}
\label{eq:productOfMeasures}
\lambda_{\mathrm B_1}(dA(\theta,U,H))=d\theta\cdot\mu_{\mathrm{Haar}}(dU)\cdot\nu(dH),
\end{equation}
for a certain absolutely continuous measure $\nu$ on $\mathrm B_1^+$.}

Let~$\mathbf \mathbf v_1\in \bigotimes_{\sigma\in \Fa}\bigotimes_{e\in \partial\sigma}V_{n_e(\Gamma_1)}$ and~$v_2\in \bigotimes_{\sigma\in \Fa}\bigotimes_{e\in \partial\sigma}V_{n_e(\Gamma_2)}$ be such that
\[
G_{\Gamma_1}(\mathbf A)=\pair(\mathbf A[\mathbf v_1])\quad \text{and}\quad G_{\Gamma_2}(\mathbf A)=\pair(\mathbf A[\mathbf v_2])
\]
for~$\mathbf A\in(\mathbb C^{2\times 2})^\Ed$. Denote~$\mathbf{H}=(H_e)_{e\in\Ed}\in \mathbb B_1^+:=(\mathrm B_1^+)^\Ed$ and let~$\nu(\mathbf{dH}):=\prod_{e\in \Ed}\nu(dH_e)$. Since~$\Gamma_1\ne\Gamma_2$ one gets
\begin{equation*}
  \begin{split}
    \lan G_{\Gamma_1, \mathbf m_1},\right.&\left. G_{\Gamma_2, \mathbf m_2} \ran_{L^2(\mathbb B_1)}\\
    \,&=\, \int_{\mathbb B_1^+}\nu(\mathbf{dH}) \det \mathbf H^{\mathbf m_1+\mathbf m_2}
    \int_{\su^\Ed}\pair(\mathbf{UH}[\mathbf v_1])\,\bl{\overline{\pair(\mathbf{UH}[\mathbf v_2])}}\,\mu_{\mathrm{Haar}}(\mathbf{dU}),
  \end{split}
\end{equation*}
where~$\det\mathbf H^{\mathbf m}:=\prod_{e\in\Ed}\det H_e^{m_e}$.
Applying~\eqref{eq:ort_U[u]} to the vectors $\mathbf u_1:=\mathbf H[\mathbf v_1]$ and $\mathbf u_2:=\mathbf H[\mathbf v_2]$ one sees that the integral over~$\su^\Ed$ vanishes for each~$\mathbf H\in\mathbb B_1^+$.
\end{proof}

\subsection{Exponential lower bound for the norms of functions~$\bm{G_{\Gamma,\mathbf m}}$}\label{sec:lowerbound} The goal of this section is to derive the following uniform lower bound for the $L^2$-norms of the functions~$G_{\Gamma,\mathbf{m}}$ on poly-balls~$\mathbb B_R$:

\begin{prop}
\label{prop:lower-bound}
There exists a (small) absolute constant~$\eta_0>0$ such that for all laminations~$\Gamma$ and multi-indices~$\mathbf m\in\mathbb Z_{\ge 0}^\Ed$, the following estimate holds:
\begin{equation}
\label{eq:lower-bound}
\|G_{\Gamma,\mathbf m}\|_{L^2(\mathbb B_R)}\ge (\eta_0 R)^{|\Gamma|+2|\mathbf m|+4|\Ed|}.
\end{equation}
The constant~$\eta_0$ is independent of~$\G$ and \bl{the number $n$ of punctures}~$\lambda_1,\dots,\lambda_n\in\Omega$.
\end{prop}
\begin{rem} The proof is postponed until the end of this section. Note that~$G_{\Gamma,\mathbf m}$ is a homogeneous polynomial of total degree~$|\Gamma|+2|\mathbf m|$ and~$\mathbb B_R$ is a poly-ball of \bl{the real} dimension~$8|\Ed|$. Therefore, one can assume~$R=1$ without loss of generality.
\end{rem}

We need some preliminaries. Recall that the space~$V_n$ can be thought of as the space~$\mathbb C[z,w]^{(n)}$ of homogeneous polynomials of degree~$n$ in two variables~$z,w$; we fix an isomorphism~$V_n\simeq\mathbb C[z,w]^{(n)}$ by identifying the monomial~$z^kw^{n-k}$ with the basis vector~$(x^ky^{n-k})^\sym\in V_n$, see Section~\ref{ReparametrizingFunctions}. For each multi-index~$\mathbf n\in \mathbb Z_{\ge 0}^\Ed$, this isomorphism induces an isomorphism
\begin{equation}
\label{eq:Vn_to_polynomials}
\mathbf V_{\mathbf n}\ \simeq\ \bigotimes_{\sigma\in\Fa,\,e\in\partial\sigma}V_{\sigma,e,n_e}\simeq\ \mathbb C[\mathbf z,\mathbf w]^{(\mathbf n)}
\end{equation}
of the space~$\mathbf V_{\mathbf n}$ and the space of homogeneous polynomials of multi-degree~$\mathbf n$ in variables~$\mathbf z=(z_{\sigma,e})_{\sigma\in\Fa,e\in\partial\sigma}$ and~$\mathbf w=(w_{e,\sigma})_{\sigma\in\Fa,e\in \partial\sigma}$, note that the total degree of these polynomials is~$2|\mathbf n|$ since~$\mathbf V_{\mathbf n}$ contains two factors~$V_{\sigma,e,n_e}$ per edge~$e\in\Ed$.

\begin{lemma}
\label{lem:Pgamma=}
Given~$\mathbf n\in \L$ and the lamination~$\Gamma=\Gamma(\mathbf n)$, let~$\mathbf v_\Gamma\in\mathbf V_{\mathbf n}$ be the vector corresponding to the function $g_\Gamma$ 
under the isomorphism~\eqref{eq:morePeterWeylRepresentation} and~$P_\Gamma\in\mathbb C[\mathbf z,\mathbf w]^{(\mathbf n)}$ be the homogeneous polynomial corresponding to~$\mathbf v_\Gamma$ under the isomorphism~\eqref{eq:Vn_to_polynomials}. Then,
\[
\textstyle P_\Gamma(\mathbf z,\mathbf w)\ =\ \pm\prod_{\sigma\in\Fa}\prod_{i=1}^3 (w_{\sigma,e_i}z_{\sigma,e_{i+1}}-z_{\sigma,e_i}w_{\sigma,e_{i+1}}\bigr)^{\frac{1}{2}(n_{e_i}+n_{e_{i+1}}-n_{e_{i+2}})},
\]
where~$e_1,e_2,e_3\in\Ed$ denote the three edges adjacent to~$\sigma$ and~$e_{i+3}:=e_i$.
\end{lemma}

\begin{proof}
Due to Lemma~\ref{gViaLoops} the vector~$\mathbf v_\Gamma$ equals to the symmetrization of vectors~$\mathbf w_{\C(\Gamma,\pi)}$ defined in~\eqref{eq:w_CGammaPi} over~$\pi\in\prod_{\sigma\in \Fa,\,e\in \partial \sigma} \Sym(n_e)$. For each face~$\sigma\in\Fa$  the symmetrization of the vectors
\[
\textstyle \bigotimes_{\text{$c$ -- chord of $\C(\Gamma,\pi)$ in $\sigma$}} u_{c,\Gamma,\pi}
\]
over~$\prod_{i=1}^3\Sym(n_{e_i})$ corresponds to the polynomial~$\pm\prod_{i=1}^3(w_{\sigma,e_i}z_{\sigma,e_{i+1}}-z_{\sigma,e_i}w_{\sigma,e_{i+1}}\bigr)^{k_{i,i+1}}$\!,
where~$k_{i,i+1}$ stands for the number of chords of~$\Gamma$ connecting the edges~$e_i$ and~$e_{i+1}$ in~$\sigma$. The claim follows by taking the product over all~$\sigma\in\Fa$.
\end{proof}

Denote
\[
\textstyle \mathbb S:=\prod_{\sigma\in\Fa,e\in\partial\sigma}\mathrm S_{\sigma,e},\qquad
\mathrm S_{\sigma,e}:=\{(z_{\sigma, e},w_{\sigma, e})\in \mathbb C^2\,\mid\, |z_{\sigma, e}|^2 + |w_{\sigma, e}|^2 = 1\}.
\]

\begin{lemma} \label{combToInt}
Let~$\mathbf v_1,\mathbf v_2\in \mathbf V_{\mathbf n}$ and~$P_1,P_2\in \mathbb C[\mathbf z,\mathbf w]^{(\mathbf n)}$ be the corresponding homogeneous polynomials of multi-degree~$\mathbf n$ (and of total degree~$2|\mathbf n|$). Then,
\[
\lan \mathbf v_1,\mathbf v_2\ran_{\mathbf V_{\mathbf n}}\ =\ \prod\limits_{e\in \Ed}\frac{(n_e+1)^2}{4\pi^4}\,\cdot \int_{\mathbb S} P_1(\mathbf z,\mathbf w)\,\overline{P_2(\mathbf z,\mathbf w})\,d\lambda_{\mathbb S}\,,
\]
where~$d\lambda_{\mathbb S}:=\prod_{\sigma\in\Fa,\,e\in\partial\sigma}\lambda_{\mathrm S_{\sigma,e}}$ is the product of the surface measures on the spheres~$\mathrm S_{\sigma,e}$.
\end{lemma}

\begin{proof} It is enough to consider the case~$\mathbf v_1 = \mathbf v_2 =: \mathbf v=\bigotimes_{\sigma\in\Fa,\,e\in\partial\sigma}v_{\sigma,e}$. For such vectors the claim follows from the component-wise identity
\[
\|v\|_{V_n}^2=\frac{n+1}{2\pi^2}\int_{\mathrm S}|P(z,w)|^2d\lambda_{\mathrm S},
\]
where~$v\in V_n$ and~$P\in \mathbb C[z,w]^{(n)}$ is the homogeneous polynomial corresponding to~$v$.

To prove this identity consider a mapping~$\su\to\mathrm S$ given by~$g\mapsto v^{(0)}g$, where~$v^{(0)}=(1,\,0)\in\mathbb C^2$ corresponds to the basis vector~$x\in V_1$ in the notation of Section~\ref{ReparametrizingFunctions}. This mapping is a diffeomorphism and the pushforward of the Haar measure on~$\su$ is ~$(2\pi^2)^{-1}\lambda_{\mathrm S}$. Since~$P$ corresponds to~$v$, one has~$P(v^{(0)}g)=\lan g(x^{\otimes n}),\overline{v}\ran_{V_n}$. Therefore,
  \begin{equation*}
    \frac{1}{2\pi^2}\int_{\mathrm S} |P(z,w)|^2\,d\lambda_{\mathrm S} = \int_{\su} |\lan g(x^{\otimes n}),\overline{v}\ran_{V_n}|^2\,d\mu_{\mathrm{Haar}} = \frac{\|v\|_{V_n}^2}{n+1},
  \end{equation*}
  where the last equality follows from Theorem~\ref{analyticPeterWeyl}.
\end{proof}

\begin{proof}[{\bf Proof of Proposition~\ref{prop:lower-bound}}] As already mentioned above, one can assume~$R=1$ due to homogeneity reasons. Let~$\mathbf v_\Gamma\in\mathbf V_{\mathbf n}$ and~$P_\Gamma\in \mathbb C[\mathbf z,\mathbf w]^{(\mathbf n)}$ correspond to the function~$g_\Gamma$ as discussed in Lemma~\ref{lem:Pgamma=}. Using factorization~\eqref{eq:fubini} of the Euclidean measure on~$\mathrm B_1$ in the same way as in the proof of Lemma~\ref{GareOrt} one gets the identity
\begin{equation*}
    \|G_{\Gamma,\mathbf {m}}\|_{L^2(\mathbb B_1)}^2\, =\, \int_{\mathbb B_1^+}\nu(\mathbf{dH})\,|\det \mathbf H^{2\mathbf m}|
    \int_{\su^\Ed}|\pair(\mathbf{UH}[\mathbf v_\Gamma])|^2\,\mu_{\mathrm{Haar}}(\mathbf{dU}),
\end{equation*}
Theorem~\ref{analyticPeterWeyl} gives
\[
\int_{\su^\Ed}|\pair(\mathbf{UH}[\mathbf v_\Gamma])|^2\,\mu_{\mathrm{Haar}}(\mathbf{dU})\ =\ \prod\nolimits_{e\in\Ed}\frac{1}{(n_e+1)^{2}}\cdot \|\,\mathbf H[\mathbf v_\Gamma]\,\|_{\mathbf V_\mathbf n}^2.
\]
Applying Lemma~\ref{combToInt} we arrive at the identity
\begin{equation}
\label{eq:NormG=}
\|G_{\Gamma,\mathbf {m}}\|_{L^2(\mathbb B_1)}^2\, =\, (4\pi^4)^{-|\Ed|}\cdot \int_{\mathbb B_1^+}\nu(\mathbf{dH})\,|\det \mathbf H^{2\mathbf m}| \int_{\mathbb S}|\,\mathbf H[P_\Gamma]\,|^2 d\lambda_{\mathbb S},
\end{equation}
where~$\mathbf H[P_\Gamma]$ is obtained from~$P_\Gamma$ by the change of variables~$(z_{\sigma,e}\,,w_{\sigma,e})\mapsto (z_{\sigma,e}\,,w_{\sigma,e})H_e$.
\bl{Finally, there exist sufficiently small absolute constants $\nu_0>0$ and $\varepsilon_0=\varepsilon_0(\nu_0)>0$ such that we have
\[
\nu(\{\,H\in\mathrm B_1^+\mid\, \det H\ge \nu_0\,\})\ \ge\ \nu_0
\]
and, using the explicit form of the polynomial~$P_\Gamma$ (see Lemma~\ref{lem:Pgamma=})),
\[
\lambda_{\mathbb S}(\{\,(\mathbf z,\mathbf w)\in \mathbb S\,\mid\,|\,\mathbf{H}[P_\Gamma]\,|\ge \varepsilon_0^{|\mathbf n(\Gamma)|}\,\})\ \ge\ \varepsilon_0^{|\Fa|}
\]
provided that~$\det H_e\ge \nu_0$ for all edges $e\in\Ed$. The desired uniform lower bound~\eqref{eq:lower-bound} for the norm~$\|G_{\Gamma,\mathbf {m}}\|_{L^2(\mathbb B_1)}$, with an appropriate $\eta_0=\eta_0(\nu_0,\varepsilon_0)>0$, follows easily.
}
\end{proof}





\section{Estimate of coefficients in the Fock--Goncharov theorem}\label{skein_algebra_section}

The goal of this section is to derive the estimate $|c_{\Delta\Gamma}|\le 4^{|\Gamma|}$ for the coefficients of the Fock--Goncharov change of the bases $g_\Gamma=\sum_{\Delta\le\Gamma}c_{\Gamma\Delta}f_\Delta$ discussed in Section~\ref{fViagAndgViaf}. To the best of our understanding, a similar exponential bound was implicitly used in~\cite[p.~483]{Kenyon} and~\cite[p.~957]{Kassel-Kenyon} but we were unable to find a proof of such an estimate in the existing literature. \bl{It is worth noting that an exponential bound for the coefficients of the \emph{inverse} change of the bases~$f_\Gamma=\sum_{\Delta\le\Gamma}\tilde{c}_{\Gamma\Delta}g_\Delta$ trivially follows from the orthogonality of~$g_\Delta$ since the norms $\|f_\Gamma\|$ are trivially exponentially bonded from above and the norms $\|g_\Gamma\|$ are exponentially bounded from below due to Proposition~\ref{prop:lower-bound}. However,} this fact does \emph{not} imply an exponential estimate of~$c_{\Gamma\Delta}$.

In order to study the coefficients~$c_{\Gamma\Delta}$ we use the well-known connection between the Kauffman skein algebra with the parameter $q=-1$ and the ring of invariants $\F(X)^{\sl}$ (e.g., see~\cite{Marche}): functions $g_{\Gamma}$ and $f_{\Gamma}$ correspond to some natural elements of the skein algebra and the matrix $c_{\Gamma\Delta}$ admits some combinatorial description through this correspondence. One of the key ingredients of the proof given below is a recent result due to D.~Thurston~\cite{Thurston} on the positivity of the so-called \emph{bracelet basis} in the skein algebra with~$q=1$, see Section~\ref{braceletAndPositivity} for more details. Another input is a representation of the skein algebra in the space of Laurent polynomials via Thurston's shear coordinates on the moduli space of hyperbolic structures on~$\Omega\setminus\{\lambda_1,\dots,\lambda_n\}$, see Section~\ref{sub:shear_coordinates}. These ideas were kindly communicated to the authors by Vladimir Fock and we believe that a core part of the proof of Theorem~\ref{thm:diagEst} should in fact be credited to him.

Several parts of the material presented in this section are classical or well-known. As in Section~\ref{representation_theory_section} we collect all them together in order to introduce a consistent notation and for the sake of completeness.

\subsection{Definition of the Kauffman skein algebra and Przytycki--Sikora theorem}\label{skeinDef} Let $q\in\mathbb C^*$, $\Sigma:=\Omega\setminus\{\lambda_1,\dots,\lambda_n\}$, and $M:= [0,1]\times \Sigma$. A framed link $L$ in $M$ is an embedding of a disjoint union of circles $\iota : \sqcup S^1\to M$ together with a continuous choice of a positive oriented basis in the fiber of $TM$ at each point of~${\iota(\sqcup S^1)}$ (i.e., a choice of everywhere linearly independent sections $v_1, v_2, v_3: \iota(\sqcup S^1) \to TM_{\iota(\sqcup S^1)}$).
We say that an (oriented) knot is trivially framed if $v_1$ is the tangent vector to the knot and $v_3$ is everywhere vertical, each framed knot is equivalent to a trivially framed one up to a framed isotopy. Denote by $L_{\mathrm{unknot}}$ the trivially framed unknot associated with a simple contractible loop in $\Sigma$.

Let $W$ be the $\mathbb C$-vector space generated by framed isotopy classes of framed links. \bl{We define a multiplication of two links $L$ and $L'$ by placing $L$ under $L'$} (with respect to the direction of the projection $[0,1]\times \Omega\to [0,1]$) and taking the union. This makes $W$ into an algebra with the unit represented by the empty link. We say that three framed links $L_1$, $L_0$ and $L_{\infty}$ form a Kauffman triple if they can be drawn identically except in a ball where they appear as shown in Figure~\ref{fig:Kauffman_triple}. Let $\CMcal{I}$ be the two-sided ideal in $W$ generated by the relations $L_1 - qL_0 - q^{-1}L_{\infty}=0$ for all Kauffman triples and $L_{\mathrm{unknot}} + q^2+q^{-2}=0$. The Kauffman skein algebra with the parameter $q$ is the quotient
\begin{equation*}
  \sk(M, q)\ :=\ W/\CMcal{I}.
\end{equation*}

\begin{figure}
  \centering
  \tempskip{\begin{minipage}{0.2\textwidth}\centering
  \includegraphics[clip, trim=2.35cm 1.8cm 10.75cm 0cm, width=\textwidth]{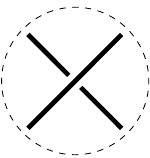}

  \bigskip
  {\large $L_1$}
  \end{minipage}\hskip 0.1\textwidth
  \begin{minipage}{0.2\textwidth}\centering
  \includegraphics[clip, trim=6.53cm 1.8cm 6.57cm 0cm, width=\textwidth]{kauffman_triple.pdf}

  \bigskip
  {\large $L_0$}
    \end{minipage}\hskip 0.1\textwidth
  \begin{minipage}{0.2\textwidth}\centering
  \includegraphics[clip, trim=10.57cm 1.8cm 2.53cm 0cm, width=\textwidth]{kauffman_triple.pdf}

  \bigskip
  {\large \ $L_\infty$}
  \end{minipage}}
  \caption{A Kauffman triple~$(L_1,L_0,L_\infty)$. The skein algebra relations are~$L_1=qL_0+q^{-1}L_\infty$ (complemented by $L_{\mathrm{unknot}}=-q^2-q^{-2}$).}
  \label{fig:Kauffman_triple}
\end{figure}

In our paper we are interested in the particular cases~$q=\pm 1$ only. Note that the relations in~$\CMcal{I}$ immediately imply that the algebras~$\sk(M,\pm 1)$ are commutative.The following theorem is due to Przytyscki and Sikora~\cite{Przytycki_Sikora}, see also~\cite{Bullock} for a close result.

\begin{theorem}\label{SkToFunctions}
 Given a knot~$K$, let~$\gamma_K$ denote its projection onto~$\Sigma$. The mapping
\begin{equation*}
  \psi: \sk(M, -1)\ \to\ \F(X)^{\sl}
\end{equation*}
defined on trivially framed knots by $\psi(K)(\rho) := -\Tr(\rho(\gamma_K))$ and extended by additivity and multiplicativity is a correctly defined isomorphism of algebras.
\end{theorem}
\begin{proof}
Recall that the algebra $\sk(M, -1)$ is commutative. The Kauffman triple relation for~$q=-1$ reflects the identity
\begin{equation*}
  \Tr(AB) + \Tr(A^{-1}B) = \Tr (A)\Tr(B),\qquad A,B\in\sl.
\end{equation*}
These observations can be completed to a proof of the fact that the mapping $\psi$ is a well-defined homomorphism of algebras as it is done in~\cite[proof of Theorem~3]{Bullock}.

One can now use the Fock--Goncharov theorem to prove that~$\psi$ is an isomorphism. Namely, each lamination~$\Gamma$ naturally gives rise to a framed link in~$M$: represent~$\Gamma$ as a disjoint union of simple curves and attach a trivial frame to each of these curves. It is easy to see that the image of thus obtained framed link under~$\psi$ \bl{equals (up to a sign)}~$f_\Gamma$ and that such links span~$\sk(M,-1)$. Using the fact that~$f_\Gamma$ is a basis of $\F(X)^{\sl}$ one concludes that $\psi$ is an isomophism.
\end{proof}

\subsection{Skein algebra with $\bm{q=1}$ and twisted representations}\label{twistedRepresentationsSubsection}
\bl{It is well known that each spin structure on $\Sigma$ provides an isomorphism between the skein algebras~$\sk(M,-1)$ and~$\sk(M,1)$. Below  we briefly discuss this correspondence.}

Recall that the unit tangent bundle $U\Sigma$ of~$\Sigma$ is defined as
\begin{equation*}
  U\Sigma := (T\Sigma\smallsetminus 0)\,/\,\mathbb R_{>0}\,,
\end{equation*}
where $0$ stands for the zero section of~$T\Sigma$. Since $\pi_2(\Sigma) = 0$, one has a short exact sequence $0\to\pi_1(S^1)\to \pi_1(U\Sigma)\to \pi_1(\Sigma)\to 0$ and concludes that $\pi_1(U\Sigma)$ is a central extension of $\pi_1(\Sigma)$ by $\mathbb Z$. Let $z\in\pi_1(U\Sigma)$ denote the image of the generator of $\mathbb Z$. A \emph{spin structure} on $\Sigma$ is an element $\xi\in H^1(U\Sigma, \mathbb Z_2)$ such that $\xi(z) = 1$.
For a trivially framed knot~$K$, let~$\widetilde{\gamma}_K$ be a loop in~$U\Sigma$ given by points of~$\gamma_K\subset \Sigma$ and projections of the first (tangential) vector of the framing. It is easy to see that isotopic framed knots define homotopic curves in~$U\Sigma$.

It is also well known that each spin structure~$\xi$ on a compact orientable surface~$\Sigma$ can be given by a vector field~$V_\xi$ on~$\Sigma$ with isolated zeroes of even index: if~$K$ is a trivially framed knot, then
\[
s_\xi(\widetilde{\gamma}_K)\ :=\ (-1)^{\xi(\widetilde{\gamma}_K)}\ =\ e^{\frac{i}{2}\mathrm{wind}(\gamma_K;V_\xi)},
\]
where~$\mathrm{wind}(\gamma_K;V_\xi)$ denotes the total rotation angle of the tangent vector to~$\gamma_K$ measured with respect to~$V_\xi$. In fact, below one can always assume that~$V_\xi$ is constant (which corresponds to the unique spin structure on the original simply connected domain~$\Omega\supset\Sigma$), so that~$\mathrm{wind}(\gamma_K;V_\xi)=\mathrm{wind}(\gamma_K)$ is the usual total rotation angle of the tangent vector along the loop~$\gamma_K$ in~$\Omega$.
\begin{lemma}\label{spinHomo}
For each spin structure $\xi$ on~$\Sigma$ the mapping
\[
\phi_{\xi}: \sk(M, q)\to \sk(M, -q),\qquad \textstyle \phi_{\xi}(L) := \prod_{K\subset L} (-s_\xi(\widetilde{\gamma}_K)) \cdot L,
\]
where the product is taken over all components~$K$ of a framed link~$L$, is a correctly defined isomorphism of algebras.
\end{lemma}

\begin{proof}
See~\cite[Theorem~1]{Barrett}. Note that the value~${\mathrm{Spin}(\widetilde{\gamma},\xi)}$ used in~\cite{Barrett} is equal to $1-\xi(\widetilde{\gamma})$ in our notation. Thus, one has~$(-1)^{\mathrm{Spin}(\widetilde{\gamma},\xi)}=\bl{(-1)^{1-\xi(\widetilde{\gamma})}=}-s_\xi(\widetilde{\gamma})$.
\end{proof}

We now introduce an additional terminology that will be used in Section~\ref{sub:shear_coordinates}.
\begin{defin}
A representation $\widetilde\rho: \pi_1(U\Sigma)\to \sl$ such that $\widetilde \rho(z) = -\Id$ is called a {twisted} $\sl$ representation on $\Sigma$. Let
\[
\Xs:=\{\,\widetilde{\rho}:\pi_1(U\Sigma)\to \sl\,\mid\,\widetilde{\rho}~\text{is a twisted $\sl$ representation on $\Sigma$}\,\}.
\]
\end{defin}

Note that conjugating a twisted representation by an element of~$\sl$ one gets another twisted representation. Furthermore, $s_\xi(z)\widetilde{\rho}(z)=\Id$ for each twisted representation~$\widetilde{\rho}$, hence one can correctly define the mapping
\begin{equation}
\label{eq:defSxi}
S_\xi: \Xs\to \X,\qquad \widetilde{\rho}\ \mapsto\ s_\xi\cdot\widetilde{\rho},
\end{equation}
as~$s_\xi(\widetilde{\gamma})\cdot\widetilde{\rho}(\widetilde{\gamma})$ depends only on the projection of a loop~$\widetilde{\gamma}\in\pi_1(U\Sigma)$ onto~$\pi_1(\Sigma)$. Let
\[
S_\xi^*\ :\ \F(\X)^{\sl}\ \to\ \F(\Xs)^{\sl}
\]
be the pullback of~$S_\xi$.

\begin{prop}
\label{prop:Xtwist} Let~$\xi$ be a spin structure on~$\Sigma$ and~$\psi:\sk(M,-1)\to\F(\X)^{\sl}$ be the isomorphism from Theorem~\ref{SkToFunctions}. Then, the algebra homomorphism
\[
\Psi:=S_\xi^*\circ \psi\circ \phi_\xi\ :\ \sk(M,1)\ \to\ \F(\Xs)^{\sl}
\]
acts on trivially framed knots as $\Psi(K)(\widetilde \rho) := \Tr(\widetilde \rho(\widetilde\gamma_K))$, $\widetilde \rho\in \Xs$. In particular, $\Psi$ does not depend on the choice of the spin structure~$\xi$ on~$\Sigma$.
\end{prop}

\begin{proof} This is a trivial corollary of Theorem~\ref{SkToFunctions}, Lemma~\ref{spinHomo} and definition~\eqref{eq:defSxi}.
\end{proof}

\begin{rem}\label{rem:Xtwist} It follows from Proposition~\ref{prop:Xtwist} that each twisted representation~$\widetilde{\rho}\in\Xs$ gives rise to an algebra homomorphism (evaluation)~$\Psi_{\widetilde{\rho}}:\sk(M,1)\to\mathbb C$ defined on trivially framed knots as $K\mapsto \Tr(\widetilde\rho(\widetilde\gamma_K))$ and extended by additivity and multiplicativity.
\end{rem}


\subsection{Twisted representations associated with hyperbolic structures on~$\bm{\Sigma}$ and the representation of~$\bm{\sk(M,1)}$ in the algebra of Laurent polynomials.}\label{sub:shear_coordinates}
In this section we discuss the class of twisted representations~$\widetilde{\rho}_{\mathbf x}$ coming from hyperbolic structures on $\Sigma$ parameterized by Thurston's shear coordinates~$\mathbf x=(x_e)_{e\in\Ed}\in \mathbb R_+^\Ed$. In particular, we show that all functions~$\mathbf x\mapsto \Tr(\widetilde\rho_{\mathbf x}(\widetilde\gamma_K))$ associated with trivially framed knots $K\in\sk(M,1)$ are Laurent polynomials in the variables~$\sqrt{x_e}$ with \emph{sign coherent} (simultaneously positive or simultaneously negative) coefficients, the  last observation was pointed out to us by Vladimir Fock. We discuss only the case when \mbox{$\Sigma=\Omega\setminus\{\lambda_1,\dots,\lambda_n\}$} is a punctured simply connected domain, see~\cite{BonahonWong} for the general construction in the case of surfaces with free non-abelian fundamental group.

Let~$\pi:\overline{\Sigma}\to\Sigma$ be the universal cover of~$\Sigma$. Classically, given a \emph{hyperbolic metric} on~$\Sigma$ one can consider a {uniformization}~$f:\mathbb H\to \overline{\Sigma}$ where~$\mathbb H:=\{z\in\mathbb C:\mathrm{Im} z>0\}$ stands for the upper complex plane, note that~$f$ is defined uniquely up to M\"obius automorphisms of~$\mathbb H$. Provided that the edges of the triangulation~$\G$ of~$\Sigma$ are drawn as hyperbolic geodesics and~$\overline\G:=\pi^{-1}(\G)$ denotes the corresponding triangulation of~$\overline{\Sigma}$, the following properties hold:

(i)$\phantom{i}$ the preimage under~$f$ of each triangle of~$\overline\G$ is an ideal triangle in~$\mathbb H$;

(ii) deck transformations of~$\overline{\Sigma}$ form a discrete subgroup of the group~$\mathrm{Aut}(\mathbb H)$.

\noindent Recall that~$\mathrm{Aut}(\mathbb H)\cong\psl$ (automorphisms of~$\mathbb H$ are given by M\"obius transforms~$z\mapsto (az+b)/(cz+d)$ with~$a,b,c,d\in\mathbb R$ defined up to a common multiple) and that 
for each~$(z_0,v_0),(z_1,v_1)\in U\mathbb H$ there exists a unique automorphism~$g\in \mathrm{Aut}(\mathbb H)$ such that~$g(z_0)=z_1$ and the vector~$g'(z_0)v_0$ has the same direction as~$v_1$.

Given a uniformization~$f:\mathbb H\to \overline{\Sigma}$ and a base point~$(z_0,v_0)\in U\mathbb H$, one can construct a \emph{twisted representation} $\widetilde{\rho}_f$ as follows. For a smooth loop~$\widetilde{\gamma}:[0,1]\to U\Sigma$ such that~$\widetilde\gamma(0)=\widetilde\gamma(1)=(\pi\circ f)(z_0,v_0)$, let~$\overline{\gamma}:[0,1]\to U\overline{\Sigma}$ denote its lift on~$U\overline{\Sigma}$ 
and let~$g_{\overline{\gamma}}(t)\in\mathrm{Aut}(\mathbb H)$ sends~$(f^{-1}\circ \overline{\gamma})(0)=(z_0,v_0)$ to~$(f^{-1}\circ\overline{\gamma})(t)$. One can now lift the path~$g_{\overline{\gamma}}:[0,1]\to\mathrm{Aut}(\mathbb H)\cong\psl$ to~$\widetilde{g}_{\overline{\gamma}}:[0,1]\to \slr$ so that~$\widetilde{g}_{\overline{\gamma}}(0)=\Id$ and set
\begin{equation}
\label{eq:tilde_rho_def}
\widetilde{\rho}_f(\widetilde{\gamma}):= (\widetilde{g}_{\overline{\gamma}}(1))^{-1}\,,
\end{equation}
note that this mapping is nothing but a lift to~$\slr$ of the automorphism of~$\mathbb H$ corresponding to the deck transformation of~$\overline{\Sigma}$ given by the projection of~$\widetilde{\gamma}$ onto~$\Sigma$. In particular, $\widetilde{\rho}_f$ does not depend on the choice of the lift~$\overline{\gamma}$ of~$\widetilde{\gamma}$. It is straightforward to check that~$\widetilde{\rho}_f$ is a twisted representation on~$\Sigma$. Moreover, the change of the base point~$(z_0,v_0)$ and/or of the uniformization~$f$ amounts to a $\sl$-conjugation of~$\widetilde{\rho}_f$.

We now introduce Thurston's \emph{shear coordinates} on the space of hyperbolic structures on~$\Sigma$. Given a collection~$\mathbf x=(x_e)_{e\in\Ed}$ of positive parameters associated to the edges of the triangulation~$\G$, one can construct a hyperbolic structure on~$\Sigma$ such that the following property holds for (any of) the corresponding uniformizations~$f_{\mathbf x}:\mathbb H\to \overline{\Sigma}$:

(iii) for each edge~$e$ of~$\G$, the quadrilateral~$\sigma_{\mathrm{left}}(e)\cup e \cup \sigma_{\mathrm{right}}(e)$ formed by two faces adjacent to~$e$ is conformally equivalent to the ideal quadrilateral \mbox{$R(x_e)\subset \mathbb H$} with vertices~$-1,0,x_e,\infty$ so that the edge~$e$ corresponds to the line \mbox{$\sqrt{-1}\cdot \mathbb R_+\subset R(x_e)$}.

In fact, one can use property (iii) to draw a preimage of the triangulation~$\overline\G$ in~$\mathbb{H}$. To do this, start with two faces adjacent to an arbitrarily chosen edge and draw them so that they form the ideal rectangle~$R(x_e)$ and then extend this drawing step by step so that each face of~$\overline{\G}$ is drawn as an ideal triangle and the property (iii) is satisfied. As a result of this procedure one obtains a diffeomorphism~$f_{\mathbf x}:\mathbb H\to \overline{\Sigma}$ which can be used to define the desired hyperbolic structure on~$\overline{\Sigma}$ and to project it onto~$\Sigma$. It can be easily shown that each hyperbolic structure on~$\Sigma$ can be obtained in this manner but we do not need this fact in our paper.

Below we denote by~$\widetilde\rho_{\mathbf x}$ the twisted $\sl$ representation corresponding to the hyperbolic metric on~$\Sigma$ and the uniformization $f_{\mathbf x}:\mathbb H\to \overline\Sigma$ constructed above. Recall that~$\widetilde\rho_{\mathbf x}$ is defined up to a $\sl$-conjugation only as we do not fix neither a base point~\mbox{$(x_0,v_0)\in\mathbb H$} nor the starting edge of the construction. Nevertheless, following Proposition~\ref{prop:Xtwist} and Remark~\ref{rem:Xtwist}, for each~$\mathbf x\in \mathbb R_+^\Ed$ one can correctly define the evaluation
\begin{equation}
\label{eq:Psi_x}
\Psi_{\mathbf x}\ :\ \sk(M,1)\ \to\ \mathbb C,\qquad L\mapsto (\Psi(L))(\widetilde \rho_{\mathbf x}),
\end{equation}
which acts on trivially framed knots as~$\Psi_{\mathbf x}(K)=\Tr(\widetilde{\rho}_{\mathbf x}(\widetilde\gamma_K))$ and is extended to the skein algebra~$\sk(M,1)$ by linearity and multiplicativity.

The following proposition is the key result of this section. Recall that, given a minimal smooth multi-curve~$\C$ on~$\Sigma=\Omega\setminus\{\lambda_1,\dots,\lambda_n\}$ (see section~\ref{fViagAndgViaf} for the definition) one can obtain an element~$L(\C)\in\sk(M,1)$ by framing all the components of~$\C$ trivially and then resolving the intersections in an arbitrary way: the resulting link depends on the choice of these resolutions but its class in~$\sk(M,1)$ does not.

\begin{prop}\label{prop:linkToPolynomials} (i) For each~$L\in\sk(M,1)$ the function~$\Psi_{\mathbf x}(L)$ is given by evaluating a Laurent polynomial \mbox{$P_L\in\mathbb C[(t_e,t_e^{-1})_{e\in\Ed}]$} at~$t_e:=\sqrt{x_e}$, where~$\sqrt{x}$ denotes the positive square root of~$x>0$. The mapping~$L\mapsto P_L$ is an algebra homomorphism.

\noindent (ii) If~$L=L(\C)$ is obtained from a minimal multi-curve as described above, then~$P_L$ has positive integer coefficients.
\end{prop}

\begin{proof} Since for each framed link~$L$ there exists a trivially framed link~$L'$ obtained from a minimal multi-curve such that~$L=\pm 2^m L'$ in~$\sk(M,1)$ for some~$m\ge 0$, one can assume that~$L$ is obtained from a minimal multi-curve. Moreover, since~$\Psi_{\mathbf x}$ is an algebra homomorphism one can further assume that~$L=L(\gamma)$ is a trivially framed knot obtained from a minimal (i.e. having no nugatory self-crossings, see Fig.~\ref{fig:localtwist}) loop~$\gamma$.

  \bl{Each minimal loop~$\gamma$ can be combinatorially encoded by a non-backtracking path on the graph $\G^\circ$. Let us assume that this path consequently crosses edges $e_0,e_1,\dots, e_{m-1}$ of $\G$, and $\sigma_0,\sigma_1,\dots,\sigma_m$, $\sigma_m=\sigma_0$, is the sequence of faces of $\G$ corresponding to the vertices of this path, so that $e_i$ lies between $\sigma_i$ and $\sigma_{i+1}$. For each edge $e_i$, }
let the point~\mbox{$a_i\in e_i$} be defined by the following condition: $a_i=\varphi_i(\sqrt{-1})$, where $\varphi_i$ denotes the uniformization~$\varphi_i:\mathbb H\to\overline{\Sigma}$ that maps the ideal quadrilateral with vertices $-1,0,x_{e_i},\infty$ onto~$\sigma_i\cup e_i\cup\sigma_{i+1}$ so that~$\varphi^{-1}_i(\sigma_i)$ has vertices~$-1,0,\infty$ whilst~$\varphi^{-1}_i(\sigma_{i+1})$ has vertices~$0,x_e,\infty$, cf. condition (iii) above.


\begin{figure}
  \centering
  \begin{minipage}{.482\textwidth}
    \centering
    \tempskip{\includegraphics[width =\textwidth]{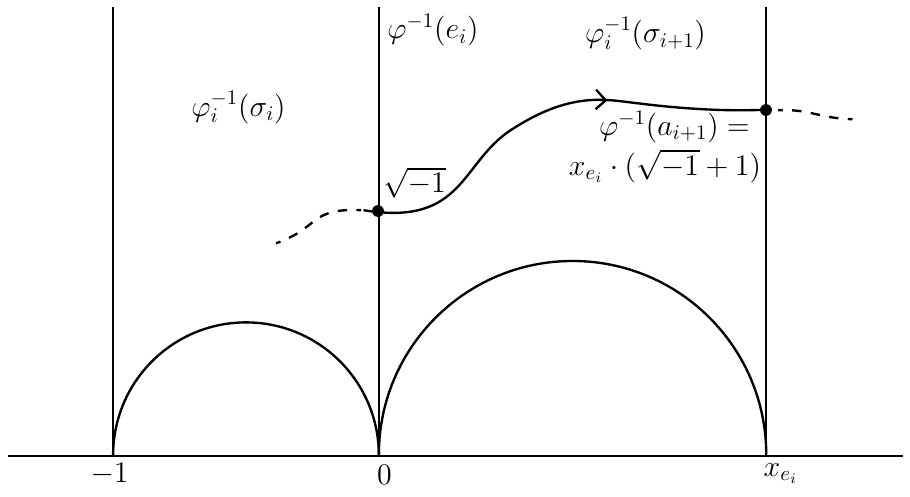}}
    \caption*{counterclockwise turn}
    \end{minipage}
    \begin{minipage}{.508\textwidth}
    \centering
    \tempskip{\includegraphics[width =\textwidth]{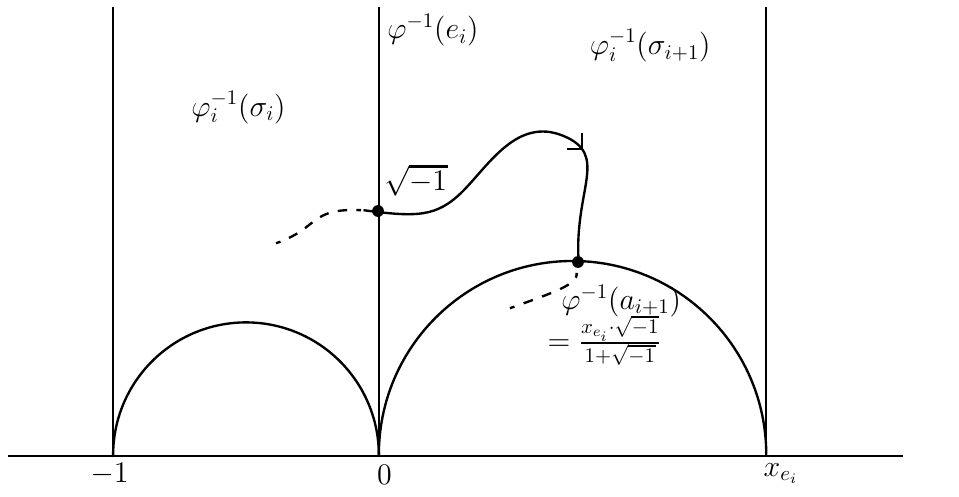}}
    \caption*{clockwise turn}
  \end{minipage}
  \caption{The piece~$\gamma(e_i,e_{i+1})$ of a loop~$\gamma$ in the chart~$\varphi_i^{-1}:\overline{\Sigma}\to\mathbb C$.}
  \label{fig:through_ideal_triangles}
\end{figure}


One can now replace~$\gamma$ by the concatenation of smooth segments~$\gamma(e_i,e_{i+1})$ crossing the edges of~$\Ed$ orthogonally and running from~$a_i\in e_i$ to~$a_{i+1}\in e_{i+1}$ inside the face~$\sigma_{i+1}$. For each~$i=0,\dots,m-1$,

(a) either the preimage~$\varphi_i^{-1}(e_{i+1})$ of the next edge crossed by~$\gamma$ connects~$x_{e_i}$ and~$\infty$, in this case we say that~$\gamma$ makes a counterclockwise turn inside~$\sigma_{i+1}$, see Fig.~\ref{fig:through_ideal_triangles},

(b) or this preimage of~$e_{i+1}$ connects~$0$ and~$x_{e_i}$: a clockwise turn inside~$\sigma_{i+1}$.

\noindent One easily sees that the composition mapping~$\varphi_i^{-1}\circ \varphi_{i+1}^{\vphantom{1}}$ (which maps the line~$\sqrt{-1}\cdot\mathbb R_+$ onto the preimage~$\varphi_i^{-1}(e_{i+1})$) is given by
\[
z\mapsto x_{e_i}\cdot (z+1)\ \ \text{in the case (a)},\qquad z\mapsto \frac{x_{e_i} z}{z+1}\ \ \text{in the case (b)},
\]
moreover the lift of the corresponding path in~$\mathrm{Aut}(\mathbb H)\cong \psl$ to~$\slr$ leads to the following matrix~$A_i$ representing the M\"obius automorphism~$\varphi_i^{-1}\circ \varphi_{i+1}^{\vphantom{1}}$:
\[
A_i=\left(\begin{array}{cc} \sqrt{x_{e_i}} & \sqrt{x_{e_i}} \\ 0 & \sqrt{x_{e_i}}^{-1} \end{array}\right) \ \text{in the case (a)},\quad A_i=\left(\begin{array}{cc} \sqrt{x_{e_i}} & 0 \\ \sqrt{x_{e_i}}^{-1} & \sqrt{x_{e_i}}^{-1} \end{array}\right)\ \text{in the case (b).}
\]

Let~$\widetilde{\gamma}$ be the loop in the unit tangent bundle~$U\Sigma$ corresponding to~$\gamma$ and~$\widetilde{\gamma}(e_i,e_{i+1})$ be its segment corresponding to~$\gamma(e_i,e_{i+1})$. From now onwards we fix the uniformization~$f:=\varphi_0:\mathbb H\to \overline{\Sigma}$. According to the definition~\eqref{eq:tilde_rho_def} we have
\[
\widetilde\rho_f(\widetilde{\gamma})\ =\ ((B_{m-1}^{-1}A_{m-1}B_{m-1})\cdot \ldots \cdot (B_1^{-1}A_1B_1) \cdot A_0)^{-1}\ =\ B_{m-1},
\]
where the matrices~$B_i:=(A_0A_1\dots A_{i-1})^{-1}\in\slr$ represent the M\"obius automorphisms~$\varphi_i^{-1}\circ \varphi_0^{\vphantom{1}}$. Since~$\Tr B_{m-1}=\Tr B_{m-1}^{-1}$ we conclude that
\begin{equation}
\label{eq:PL-def}
\Tr(\widetilde\rho_f(\widetilde{\gamma})) = \Tr(A_0A_1\dots A_{m-1})
\end{equation}
is a Laurent polynomial in the variables~$(\sqrt{x_e})_{e\in\Ed}$ with positive integer coefficients. In particular we have proved (ii).

The argument given above defines a required Laurent polynomial~$P_L$ for each element~$L=L(\gamma)\in\sk(M,1)$ obtained from a minimal loop~$\gamma$. This definition can be then extended to the whole skein algebra~$\sk(M,1)$ by linearity and multiplicativity. Finally, since the mapping~$L\mapsto P_L((t_e)_{e\in\Ed})$ is an algebra homomorphism for each specification \mbox{$t_e:=\sqrt{x_e}$}, the same is also true for formal variables~$t_e$.
\end{proof}


\subsection{Positivity of the bracelet basis of the skein algebra $\bm{\sk(M,1)}$ and the estimate of Fock-Goncharov coefficients}\label{braceletAndPositivity}

The last ingredient we need is the following positivity result. Let~$\Gamma$ be a lamination consisting of~$m_1$ copies of a simple loop~$l_1$, $m_2$ copies of a simple loop~$l_2$, \dots, $m_k$ copies of a simple loop~$l_k$; we assume that \bl{the loops~$l_i$ are in distinct homotopy classes}. The {\it bracelet} of $\Gamma$ is the multi-curve
\begin{equation*}
  \mathfrak{b}(\Gamma) := l_1^{m_1}\cup \dots\cup l_k^{m_k},\qquad \Gamma= (m_1\cdot l_1)\cup \dots\cup (m_k\cdot l_k),
\end{equation*}
where~$l_i^{m_i}$ denotes the minimal loop that travels~$m_i$ times along the simple loop~$l_i$, while~$m_i\cdot l_i$ stands for the~$m_i$ copies of~$l_i$.

Recall that we denote by~$L(\C)$ the element of~$\sk(M,1)$ obtained from a multi-curve~$\C$. It is easy to see that, for each simple loop~$l$, one has
\[
L(l^{m+2})=L(l)L(l^{m+1})-2L(l^m),\quad m\ge 0.
\]
By induction in~$m$ this gives $L(l^m)=2T_m(\tfrac{1}{2}L(l))$, $m\ge 0$, and hence
\begin{equation}
\label{eq:L_bG=}
\textstyle L(\mathfrak{b}(\Gamma))=\prod_{i=1}^k (2T_{m_i}(\tfrac{1}{2}L(l_i))),\qquad \Gamma= (m_1\cdot l_1)\cup \dots\cup (m_k\cdot l_k).
\end{equation}
where~$T_m(x)=\cos (m\arccos x)$ is the $m$-th Chebyshev polynomial of the first kind. The elements~$L(\Gamma)$ form a basis of~$\sk(M,1)$ (e.g., this follows from Theorem~\ref{thm:FG}, Theorem~\ref{SkToFunctions} and Lemma~\ref{spinHomo}). Therefore, the elements~$L(\mathfrak{b}(\Gamma))$ also form a basis of~$\sk(M,1)$. In particular, for each multi-curve~$\C$ there exists a unique decomposition
\begin{equation}
\label{eq:SkeinViaBG}
\textstyle L(\C)=\sum_{\Gamma\,-\,\text{lamination}} \alpha_{\C\Gamma}L(\mathfrak{b}(\Gamma)).
\end{equation}
It was shown by D.~Thurston that~$L(\mathfrak{b}(\Gamma))$ is in fact a \emph{positive} basis of~$\sk(M,1)$:
\begin{theorem}\label{th:Thurston}
For each multi-curve $\C$ the coefficients $\alpha_{\C\Gamma}$ are sign coherent: either $\alpha_{\C\Gamma}\ge 0$ for all~$\Gamma$ or~$\alpha_{\C\Gamma}\le 0$ for all~$\Gamma$. If $\C$ is a minimal multi-curve, then all~$\alpha_{\C\Gamma}\ge 0$.
\end{theorem}
\begin{proof}
See~\cite[Theorem~2]{Thurston}.
\end{proof}

Recall that we denote by~$|\C|$ the minimal number of intersections of a multi-curve homotopic to~$\C$ with the edges of the triangulation~$\G$ and by~$\Gamma(\C)$ the corresponding lamination, see Section~\ref{fViagAndgViaf} for more details. The next lemma is a simple combination of the positivity results provided by Proposition~\ref{prop:linkToPolynomials}(ii) and Theorem~\ref{th:Thurston}.

\begin{lemma} \label{bracelet_coef_lemma}
  For each multi-curve $\C$ one has $\sum_{\Gamma - \mathrm{lamination}}|\alpha_{\C\Gamma}|\leq 2^{|\C|}$.
\end{lemma}

\begin{proof}
As each multi-curve is homotopic to a minimal one up to nugatory self-crossings and homotopically trivial components one can assume that~$\C$ is minimal. Let~$P_{L(\C)}$ be the Laurent polynomial given by Proposition~\ref{prop:linkToPolynomials} and~$\mathbf x=\mathbf 1\in \mathbb R_+^\Ed$ be the vector consisting of units: $x_e=1$ for all~$e\in\Ed$. It follows from the construction of~$P_{L(\C)}$ (see~\eqref{eq:PL-def}) that~$P_{L(\C)}(\mathbf 1) \leq 2^{|\C|}$. Applying the evaluation~$\Psi_{\mathbf 1}$ to the identity~\eqref{eq:SkeinViaBG} and noting that~$P_{L(\mathfrak{b}(\Gamma))}(\mathbf 1)\ge 1$ one obtains the desired estimate.
\end{proof}

\bl{We are now in the position to prove the main result of this section.} Recall that we are interested in an exponential upper bound for the coefficients of the Fock--Goncharov change of the bases
\[
\textstyle g_\Gamma\ =\ \sum_{\Delta\le\Gamma}c_{\Gamma\Delta}f_\Delta.
\]

\begin{theorem}\label{thm:diagEst} For each pair of laminations~$\Gamma,\Delta$ one has~$|c_{\Gamma\Delta}|\le 4^{|\Gamma|}$.
\end{theorem}

\begin{proof}
Recall that the function~$f_\Delta\in\F(\X)^{\sl}$ corresponds to~$\pm L(\Delta)\in\sk(M,1)$ under the isomorphisms provided by Theorem~\ref{SkToFunctions} and Lemma~\ref{spinHomo}. Due to Lemma~\ref{gViaLoops}, the function~$g_{\Gamma}$ correspond to the averaging of the signed knots~$\pm L(\C(\Gamma,\pi))$ over permutations~$\pi\in \prod_{\sigma\in\Fa,\,e\in\partial\sigma}\Sym(n_e)$. Let
\[
L(\C)\ =\ \sum\nolimits_{\Delta - \text{lamination},\,\Delta\le\Gamma} \beta_{\C\Delta}L(\Delta),\qquad \C=\C(\Gamma,\pi).
\]
Then, $c_{\Gamma\Delta}$ is equal to the average of the signed coefficients~$\beta_{\C\Delta}$ over the permutations~$\pi$ and it is enough to prove that~$|\beta_{\C\Delta}|\le 4^{|\C|}=4^{|\Gamma|}$ for each multi-curve~$\C=\C(\Gamma,\pi)$.

Using the bracelet basis of~$\sk(M,1)$ discussed above one writes
\[
L(\C)\ =\ \sum\nolimits_{\Gamma' - \text{lamination},\,\Gamma'\le\Gamma} \alpha_{\C\Gamma'}L(\mathfrak{b}(\Gamma')).
\]
Now one can expand each of the bracelet links~$L(\mathfrak{b}(\Gamma'))\in\sk(M,1)$ via the functions~$L(\Delta)$ with~$\Delta\subset\Gamma'$ using~\eqref{eq:L_bG=}:
\[
L(\mathfrak{b}(\Gamma'))=\sum_{\Delta\subset\Gamma'}\alpha'_{\Gamma'\Delta}L(\Delta),
\]
where each coefficient~$\alpha'_{\Gamma'\Delta}$ is the product over homotopically distinct loops of~$\Gamma'$ of some coefficients of the Chebyshev polynomials (corresponding to these loops)
\[
2T_m(\tfrac{1}{2}x)\ =\ \sum_{k=0}^{\lfloor \frac{m}{2}\rfloor} (-1)^k\cdot \frac{m}{m-k}\binom{m-k}{k}x^{m-2k}.
\]
Using the trivial upper bound~$2^{m}$ for these coefficients one obtains the crude estimate \mbox{$|\alpha'_{\Gamma'\Delta}|\le 2^{|\Gamma'|}\le 2^{|\Gamma|}$}. Therefore,
\[
|\beta_{\C\Delta}|\ \le\ 2^{|\Gamma|}\cdot\!\! \sum_{\Gamma': \Delta\subset\Gamma'\le \Gamma}|\alpha_{\C\Gamma'}|\ \le\ 4^{|\Gamma|}
\]
due to Lemma~\ref{bracelet_coef_lemma}.
\end{proof}

\begin{rem} It is worth noting that the constant~$4$ in the exponential upper bound provided by Theorem~\ref{thm:diagEst} is far from being optimal. For instance one can easily improve the constant~$2$ in Lemma~\ref{bracelet_coef_lemma} to~$\frac{1}{2}(\sqrt{5}+1)$, which is equal to the~$L^2$-norm of the matrices~$A_i$ with~$x_e=1$, not speaking about a huge overkill in the estimate of coefficients~$\alpha'_{\Gamma'\Delta}$ used above.
\end{rem}




\section{Proof of the main results}\label{complex_analisys_section} This section is organized as follows. In Section~\ref{sub:Expansions_in_BR} we discuss the expansions of~$\su^\Fa$-invariant functions on poly-balls~$\mathbb B_R\subset \sl^\Ed$ via the functions~$F_{\Gamma,\mathbf m}$ constructed in Section~\ref{representation_theory_section} (see \eqref{eq:G_Gamma_m-def}). Note that the exponential estimates provided by Proposition~\ref{prop:lower-bound} and Theorem~\ref{thm:diagEst} play a crucial role in the proof of Lemma~\ref{lem:estimate-of-coeff}. We then pass to the analysis of holomorphic functions on the variety
\begin{equation}
\label{eq:Zunip-def}
Z_\mathrm{unip}:=\{\,\mathbf A\in \sl^{\Ed}\,\mid\,\Tr(\h(\mathbf A,\lloop{\lambda_i}))=2~\text{for all}~i=1,\dots,n\,\}
\end{equation}
(here and below, the holomorphicity condition is always required on the regular part of the variety only). Section~\ref{sub:HolExt} is devoted to extensions of holomorphic functions from~$Z_{\mathrm{unip}}\cap\mathbb B_R$ to~$\mathbb B_R$ and is based upon Manivel's Ohsawa--Takegoshi type theorem, it is needed for the existence part of Theorem~\ref{main_theorem_actually}. Section~\ref{sub:Nullstellensatz} contains a version of Hilbert's Nullstellensatz for holomorphic functions vanishing on~$Z_{\mathrm{unip}}\cap\mathbb B_R$, we need this for the uniqueness part of Theorem~\ref{main_theorem_actually}. Finally, Section~\ref{sub:Expansions_macro} is devoted to the proof of our main result, Theorem~\ref{main_theorem_actually}, on the existence and uniqueness of expansions of holomorphic functions on~$X_\mathrm{unip}$ in the basis~$f_\Gamma$ indexed by macroscopic laminations~$\Gamma$ (i.e., those not containing any of the loops $\lloop{\lambda_i}$ surrounding only one puncture).

\subsection{Expansions of~$\bm{\su^\Fa}$-invariant holomorphic functions on~$\bm{\mathbb{B}_R}$ via~$\bm{F_{\Gamma,\mathbf{m}}}$.}\label{sub:Expansions_in_BR}

Recall that the action of the group~$\sl^{\Fa}$ on the space~$(\mathbb C^{2\times 2})^\Ed$ is given by \eqref{eq:actionRule}.
We start with a simple preliminary lemma.

\begin{lemma}\label{averaging}
  Let a holomorphic function~$F:(\mathbb C^{2\times 2})^\Ed\to\mathbb C$ be invariant under the action of~$\su^\Fa\subset \sl^\Fa$. Then,~$F$ is also invariant under the action of~$\sl^{\Fa}$.
\end{lemma}

\begin{proof}
Given~$\mathbf A\in (\mathbb C^{2\times 2})^{\Ed}$ define a holomorphic function~$\varphi_\mathbf A$ on $\sl^\Fa$ by
\[
\varphi_{\mathbf A}(\mathbf C):=F(\mathbf C[\mathbf A]),
\]
By our assumption, this function is constant on~$\su^\Fa$. Therefore, the $d$-th derivative of~$\varphi_{\mathbf A}$ at the identity is a~$\mathbb C$-linear functional on the space~$\Sym^d\,(\mathfrak{sl}_2(\mathbb C))^{\Fa}$ which vanishes on the ($\mathbb R$-linear) subspace~$\Sym^d(\mathfrak{su}_2(\mathbb C))^{\Fa}$. Since $\mathfrak{sl}_2(\mathbb C)\simeq \mathbb C\otimes_{\mathbb R} \mathfrak{su}_2(\mathbb{C})$ we conclude that all these derivatives vanish. Therefore,~$\varphi_{\mathbf A}$ is constant on~$\sl^{\Fa}$.
\end{proof}

\begin{lemma}
\label{lem:Taylor-in-F}
(i) Let coefficients~$p_{\Gamma,\mathbf m}$ (indexed by laminations~$\Gamma$ and tuples~$\mathbf m\in \mathbb Z_{\ge 0}^\Ed$) satisfy the estimate~$|p_{\Gamma,\mathbf m}|=O(r^{-(|\Gamma|+2|\mathbf m|+4|\Ed|)})$ for some~$r>0$. Then,
the series
\begin{equation}
\label{eq:pFgn-conv}
  F(\mathbf{A}) := \sum\limits_{\Gamma,\mathbf m}p_{\Gamma,\mathbf m}F_{\Gamma, \mathbf m}(\mathbf{A})
\end{equation}
converges absolutely and uniformly on compact subsets of~$\mathbb B_r$.

\noindent (ii) Moreover, if both~$p_{\Gamma,\mathbf m},\tilde{p}_{\Gamma,\mathbf m}$ satisfy the exponential upper bound given above and the corresponding functions~$F,\widetilde{F}$ coincide on~$\mathbb B_{r}$, then~$p_{\Gamma,\mathbf m}=\tilde{p}_{\Gamma,\mathbf m}$ for all~$\Gamma$ and~$\mathbf m$.
\end{lemma}

\begin{proof}
(i) Using simple inequalities~$|\Tr(A_1A_2)|^2\le (\Tr(A_1A_1^*))^{1/2}(\Tr(A_2A_2^*))^{1/2}$, where $A_1,A_2\in\mathbb C^{2\times 2}$, and $|\Tr(B_1B_2)|\le \Tr B_1\Tr B_2$, where $B_1=B_1^*\ge 0$ and $B_2=B_2^*\ge 0$, one easily sees that
\[
\textstyle |\Tr(A_1\dots A_n)|\le \prod_{i=1}^n(\Tr(A_iA_i^*))^{1/2}\quad \text{for all}~n\ge 2~\text{and}~A_1,\dots A_n\in\mathbb C^{2\times 2}.
\]
Together with the trivial bound~$|\det A|\le \Tr(AA^*)$ this yields the following estimate:
\[
|F_{\Gamma,\mathbf m}(\mathbf A)|\le r^{|\Gamma|+2\mathbf m}\quad \text{for~all}~\mathbf A\in \overline{\mathbb B}_r.
\]
Therefore, the series~\eqref{eq:pFgn-conv} converges absolutely and uniformly on each closed poly-ball~$\overline{\mathbb B}_{r'}$ with~$r'<r$ since $\sum_{\Gamma,\mathbf m} (r'/r)^{|\Gamma|+2|\mathbf m|} \le (1\!-\!r'/r)^{-|\Ed|}\cdot(1\!-\!(r'/r)^2)^{-|\Ed|}<\infty$.

\noindent (ii) Grouping the terms in~\eqref{eq:pFgn-conv} according to their degree of homogeneity~$d=|\Gamma|+2|\mathbf m|$, one sees that each term in the series
\begin{equation*}
  \textstyle \tilde{F}(\varepsilon \mathbf A)-F(\varepsilon \mathbf A)\ =\ \sum\nolimits_{d\ge 0}\varepsilon^dP^{(d)}(\mathbf A)
\end{equation*}
must vanish:
\begin{equation*}
\textstyle P^{(d)}(\mathbf A):=\sum\nolimits_{\Gamma, \mathbf m\,:\, |\Gamma| + 2|\mathbf m| = d}\,(\tilde{p}_{\Gamma,\mathbf m}-p_{\Gamma,\mathbf m})F_{\Gamma, \mathbf m}(\mathbf A)=0
\end{equation*}
for all $\mathbf A\in (\mathbb C^{2\times 2})^\Ed$. Due to Lemma~\ref{GisBasis} and Lemma~\ref{lem:GGmViaFDn}, the polynomials~$F_{\Gamma,\mathbf {m}}$ are linearly independent. Hence all the coefficients~$\tilde{p}_{\Gamma,\mathbf m}-p_{\Gamma,\mathbf m}$ must vanish.
\end{proof}

The next lemma shows that each $\su^\Fa$-invariant holomorphic function defined on a poly-ball~$\mathbb B_R$ necessarily admits an expansion~\eqref{eq:pFgn-conv} in a \emph{smaller} poly-ball~$\mathbb B_{\frac{1}{5}\eta_0R}$. Recall that the absolute constant~$\eta_0$ is introduced in Proposition~\ref{prop:lower-bound}.

\begin{lemma}\label{lem:estimate-of-coeff}
Given~$r>0$, denote~$R:=5\eta_0^{-1}r$. Let~$F\in L^2(\mathbb B_R)$ be an $\su^{\Fa}$-invariant holomorphic function on~$\mathbb B_R$. Then, there exist coefficients~$p_{\Gamma,\mathbf m}$ indexed by laminations~$\Gamma$ and tuples~$\mathbf m\in \mathbb{Z}_{\ge 0}^\Ed$ such that the following holds:
\begin{equation}
\label{eq:pGn-bound}
  |p_{\Gamma, \mathbf m}|\ \leq\ r^{-(|\Gamma| + 2|\mathbf m| + 4|\Ed|)}\cdot \|F\|_{L^2(\mathbb B_R)}
\end{equation}
for all~$\Gamma,\mathbf m$, and
\begin{equation}
\label{eq:FinB15}
  F(\mathbf{A}) = \sum\limits_{\Gamma,\mathbf m}p_{\Gamma,\mathbf m}F_{\Gamma, \mathbf m}(\mathbf{A})\quad \text{for~all}\ \ \mathbf{A}\in \mathbb{B}_{r}.
\end{equation}
\end{lemma}

\begin{rem} Recall that Lemma~\ref{lem:Taylor-in-F}(ii) guarantees the uniqueness of the expansion~\eqref{eq:FinB15} provided that~$p_{\Gamma,\mathbf m}=O_{|\Gamma|+2|\mathbf m|\to\infty}(r^{-(|\Gamma| + 2|\mathbf m|)})$.
\end{rem}

\begin{proof} Let~$F(\mathbf A) = \sum_{d\geq 0} F^{(d)}(\mathbf A)$ be the Taylor expansion of the function~$F$ at the origin. Since this expansion is unique, the homogeneous polynomials~$F^{(d)}$ of degree~$d$ are~$\su^\Fa$-invariant. Lemma~\ref{averaging} implies that these polynomials are also~$\sl^\Fa$-invariant so one can expand them in the basis~$G_{\Delta,\mathbf k}$ with~$|\Delta|+2|\mathbf k|=d$ due to Lemma~\ref{GisBasis}. More precisely, Lemma~\ref{GareOrt} yields
\begin{equation*}
  F^{(d)}\ = \sum\limits_{\Delta, \mathbf k\,:\,|\Delta| + 2|\mathbf k| = d} q_{\Delta, \mathbf k} G_{\Delta, \mathbf k},\quad \text{where}\quad q_{\Delta,\mathbf k}= \frac{\lan F, G_{\Delta, \mathbf k} \ran_{L^2(\mathbb B_R)}}{\|G_{\Delta, \mathbf k}\|^2_{L^2(\mathbb B_R)}}.
\end{equation*}
Due to Proposition~\ref{prop:lower-bound} we have the following estimate:
\[
|q_{\Delta,\mathbf k}|\le \|F\|_{L^2(\mathbb B_R)}^{\vphantom{-1}}\cdot \|G\|^{-1}_{L^2(\mathbb B_R)} \le (\eta_0R)^{-(d+4|\Ed|)}\|F\|_{L^2(\mathbb B_R)}.
\]
We now use the expansion~$G_{\Delta,\mathbf k} = \sum_{\Gamma\le\Delta} c_{\Delta\Gamma} F_{\Gamma, \mathbf k+\frac{1}{2}(\mathbf n(\Delta)-\mathbf n(\Gamma))}$ provided by Lemma~\ref{lem:GGmViaFDn}, which leads to the following identity:
\begin{equation*}
  F^{(d)}\ = \sum\limits_{\Gamma, \mathbf m\,:\,|\Gamma| + 2|\mathbf m| = d} p_{\Gamma, \mathbf m} F_{\Gamma, \mathbf m},\quad \text{where}\quad
  p_{\Gamma, \mathbf m}\ := \sum_{\Delta\,:\,\Gamma \leq \Delta\leq \Gamma + 2\mathbf m}c_{\Delta\Gamma}q_{\Delta, \mathbf m - \frac{1}{2}(\mathbf n(\Delta) - \mathbf n(\Gamma))}.
\end{equation*}
One can now easily deduce~\eqref{eq:pGn-bound} from the estimate~$|c_{\Delta\Gamma}|\le 4^{|\Delta|}$ provided by Theorem~\ref{thm:diagEst}, the estimate of coefficients~$q_{\Delta,\mathbf k}$ given above and the following crude bound:
\[
\textstyle \sum_{\Delta:|\Delta|\le d}4^{|\Delta|}\ \le\ \sum_{n=0}^d \binom{n+|\Ed|-1}{n}4^n\ \le\ 5^{d+|\Ed|}.
\]
The exponential upper bound~\eqref{eq:pGn-bound} and Lemma~\ref{lem:Taylor-in-F}(i) immediately imply~\eqref{eq:FinB15}.
\end{proof}

\subsection{Holomorphic extensions from~$\bm{Z_\mathrm{unip}\cap\mathbb{B}_R}$ to~$\bm{\mathbb{B}_R}$} \label{sub:HolExt}
\bl{Recall that our ultimate goal is to show that each holomorphic $\sl$-invariant function on $X_{\mathrm{unip}}$ can be expanded (i.e., written as the sum of an infinite series) in functions $f_\Gamma$ enumerated by macroscopic laminations. The variety $X_{\mathrm{unip}}$ is contained in the variety $X$, which is related to the affine variety $\sl^{\Ed}\subset (\mathbb C^{2\times 2})^{\Ed}$ via the map $\phi:\mathbf A\mapsto \h(\mathbf A,\cdot)$ defined by~\eqref{eq:triangularMap}.

We already proved in Lemma~\ref{lem:estimate-of-coeff} that each $\su^{\Fa}$-invariant holomorphic function $F$ defined in a poly-ball $\mathbb B_R\subset (\mathbb C^{2\times 2})^{\Ed}$ can be written as the sum of a series of functions $F_{\Gamma,\mathbf m}$ (whose restrictions to $\sl^\Ed$ are equal to $\phi^*f_\Gamma$). However, functions that we want to expand are defined only on $X_{\mathrm{unip}}\subsetneq X$. If we pull such a function back along the map $\phi$, we get a function $f$ defined on the \emph{subvariety} $Z_{\mathrm{unip}}\subset (\mathbb C^{2\times 2})^{\Ed}$ only; see~\eqref{eq:Zunip-def}. Thus, a reasonable way to use Lemma~\ref{lem:estimate-of-coeff} is to solve an interpolation problem that consists in finding a function $F:\mathbb B_R\to \mathbb C$ such that $F\vert_{Z_{\mathrm{unip}}\cap \mathbb B_R} = f$. Once such a function~$F$ is found, we can use Lemma~\ref{lem:estimate-of-coeff} to expand $F$ as a series of $F_{\Gamma,\mathbf m}$ and then restrict this expansion to $Z_{\mathrm{unip}}$ to obtain an expansion of the initial function~$f$.

Certainly, we also need to control the norm of the solution $F$ as Theorem~\ref{thm:main-thm} requires an explicit estimate on the coefficients of the series. Let us point out that, even without requiring such a quantitative control (and even locally), this interpolation problem could, a priori, be not always solvable since the variety $Z_\mathrm{unip}$ is \emph{not smooth}. Luckily, there exists a series of results -- so-called Ohsawa--Takegoshi-type theorems (we refer the interested reader to a very helpful introduction to this subject due to Demailly~\cite{demailly}) -- which provide an affirmative answer to the question of finding an extension of a holomorphic function from a (non-smooth) subvariety. Roughly speaking, the interpolation problem is always solvable provided that the ambient variety $\mathbb B_R$ is a Stein manifold and that the function $f$ belongs to a certain weighted $L^2$-space on $Z_\mathrm{unip}\cap\mathbb B_R$. Moreover, in this case the extended function $F$ can be also taken from a certain weighted $L^2$-space that naturally arises from the potential theory arguments.

It is worth noting that Ohsawa--Takegoshi-type theorems are usually formulated in a more general context of Hermitian bundles over Stein manifolds. In this context, $f$ should be viewed as a section of a Hermitian line bundle $L$ on $\mathbb B_R$ while $Z_{\mathrm{unip}}$ is the zero set of a section $s$ of another Hermitian vector bundle $E$; the required estimates involve curvature forms of $L$ and $E$. In what follows, we rely upon a very particular case of Manivel's Ohsawa--Takegoshi type theorem~\cite{Manivel} and do not need such generality: both $L$ and $E$ will be trivial vector bundles.

Given a Hermitian manifold $B$ and a holomorphic map $s: B\to \mathbb C^r$ we denote by $TB$ the holomorphic tangent bundle of~$B$, by
$\Lambda^mds:\Lambda^m(TB)\to \Lambda^m\mathbb C^r$ the corresponding morphism of $m$-th exterior powers, and by~$|\Lambda^mds|$ its operator norm.}
\begin{theorem}\label{Demailly_BR}
Let~$k=4|\Ed|$,~$\mathbb B_R \subset (\mathbb C^{2\times 2})^{\Ed}\cong \mathbb C^k$ be a poly-ball defined by~\eqref{eq:B_Rdef}, and~$s:\mathbb B_R\to \mathbb C^r$ be a holomorphic map such that~$|s(\mathbf A)|\le e^{-1}$ for all~$\mathbf A\in \mathbb B_R$. Let~$Z:=s^{-1}(0)$ and assume that~$ds(\mathbf A)$ is of maximal rank for a generic point~$\mathbf A\in Z$. Then, each holomorphic function~$f:Z^{\mathrm{reg}}:=Z\setminus Z^{\mathrm{sing}}\to\mathbb C$ such that
\[
\int_Z |f|^2|\Lambda^rds|^{-2}\,\omega^{k-r} < \infty,
\]
admits a holomorphic extension~$F:\mathbb B_R\to \mathbb C$ such that~$F|_{Z^{\mathrm{reg}}}=f$ and
\[
\int_{\mathbb B_R}|F|^2 (-|s|^r\log|s|)^{-2}\,\omega^k\ \leq\ C_{r,k}\int_{Z}|f|^2|\Lambda^rds|^{-2}\,\omega^{k-r},
\]
where~$\omega=\frac{i}{2}\sum_{e\in \Ed}\Tr(dA_e\wedge dA^*_e)$ and the constant~$C_{r,k}>0$ depends only on~$r$ and~$k$.
\end{theorem}
\begin{proof} See~\cite[p.~54,~Theorem~4.1]{demailly}, note that the poly-ball~$\mathbb B_R$ is a weakly pseudo-convex domain. Since in our case the holomorphic vector bundles $L,E$ are trivial with a constant Hermitian metric, the curvature tensors~$\Theta(L),\Theta(E)$ vanish and the condition (a) of~\cite[p.~54,~Theorem~4.1]{demailly} boils down to the trivial inequality
\[
\partial \overline{\partial} \log |s|^2 = |\partial s|^2\cdot |s|^{-2}- |\langle\partial s,s\rangle|^2\cdot |s|^{-4} \ge 0,
\]
which holds true for all one-dimensional restrictions~$z\mapsto s(\mathbf A+z\mathbf B)$ of~$s$.
\end{proof}

We now apply Theorem~\ref{Demailly_BR} to the variety~$Z_\mathrm{unip}\subset (\mathbb C^{2\times 2})^\Ed$ defined by~\eqref{eq:Zunip-def}. For our purposes it is enough to assume that the function~$f$ is bounded, we also do not need the sharp weight~$(-|s|^r\log|s|)^{-2}\ge 1$ in the $L^2$-norm of its extension $F$.

\begin{prop}\label{appliedDemailly} Let~$R>\sqrt{2}$ and~$f:Z_{\mathrm{unip}}^{\mathrm{reg}}\cap \mathbb B_R\to \mathbb C$ be a bounded holomorphic function. Then, there exists a holomorphic function $F: \mathbb B_R\to \mathbb C$ such that
\mbox{$F|_{Z_{\mathrm{unip}}^{\mathrm{reg}}\cap \mathbb B_R} = f$} and
\begin{equation*}
  \|F\|_{L^2(\mathbb B_R)}\ \leq\ \mathrm{const}(R,\G)\cdot \|f\|_{L^{\infty}(Z_{\mathrm{unip}}^{\mathrm{reg}}\cap \mathbb B_R)},
\end{equation*}
for some constant depending on $R$ and the triangulation $\G$ but independent of~$f$.
\end{prop}

\begin{proof} Denote~$s:=((F_{\varnothing, \mathbf n^e}-1)_{e\in \Ed},(F_{\lloop{\lambda_i}}-2)_{i=1}^n):(\mathbb C^{2\times 2})^\Ed \to \mathbb C^{|\Ed|+n}$ and fix a small positive constant $c(R,\G)$ so that~$c(R,\G)\cdot |s|\le e^{-1}$ on~$\mathbb B_R$. Due to Theorem~\ref{Demailly_BR}, the function~$f$ admits a holomorphic extension~$F$ from~$Z_{\mathrm{unip}}^{\mathrm{reg}}\cap\mathbb B_R$ to~$\mathbb B_R$ such that
\[
\|F\|_{L^2(\mathbb B_R)}^2\ \leq\ C_{|\Ed|+n,4|\Ed|}\cdot c(R,\G)^{-2(|\Ed|+n)}\cdot \int_{Z_{\mathrm{unip}}\cap \mathbb B_R}|\Lambda^r ds|^{-2}\omega^{k-r}\cdot \|f\|_{L^{\infty}(Z_{\mathrm{unip}}^{\mathrm{reg}}\cap \mathbb B_R)}^2.
\]
Therefore, it is enough to check that~$|\Lambda^rds|^{-2}\in L^1_{\mathrm{loc}}(Z_{\mathrm{unip}})$.

Let~$T\subset\Ed$, $|T|=n$, be a spanning tree (on $n+1$ vertices~$\partial\Omega,\lambda_1,\dots,\lambda_n$) of the triangulation~$\G$ such that the graph $T\setminus\{\lambda_k,\ldots,\lambda_n\}$ is connected for all~$k=n,\dots,1$. Note that the mapping
\begin{equation}
\label{eq:new_coordinates}
\mathbf A=(A_e)_{e\in\Ed}\ \mapsto\ ((A_e)_{e\in\Ed\setminus T},(B_i)_{i=1}^n),\qquad B_i:=\h^\vee(\mathbf A, \lloop{\lambda_i}),
\end{equation}
is a smooth bijection in a vicinity of~$\sl^\Ed\cong \sl^{(\Ed\setminus T)\cup\{1,...,n\}}$: one can iteratively reconstruct all the missing matrices~$(A_e)_{e\in T}$ from the holonomies~$B_i$, $i=n,\dots,1$. 
To compute~$\Lambda^r ds$, we can view the mapping~$s$ as acting coordinate-wise in the new coordinates $((A_e)_{e\in\Ed\setminus T},(B_i)_{i=1}^n)$:
\[
A_e\mapsto \det A_e-1,\qquad B_i\mapsto s_0(B_i):=(\det B_i-1,\Tr B_i-2).
\]
(More accurately, one should multiply~$\det B_i$ by several factors~$(\det A_{e})^{-1}$ corresponding to the edges incident to~$\lambda_i$ and possibly by factors~$(\det B_j)^{-1}$, $j>i$, coming from earlier steps of the reconstruction of~$(A_e)_{e\in T}$, but all these additional factors do not affect~$\Lambda^r ds$ on~$\sl^\Ed$). As the gradient of~$A\mapsto \det A$ does not vanish on~$\sl$ there is nothing to check for the first coordinates~$(A_e)_{e\in\Ed\setminus T}$. For the coordinates~$(B_i)_{i=1}^n$, the only degeneracy of~$ds_0\wedge ds_0$ on~$s_0^{-1}(0)$ is at~$B=\Id$. Writing
\[
B = \begin{pmatrix} t - ix & y+iz \\ -y+iz &  t+ix \end{pmatrix},\qquad t,x,y,z\in\mathbb C,
\]
$\omega = \frac{i}{4}(dt\wedge d\overline{t}+dx\wedge d\overline{x} + dy\wedge d\overline{y} + dz\wedge d\overline{z})$, it remains to check that
\[
    \int\nolimits_{{t=1,\ x^2+y^2+z^2=0,\ |x|,|y|,|z|\le 1}}\ \frac{\omega\wedge \omega}{|x|^2 + |y|^2 + |z|^2}\ <\ \infty,
\]
which is straightforward.
\end{proof}

\begin{rem}
\label{rem:appliedDemailly}
It is easy to see that Proposition~\ref{appliedDemailly} remains true (with exactly the same proof) if one replaces the variety~$Z_\mathrm{unip}$ by
\begin{equation}
\label{eq:Zk-def}
Z_k:=\{\mathbf A\in \sl^{\Ed}\,\mid\,\Tr(\h(\mathbf A,\lloop{\lambda_i}))=2~\text{for all}~i=1,\dots,k\}
\end{equation}
(note that~$Z_{\mathrm{unip}}=Z_n\subset Z_{n-1}\subset \dots \subset Z_1\subset Z_0=\sl^\Ed$).
\end{rem}

\subsection{Nullstellensatz for holomorphic functions vanishing on~$\bm{Z_\mathrm{unip}\cap\mathbb{B}_R}$} \label{sub:Nullstellensatz} The goal of this section is to prove the following analogue of Hilbert's Nullstellensatz for holomorphic functions on~$\mathbb B_R$.
\begin{prop}\label{explicitNullstellensatz}
Let $R>R'>\sqrt{2}$ and a holomorphic function $F: \mathbb B_R\to \mathbb C$ vanishes on the set $Z_\mathrm{unip}\cap \mathbb B_R$. Then there exist bounded holomorphic functions $H_e:\mathbb B_{R'}\to \mathbb C$, $e\in \Ed$, and $H_i: \mathbb B_{R'}\to \mathbb C$, $i = 1,\dots, n$, such that
\begin{equation*}
  F(\mathbf A) = \sum\limits_{e\in \Ed}H_e(\mathbf A)(F_{\varnothing, \mathbf n^e}(\mathbf A) - 1) + \sum\limits_{i = 1}^nH_i(\mathbf A)(F_{\lloop{\lambda_i}}(\mathbf A) - 2)\quad \text{for all}\ \ \mathbf A\in\mathbb B_{R'}.
\end{equation*}
\end{prop}

\begin{rem} (i) It is worth noting that, if~$F$ were a polynomial, then the result would follow from the classical Hilbert's Nullstellensatz as one can easily check that the radical of the ideal~$\CMcal{I}=\langle(\det A_e -1)_{e\in\Ed},(\Tr \h^\vee (\mathbf A,\lloop{\lambda_i})-2)_{i=1}^n\rangle$ coincides with~$\CMcal{I}$.

\noindent (ii) Although one can prove Proposition~\ref{explicitNullstellensatz} by means of algebraic geometry, below we prefer to take advantage of the analytic tools already introduced in the previous section
(namely, Proposition~\ref{appliedDemailly} and Remark~\ref{rem:appliedDemailly}) and the following simple lemma.
\end{rem}

\begin{lemma}
\label{lem:divison_is_OK}
Let~$t(B):=\Tr B -2$,~$Y:=\{\,B\in\sl\,\mid\,t(B)=0\,\}$, and $V\subset\sl$ be an open subset. If a holomorphic function~$f:V\to\mathbb C$ vanishes on $Y\cap V$, then the ratio~$f/t$ is a locally bounded holomorphic function on~$V$.
\end{lemma}

\begin{proof}
If $\Id\notin V$, this fact easily follows from the local analysis since $Y\cap V$ is smooth in $V$ and $dt$ does not vanish there. In the opposite case, let~$U\Subset V$ be a small neighborhood of~$\Id\in V$. The same local analysis implies that the ratio~$f/t$ can be viewed as a holomorphic function on~$V\setminus U$ and the Hartogs extension theorem guarantees the existence of its holomorphic extension in the whole~$V$.
\end{proof}

\begin{proof}[Proof of Proposition~\ref{explicitNullstellensatz}] Let $Z_{\mathrm{unip}}=Z_n\subset Z_{n-1}\subset \ldots \subset Z_1\subset Z_0=\sl^\Ed$ be given by~\eqref{eq:Zk-def} and~$B_i:=\h(\mathbf A,\lloop{\lambda_i})$ as in the proof of Proposition~\ref{appliedDemailly}; recall that each~$Z_k$ is cut off by the equations~$t_1=\ldots=t_k=0$, where~$t_i:=\Tr B_i-2$.

Consider the meromorphic function~$F/t_n$ on~$Z_{n-1}\cap \mathbb B_R$ and recall the coordinates~\eqref{eq:new_coordinates} from the proof of Proposition~\ref{appliedDemailly}. If one fixes all~$(A_e)_{e\in\Ed\setminus T}$ and~$(B_i)_{i=1}^{n-1}$, then Lemma~\ref{lem:divison_is_OK} guarantees that the ratio~$F/t_n$ is a holomorphic function of the last coordinate~$B_n$. Therefore, there exists a holomorphic function~$h_n:Z_{n-1}\cap\mathbb B_R\to \mathbb C$ such that~$F=t_n\cdot h_n$ on~$Z_{n-1}\cap\mathbb B_R$. Due to Remark~\ref{rem:appliedDemailly}, for each~$R_{n-1}<R$ one can find a holomorphic extension~$H_n:\mathbb B_{R_{n-1}}\to\mathbb C$ of the (bounded on~$Z_{n-1}\cap B_{R_{n-1}}$) function~$h_n$. As~$t(B_n)=\Tr\h(\mathbf A,\lloop{\lambda_n})-2$ on~$\sl^\Ed$, one has
\[
F_{n-1}(\mathbf A):=F(\mathbf A)-H_n(\mathbf A)(F_{\lloop{\lambda_n}}(\mathbf A) -2)=0\quad \text{for all}\ \ \mathbf A\in Z_{n-1}\cap\mathbb B_{R_{n-1}}
\]
and one can iterate this procedure considering the ratio~$F_{n-1}/t_{n-1}$ on~$Z_{n-2}\cap \mathbb B_{R_{n-2}}$ with~$R_{n-2}<R_{n-1}<R_n$, etc. After~$n$ steps, one obtains the existence of holomorphic functions~$H_i:\mathbb B_{R_0}\to\mathbb C$ such that
\[
F_0(\mathbf A):=F(\mathbf A)-\sum_{i=1}^nH_i(\mathbf A)(F_{\lloop{\lambda_i}}(\mathbf A) -2)=0\quad \text{for all}\ \ \mathbf A\in \sl^\Ed\cap\mathbb B_{R_0},
\]
where~$R_0$ can be chosen so that~$R'<R_0<R$.

In order to complete the proof and to construct the required functions~$(H_e)_{e\in\Ed}$ in addition to~$(H_i)_{i=1}^n$ one simply repeats the same arguments for the collection of functions $(\det A_e-1)_{e\in\Ed}$ instead of~$(t_i)_{i=1}^n$. Since the manifold~$\sl^\Ed\subset (\mathbb C^{2\times 2})^\Ed$ is smooth, there is even no need in an analogue of Lemma~\ref{lem:divison_is_OK} along this procedure.
\end{proof}

\subsection{Expansions of holomorphic functions on~$\bm{X_{\mathrm{unip}}}$ via~$\bm{f_\Gamma}$, $\bm{\Gamma}$ -- macroscopic} \label{sub:Expansions_macro} \bl{Recall that the mapping~$\phi:\mathbf A\mapsto \h(\mathbf A,\cdot)$} is given by~\eqref{eq:triangularMap} and let
\[
\mathbb D_R:=\phi(Z_{\mathrm{unip}}\cap {\mathbb B}_R)\subset X_{\mathrm{unip}}\,.
\]

\begin{theorem}\label{main_theorem_actually}
Given~$r>\sqrt{2}$, denote~$R:=5\eta_0^{-1}r$. Let~$f:\overline{\mathbb D}_R\to\mathbb C$ be a bounded holomorphic function on~$\mathbb D_R$. Then, there exist coefficients~$p_\Gamma$ indexed by macroscopic laminations~$\Gamma$ such that the following holds:
\begin{equation}
\label{eq:PG-bound}
|p_\Gamma|\ \le\ r^{-|\Gamma|}\cdot \mathrm{const}(r,\G)\cdot \|f\|_{L^{\infty}(\mathbb D_R)}
\end{equation}
for all~$\Gamma$ and
\begin{equation}
\label{eq:pFg-conv}
  f(\rho) \ =\!\! \sum_{\Gamma~-~\mathrm{macroscopic}} p_{\Gamma}f_{\Gamma}(\rho)\quad \text{for all}\ \ \rho\in \mathbb D_r.
\end{equation}
Moreover, this expansion is unique provided that~$\frac{1}{5}\eta_0r>\sqrt{2}$ and~$p_\Gamma=O_{|\Gamma|\to\infty}(r^{-|\Gamma|})$.
\end{theorem}

\begin{proof}
To prove the existence of coefficients~$p_\Gamma$ consider a holomorphic function~$F:\mathbb B_R\to\mathbb C$ obtained by applying Proposition~\ref{appliedDemailly} to the function $f\circ\phi:Z_{\mathrm{unip}}\cap \mathbb B_R\to\mathbb C$. Recall that one has
\[
\|F\|_{L^2(\mathbb B_R)}\ \le\ \mathrm{const}(R,\G)\cdot \|f\|_{L^{\infty}(\mathbb D_R)}.
\]
Note that we can in addition assume that the function~$F$ is invariant under the action of $\su^\Fa$ on~$\mathbb B_R\subset \sl^\Ed$. Indeed, one can always replace~$F$ by its average~$\lan F\ran_{\su^{\Fa}}$ over the orbits of this action: the norm of~$F$ in $L^2(\mathbb B_R)$ does not increase under this averaging due to~\eqref{eq:productOfMeasures} and one still has~$\lan F\ran_{\su^{\Fa}}=f\circ\phi$ on~$Z_{\mathrm{unip}}\cap\mathbb B_R$ since the mapping~$\phi$ is invariant under the action of~$\sl^\Fa$.

Thus we can use Lemma~\ref{lem:estimate-of-coeff} to expand~$F$ in~$\mathbb B_r$ as
\[
\textstyle F(\mathbf A) = \sum_{\Gamma,\mathbf m} p_{\Gamma,\mathbf m}F_{\Gamma,\mathbf m}(\mathbf A),\qquad \mathbf A\in\mathbb B_r,
\]
where~$|p_{\Gamma,\mathbf m}|\le r^{-(|\Gamma|+2|\mathbf m|+4|\Ed|)}\cdot \|F\|_{L^2(\mathbb B_R)}$ and the series converges absolutely and uniformly on each smaller poly-ball~$\overline{\mathbb B}_{r'}$, $r'<r$.

Note that each lamination can be represented as a disjoint union of a macroscopic lamination~$\Gamma$, $k_1$ copies of \bl{the loop~$\lloop{\lambda_1}$ surrounding only the puncture $\lambda_1$,} $k_2$ copies of \bl{the loop~$\lloop{\lambda_2}$,} and so on. For shortness, below we use the notation~$\Gamma\sqcup\lloop{\lambda_{\mathbf k}}$ to describe such a lamination, where~$\mathbf k=(k_i)_{i=1}^n\in\mathbb Z_{\ge 0}^n$. By definition,
 \[
 \textstyle F_{\Gamma\sqcup\lloop{\lambda_{\mathbf k}},\mathbf m}(\mathbf A)=F_{\Gamma}(\mathbf A)\cdot\prod_{i=1}^n(F_{\lloop{\lambda_i}}(\mathbf A))^{k_i}\cdot \prod_{e\in\Ed}(\det A_e)^{m_e}.
 \]
Since one has~$\det A_e=1$ and~$F_{\lloop{\lambda_i}}(\mathbf A)=2$ for~$\mathbf A\in Z_{\mathrm{unip}}$, we get the identity
\[
F(\mathbf A)\ =\!\!\sum_{\Gamma~-~\mathrm{macroscopic}} p_{\Gamma}F_{\Gamma}(\mathbf A),\qquad \mathbf A\in Z_{\mathrm{unip}}\cap \mathbb B_r,
\]
where $p_{\Gamma}\ :=\!\!\sum_{\mathbf k\in \mathbb Z_{\ge 0}^n,\mathbf m\in \mathbf{Z}_{\ge 0}^\Ed}p_{\Gamma\sqcup\lloop{\lambda_{\mathbf k}},\mathbf m}\cdot 2^{|\mathbf k|}$. As~$\mathbf n(\Gamma\sqcup\lloop{\lambda_{\mathbf k}})\ge \mathbf n(\Gamma)+2|\mathbf k|$, it is easy to see that
\[
\begin{split}
\textstyle |p_{\Gamma}|\ &\le\ \sum_{\mathbf k\in \mathbb Z_{\ge 0}^n,\mathbf m\in \mathbf{Z}_{\ge 0}^\Ed} 2^{|\mathbf k|}r^{-(|\Gamma\sqcup\lloop{\lambda_{\mathbf k}}|+2|\mathbf m|+4|\Ed|)}\cdot \|F\|_{L^2(\mathbb B_R)}\\
& \le\ (1\!-\!2r^{-2})^n(1\!-\!r^{-2})^{|\Ed|}\cdot r^{-(|\Gamma|+4|\Ed|)}\|F\|_{L^2(\mathbb B_R)}.
\end{split}
\]
This gives the desired exponential upper bound~\eqref{eq:PG-bound} for~$|p_\Gamma|$.

We now move on to the uniqueness of expansion~\eqref{eq:pFg-conv}. Assume that two sequences of coefficients~$p_\Gamma,\tilde{p}_\Gamma$ satisfy the upper bound~$|p_\Gamma|,|\tilde{p}_\Gamma|=O(r^{-|\Gamma|})$ as~$|\Gamma|\to\infty$ and that the corresponding series~\eqref{eq:pFg-conv} coincide on~$\mathbb D_r$. Denote
\begin{equation}
\label{eq:xF-def}
F(\mathbf A)\ :=\!\!\sum_{\Gamma~-~\mathrm{macroscopic}} (\tilde{p}_{\Gamma}-p_\Gamma)F_{\Gamma}(\mathbf A)\qquad\text{for}\ \ \mathbf A\in \mathbb B_r
\end{equation}
Recall that this series converges absolutely and uniformly on compact subsets of~$\mathbb B_r$ due to Lemma~\ref{lem:Taylor-in-F}(i). Therefore,~$F$ is a holomorphic function on~$\mathbb B_r$ which vanishes along~$Z_{\mathrm{unip}}$. Due to Proposition~\ref{explicitNullstellensatz} one can find bounded holomorphic functions~$(H_e)_{e\in \Ed}$ and~$(H_{\lambda_i})_{1\le i\le n}$, defined on~$\mathbb B_{r'}$, $r'<r$, such that
\[
 F(\mathbf A)\ =\ \sum_{e\in\Ed}H_e(\mathbf A)(F_{\varnothing,\mathbf n^e}(\mathbf A)-1)+\sum_{i=1}^nH_{\lambda_i}(\mathbf A)(F_{\lloop{\lambda_i}}(\mathbf A)-2),\qquad \mathbf A\in \mathbb B_{r'}.
\]
Moreover, since all the functions~$F$, $F_{\varnothing,\mathbf n^e}$ and~$F_{\lloop{\lambda_i}}$ are~$\su^\Fa$-invariant, one can assume that all the functions~$H_\alpha$, $\alpha\in\Ed\cup\{\lambda_1,\dots,\lambda_n\}$ are invariant as well by averaging each of them over orbits of the action of the compact group~$\su^\Fa$ on~$\mathbf A\in\mathbb B_r$.

We proceed with expanding these functions in~$F_{\Gamma,\mathbf m}$ as provided by Lemma~\ref{lem:estimate-of-coeff}. Namely, for each~$\varrho<\frac{1}{5}\eta_0r'$ and $\alpha\in\Ed\cup\{\lambda_1,\dots,\lambda_n\}$ one has
\[
H_\alpha(\mathbf A)=\sum_{\Gamma,\mathbf m}p_{\Gamma,\mathbf m}^{(\alpha)}F_{\Gamma,\mathbf m}^{\phantom{()}}(\mathbf A) \quad \text{for}\ \ \mathbf A\in \mathbb B_\varrho,\qquad \text{where}\quad |p_{\Gamma,\mathbf m}^{(\alpha)}|=O(\varrho^{-(|\Gamma|+2|\mathbf m|+4|\Ed|)}).
\]
We arrive at the following expansion on~$\mathbb B_{\varrho}$:
\[
\begin{split}
F(\mathbf A) &= \sum_{e\in \Ed}\sum_{\Gamma, \mathbf m}p_{\Gamma, \mathbf m}^{(e)}F_{\Gamma, \mathbf m}^{\phantom{()}}(\mathbf A)(F_{\varnothing, \mathbf n^e}(\mathbf A) - 1) + \sum_{i = 1}^n\sum\limits_{\Gamma, \mathbf m}p_{\Gamma, \mathbf m}^{(\lambda_i)}F_{\Gamma, \mathbf m}^{\phantom{()}}(\mathbf A)(F_{\lloop{\lambda_i}}(\mathbf A) - 2).
\end{split}
\]
As~$F_{\Gamma, \mathbf m}(\mathbf A)F_{\varnothing, \mathbf n^e}(\mathbf A)=F_{\Gamma, \mathbf m+\mathbf n^e}(\mathbf A)$ and~$F_{\Gamma, \mathbf m}(\mathbf A)F_{\lloop{\lambda_i}}(\mathbf A)=F_{\Gamma\sqcup\lloop{\lambda_i}, \mathbf m}(\mathbf A)$, the last expansion
must coincide with~\eqref{eq:xF-def} due to the uniqueness part of Lemma~\ref{lem:Taylor-in-F}. In particular, for each macroscopic lamination~$\Gamma_0$ we get
\[
\begin{split}
(\tilde{p}_{\Gamma_0}\!-\!p_{\Gamma_0})F_{\Gamma_0}(\mathbf A)\ =&\ \sum_{e\in \Ed}\sum_{\mathbf k,\mathbf m} p_{\Gamma_0\sqcup\lloop{\lambda_{\mathbf k}}, \mathbf m}^{(e)}F_{\Gamma_0\sqcup\lloop{\lambda_{\mathbf k}}, \mathbf m}^{\phantom{()}}(\mathbf A)(F_{\varnothing, \mathbf n^e}(\mathbf A) - 1)\\
+&\ \sum_{i=1}^n \sum_{\mathbf k,\mathbf m} p_{\Gamma_0\sqcup\lloop{\lambda_{\mathbf k}}, \mathbf m}^{(\lambda_i)}F_{\Gamma_0\sqcup\lloop{\lambda_{\mathbf k}}, \mathbf m}^{\phantom{()}}(\mathbf A)(F_{\lloop{\lambda_i}}(\mathbf A) - 2).
\end{split}
\]
Let~$\varrho$ be chosen so that $\varrho>\sqrt{2}$. Then one can substitute~$\mathbf A=\bm{\Id}=(\Id)_{e\in \Ed}$ into the last equality. The right-hand side vanishes and hence~$\tilde{p}_{\Gamma_0}=p_{\Gamma_0}$.
\end{proof}

\begin{proof}[Proof of Theorem~\ref{thm:main-thm}] Due to the existence part of Theorem~\ref{main_theorem_actually} an entire function \mbox{$f\in\Fhol(X_{\mathrm{unip}})^{\sl}$} admits an expansion~\eqref{eq:expansion-of-f} on each bounded subset of~$X_{\mathrm{unip}}$. It  follows from the uniqueness part of Theorem~\ref{main_theorem_actually} that the coefficients of all these expansions coincide provided that the corresponding subsets of~$X_{\mathrm{unip}}$ are big enough. Finally, the estimate~\eqref{eq:PG-bound} implies~\eqref{eq:estimate-of-pG}.
\end{proof}

\subsection*{Acknowledgements} The authors are grateful to Vladimir Fock for several discussions of the material presented in Section~\ref{skein_algebra_section}. We would like to thank Julien Dub\'edat for a feedback on~\cite{Dubedat}, Adrien Kassel and Richard Kenyon for a feedback on~\cite{Kassel-Kenyon,Kenyon}, \bl{and the anonymous referee for useful comments on the first version of this paper.} M.B.~is also grateful to Jason Starr and David E.~Speyer for a MathOverflow discussion of a purely algebro-geometric approach to the material of Section~\ref{sub:Nullstellensatz}. This project was started in 2015/16 during the SwissMAP master class in planar statistical physics, the authors are grateful to the University of Geneva for the hospitality during the program.


\end{document}